\documentclass[12pt]{article}

\AtBeginDocument{%
   \setlength\abovedisplayskip{0.05in}
   \setlength\belowdisplayskip{0.08in}}

\usepackage[onehalfspacing]{setspace}

\makeatletter
\renewcommand\@footnotetext[1]{%
  \insert\footins{%
    \reset@font\footnotesize
    \interlinepenalty\interfootnotelinepenalty
    \splittopskip\footnotesep
    \splitmaxdepth \dp\strutbox \floatingpenalty \@MM
    \hsize\columnwidth \@parboxrestore
    {\setstretch{1.0}\protect\@makefntext{%
      \rule{\z@}{\footnotesep}\ignorespaces#1}}}}
\makeatother

\usepackage{geometry}
\usepackage{enumitem}
\usepackage{lineno}
\usepackage{booktabs}
\usepackage{amssymb}
\usepackage{amsmath}
\usepackage{amsthm}
\usepackage{caption}
\usepackage{epsfig}
\usepackage{graphicx}
\usepackage{graphics}
\usepackage{float}
\usepackage{subfigure}
\usepackage{multirow}
\usepackage{xcolor}
\usepackage{lineno}
\usepackage{fullpage}
\usepackage[normalem]{ulem} 
\usepackage{makeidx}
\usepackage{xspace}
\usepackage{wrapfig}
\usepackage{lscape}
\usepackage{mathrsfs}
\usepackage{adjustbox}
\usepackage{multirow}
\usepackage{varwidth}
\usepackage{rotating}
\usepackage{lscape}
\usepackage[x11names,dvipsnames]{xcolor}
\usepackage{tikz}
\usetikzlibrary{automata, positioning}
\usepackage[bottom]{footmisc}
\usepackage[colorlinks,citecolor=darkblue,urlcolor=darkblue,linkcolor=darkblue]{hyperref} 
\graphicspath{{C:/Users/flori/Desktop/folders/minimumWagePaper/paper/}}
\makeindex
\newtheorem{theorem}{Theorem}
\newtheorem{assumption}{Assumption}
\newtheorem{corollary}{Corollary}

\newtheorem{lemma}{Lemma}

\newtheorem{proposition}{Proposition}
\newtheorem{definition}{Definition}

\newtheorem{claim}{Claim}
\newtheorem{obs}{Observation}

\definecolor{purple}{rgb}{0.6, 0.4, 0.8}
\definecolor{darkred}{rgb}{1, 0.1, 0.3}
\definecolor{darkblue}{rgb}{0.0, 0.0, 0.55}
\definecolor{darkgreen}{rgb}{0,0.6,0.5}
\definecolor{forestgreen}{rgb}{0.0, 0.46, 0.37}
\definecolor{bittersweet}{rgb}{1.0, 0.44, 0.37}
\definecolor{navy}{rgb}{0.0, 0.0, 0.55}
\definecolor{brown}{rgb}{0.53, 0.18, 0.09}
\definecolor{Green}{rgb}{0.0, 0.47, 0.44}

\newcommand {\mm}[1] {\ifmmode{#1}\else{\mbox{$#1$}}\fi}

\newcommand\E{\mathbb{E}}
\newcommand\Proba{\mathbb{P}}
\newcommand\R{\mathbb{R}}

\usepackage{dsfont}

\newcommand{\1}{\mathbbm{1}}
\newcommand{\DKL}{D_{\mathrm{KL}}}

\usepackage{sectsty}
\sectionfont{\centering\Large}
\subsectionfont{\large}
\usepackage{fancyhdr}
\usepackage{natbib}

\usepackage{pst-node, pst-plot, pst-poly}
\usepackage{rotating}

\usepackage{tikz}
\usepackage{pstricks}
\usepackage{tikz-cd}
\usepackage{siunitx}
\usepackage{pgfplots}
\usepgfplotslibrary{fillbetween}
\usepackage{algorithm}
\usepackage{algorithmic}
\usepackage{bbm}
\usepackage{csquotes}

\usepackage{scalerel}
\usetikzlibrary{shapes,arrows,trees,calc}
\pgfplotsset{compat=1.12}

\usetikzlibrary{trees}
\allowdisplaybreaks 
\begin{document}

\title{ \vspace{-3.8em}
\mbox{\small\bf\MakeUppercase{Motivating Innovation with a Misspecified Roadmap}}\footnote{I am very thankful to Drew Fudenberg for his support and mentorship. I am also grateful to Ian Ball, Glenn Ellison, Xavier Gabaix, Ani Ghosh, Robert Gibbons, Itzhak Gilboa, Gustavo Manso, Alfonso Maselli, Stephen Morris, Parag Pathak, Alejandro Rivera, Alexander Wolitzky, and seminar participants at the MIT Theory Lunch for insightful discussions.
}
\vspace{-0.5em}
}
\author{\small\MakeUppercase{Florian Mudekereza}\footnote{Department of Economics, MIT, \href{mailto:florianm@mit.edu}{\texttt{\footnotesize florianm@mit.edu}}.}}
\date{}

\maketitle
\thispagestyle{empty}
\setcounter{page}{0}
\vspace{-0.35in}
\begin{abstract}
We analyze a principal-agent relationship where a principal communicates a \textit{roadmap} to guide an agent who is learning the value of innovation. However, the agent is concerned that the roadmap is \textit{misspecified}. We find that the agent can fall into a \textit{breakthrough trap}, where early unexplained success triggers a loss of trust in the roadmap, such that \textit{no} contract can motivate him to continue innovating. We also obtain an upper bound on the frequency of innovative activity that tightens as the degree of misspecification increases, which can cause ``exploration-exploitation'' \textit{cycles} to emerge endogenously over time. 
\end{abstract}

\vspace{0.04in}

\begin{singlespace}
{\noindent\small\textit{Discovery commences with the awareness of anomaly, i.e., with the recognition that nature has somehow violated the paradigm-induced expectations that govern normal science.}
\par \hfill-- \citet[][pp. 52--53]{kuhn70}}
\end{singlespace}
\vspace{-0.04in}

\section{Introduction}

A central theme in the management of innovation is that organizations must commit to a ``business model'' in order to coordinate effort, while simultaneously recognizing that the model is an imperfect description of an environment that is not yet well understood.
\citet[][Chapter 8.2]{mgt11} develop this idea using the notion of a ``roadmap'' for innovative projects conducted under high uncertainty: a roadmap is intended to discipline employees' decisions and communication even when it is incomplete, so it should be accompanied with a ``robust mindset'' that acknowledges potential errors or \textit{misspecification}.\footnote{\citet[][pp. 172--173]{mgt11} note: ``A map, even a bad one, can help the organization to cope with the
threatening and uncontrollable nature of the unknown. It helps to make
team members action-oriented and discover other sources of information that may be more appropriate than the bad map. The role of a map in the
ability to cope with unk unks, which we call `robust mind-set' [...]''} 
We formalize this mindset in a principal-agent framework under moral hazard and find that it fundamentally changes the economics of innovation. Specifically, unlike standard economic theory, which predicts that an early breakthrough lowers the agency cost of future innovation \citep[e.g.,][]{manso2011}, our framework predicts the opposite effect.  
\par To fix ideas, consider a materials-science lab where a Chief Scientific Officer (the principal) hires a Research Scientist (the agent) to discover a new compound. The principal has access to the lab's historical data, protocols, and simulation tools that jointly produce a disciplined set of predictive scenarios for experimental outcomes, which we refer to as ``structured knowledge'': a set of plausible outcome distributions $Q$ and a prior belief $\mu$ that assigns likelihoods to these distributions. Since monitoring the scientist's actions may not be feasible, the principal offers guidance by communicating her structured knowledge $(\mu,Q)$ as a \textit{roadmap}. The scientist, however, is a statistically sophisticated expert; while he treats $\mu$ as a credible summary statistic of the lab's track record, he is concerned that $Q$ is \textit{misspecified}. This may reflect the sentiment of scientists in large labs who believe protocols are well-calibrated, but still recognize that models are approximations of real experiments and thus may fail to identify new compounds. The central question then becomes how incentive provision interacts with \textit{learning} when the evidence generated by the project endogenously affects the scientist's trust in models.

Our baseline framework embeds a robust-control mindset into a two-period, two-action agency environment that isolates the tension between incentive schemes that motivate ``routine'' and ``innovative'' activities. Our main result, Theorem \ref{thm:dynamic:trap-scaleup}, is an impossibility result that arises when the agent's concern for misspecification evolves endogenously. We identify a ``Good News is Bad News'' phenomenon, where early success in a project can paradoxically reduce, rather than increase, the feasibility of sustaining innovation. Specifically, (i) \textit{Good News}: standard Bayesian theory implies that early success relaxes incentive constraints by signaling higher project quality; (ii) \textit{Bad News}: early unexplained success triggers a loss of trust in models and hence compresses the set of continuation payoffs. From a graphical perspective, Bayesian updating raises the slope of the incentive curve, whereas the misspecification effect bends this curve downward by imposing an upper bound on the principal's \textit{incentive capacity}. When this misspecification effect dominates the Bayesian-updating effect, Theorem \ref{thm:dynamic:trap-scaleup} shows that no contract can motivate the agent to continue innovating after achieving a breakthrough. We refer to this implementability failure as a \textit{breakthrough trap}. Once in this trap, increasing the post-breakthrough value of innovation no longer raises equilibrium continuation effort and ceases to affect the principal's expected payoff. This negative finding suggests that in innovation projects where early breakthroughs are poorly explained by models, high-powered incentives may exhibit a weak association with future innovation effort.
\par To illustrate the breakthrough trap, consider a lab in 2023 whose principal designs a roadmap for a project on room-temperature superconductivity  based on an infamous material called LK-99 \citep[][]{lk23}. Suppose a scientist follows the roadmap and obtains an early successful experimental outcome.  Interpreting this success as a positive signal of the project's quality, the principal encourages the scientist to scale up the experiments. However, after auditing the mechanism behind the early success, the scientist realizes that it was an \textit{anomaly} driven by an incidental byproduct of the process---an impurity that can reproduce the same behavior---rather than by the mechanism the roadmap relies on \citep[][]{naturelk23,nature23}. In other words, the scientist realizes that the roadmap is misspecified---it did not produce the success for reasons it claimed. Since success can occur through channels that the roadmap does not discipline, the scientist concludes that exerting effort to scale up the LK-99 project is no longer justifiable; as \citet[][p. 18]{naturef23} notes: ``The wave of excitement caused by LK-99---the purple crystal that was going to change the world---has now died down after studies showed it wasn’t a superconductor.'' This is the sort of agency dynamics that our framework aims to capture:  sophisticated statisticians routinely reassess their models in light of new evidence, which can lead them to abandon a project if they realize that their models are very misspecified.
\par The primary economic tension in our framework is not that innovation is hard to initiate, but rather that it can become harder to sustain innovation effort precisely after the most encouraging evidence is realized. Theorem \ref{thm:dynamic:trap-scaleup} reveals that motivating continued innovation depends critically on how informative early evidence is about model misspecification. Thus, what is more relevant for subsequent innovative activity is not whether early evidence is good or bad per se, but whether it is \textit{explained} or \textit{unexplained} by the roadmap: unexplained breakthroughs trigger a loss of trust in the roadmap and hence reduce incentive leverage, whereas explained failures need not. We build on this insight to propose three managerial tools that can mitigate breakthrough traps. (1) ``Tolerating early failure'' by reporting \textit{coarse} performance evaluations can preserve incentives by attenuating the extent to which early success constitutes a misspecification ``shock'' (Corollary \ref{prop:dynamic:feedback-optimal}). (2) A principal can \textit{screen} agents at hiring using a menu of path-dependent contracts to select agents who are most resilient to unexplained evidence (Corollary \ref{thm:dynamic:screening}). (3) When neither feedback design nor screening is feasible, replacing the agent after an early success can restore incentives, assuming a new agent would observe the past success without experiencing the incumbent's misspecification shock (Corollary \ref{prop:dynamic:turnover}). Notice that this turnover policy is the opposite logic of \citeauthor{manso2011}'s (\citeyear{manso2011}) termination policy prescribing to replace the agent after early  failure.  
Our framework therefore provides a new explanation for the so-called ``paradoxes of success,'' where employees (e.g., founders, CEOs) tend to get replaced after achieving major milestones \citep[][]{founder03,founder18}. The key is that replacement here is not due to changes in skill or technology, but instead it is due to the epistemic cost of discovery---the act of achieving a breakthrough is itself a diagnostic shock that can degrade the agent's trust in models.
\par Our second main result, Theorem \ref{thm:cycles:frontier}, is the long-run accumulation of Theorem \ref{thm:dynamic:trap-scaleup}’s local mechanism, which establishes an  upper bound on innovation frequency. On an infinite-horizon path, Theorem \ref{thm:cycles:frontier} shows that equilibrium innovation intensity is subject to a \emph{speed limit}: the long-run frequency of innovation is bounded above by a quantity that is inversely proportional to the degree of misspecification. This happens because repeated innovation generates model-diagnostic evidence that raises the agent's misspecification concerns and therefore tightens incentive capacity. As a result, equilibrium behavior can exhibit endogenous ``exploration-exploitation'' \textit{cycles} in which periods of active innovation are followed by phases where the organization prioritizes routine tasks to prevent the agent's misspecification concerns from overshooting. Such cycles have been documented in practice but are attributed to several factors, including technological change \citep[e.g.,][]{innov90}. However, in our framework, innovation cycles arise for a novel reason: they act as a mechanism for building \textit{trust}; organizations must manage both their employees' efforts and their trust in the roadmap. Our framework therefore predicts that innovation intensity is upper bounded and exhibits endogenous cycles in high-uncertainty domains (e.g., materials science), where diagnostic evidence indicates poor roadmap fit. This  suggests that organizations may speed up innovation intensity by improving the statistical performance of their roadmaps, e.g.,  tuning model parameters.

\par We also show that breakthrough traps persist even in the long run when a sophisticated agent is repeatedly taking actions based on a misspecified roadmap. A natural way to shut down the mechanism that generates these traps is to have the agent behave as a na{\"i}ve, unsophisticated statistical learner. A leading example is an agent who is too ``lenient'' in his assessment of the roadmap's performance relative to realized data in the sense that he attributes too much unexplained evidence to sampling variability. When this is the case, the agent's misspecification concerns will vanish in the long run despite the presence of misspecification, and as a result, our principal-agent framework converges to \citeauthor{manso2011}'s (\citeyear{manso2011}) expected-utility environment, where breakthrough traps do \textit{not} arise.  

\par We now describe our framework's primitives. The agent is forward-looking and has misspecification concerns. To formalize the separation between his trust in the prior $\mu$ and his concern about misspecification of models in $Q$, our baseline framework follows the robust-control literature by adopting \citeauthor{hansenmiss25}'s (\citeyear{hansenmiss25}) \textit{average robust control} (ARC) criterion, which was extended to dynamic settings by \citet{lanzani2025}. Under ARC, the agent first applies \citeauthor{hansen01}'s (\citeyear{hansen01}) multiplier-robustness adjustments \textit{within} each model in $Q$, reflecting his fear of misspecification. He then evaluates contracts by averaging \textit{across} models using $\mu$, reflecting his trust in the prior.  The principal is risk neutral, forward-looking, and assuming she has access to aggregate information across multiple projects, she is neutral to both model ambiguity and misspecification.  Allowing the principal to be averse to either ambiguity or misspecification would only reinforce the breakthrough-trap mechanism (Section \ref{sec:lit}); a similar reinforcement happens when the agent is allowed to distrust $\mu$ (Appendix \ref{sec:distrust}).  \citeauthor{manso2011}'s (\citeyear{manso2011}) seminal framework is the special case where the agent is also neutral to both ambiguity and misspecification. 

\par\noindent--- \textit{Related Work}: Our main contribution is to the theoretical literature on incentive provision for innovation, pioneered by \citet{manso2011}, who argues that contracts that tolerate failure in the short run and reward success in the long run are best suited for motivating innovation \citep[e.g.,][]{manso13,innov14}. In contrast, we identify new dynamic constraints when a sophisticated misspecification-averse agent can learn: observing unexplained successful outcomes tightens incentive capacity by endogenously increasing the agent's misspecification concerns, generating both local implementability failures and long-run upper bounds on innovation frequency. Since these results are driven by the endogenous evolution of the agent's misspecification concerns, they cannot be generated by classical frameworks such as expected utility or maxmin \citep{gilboa89} because these frameworks do not capture concerns for model misspecification. In fact, these two frameworks constitute extreme long-run limits of our framework, where the agent is, respectively, either too lenient or too demanding in his evaluation of the organization's roadmap in light of observed data (Observation \ref{prop:dynamic:infty:limit-pref}). 
\par We also contribute to the literature on contracts under ambiguity by introducing concerns for model misspecification on the \textit{agent side} of moral-hazard environments. \citet{amb11} are the first to introduce ambiguity only on the agent side by assuming the agent has an incomplete preference \`a la \citet{bew02}. In contrast, we emphasize a \textit{separation} between the agent's attitudes toward the organization's prior belief and models,  
which is not possible in popular  frameworks such as maxmin preferences \citep{gilboa89}, Choquet preferences \citep{schmeidler89}, Bewley preferences \citep{bew02}, or smooth ambiguity preferences \citep{smooth05}. This separation is key  because it allows our agent to learn about model likelihood and model misspecification separately. \citeauthor{hansenmiss25}'s (\citeyear{hansenmiss25}) ARC preference therefore provides a tractable and economically meaningful  way of representing \citeauthor{mgt11}'s (\citeyear{mgt11}) ``robust mindset'' in high-uncertainty contracting settings. Unlike \citet{lanzani2025} who assumes a myopic agent in the dynamic extension of ARC, our agent is forward-looking, so we build on \citeauthor{mmr06_JET}'s (\citeyear{mmr06_JET}) insights to obtain a \textit{recursive} representation of ARC, which is dynamically consistent (Appendix \ref{subsec:dynamic:mmr}). We should note, however, that the ARC functional form is not essential; as we show in Appendices \ref{sec:smooth}--\ref{sec:distrust},  our  qualitative results continue to hold under \citeauthor{hansenmiss25}'s (\citeyear{hansenmiss25}) general classes of misspecification-robust decision criteria, which allow more complex attitudes toward misspecification. 
Section \ref{sec:lit} discusses our connection to the broader moral-hazard literature. 

\par\noindent--- \textit{Outline}: The remainder of the paper proceeds as follows. Section \ref{sec:dynamicmodel} introduces our framework. Section \ref{sec:dynamic} analyzes breakthrough traps and Section \ref{sec:avoidtrap} examines how to avoid them. Section \ref{sec:cycles} considers infinite-horizon settings. Section \ref{sec:discussion} discusses some extensions.

\section{Framework}\label{sec:dynamicmodel}

We start with a single-agent problem that isolates how an agent learns from early innovative activity. We then use this decision problem as the basis for the agency problem.

\subsection{Single-agent benchmark}\label{subsec:dynamic:single-agent}

There are two periods $t\in\{1,2\}$ and two outcomes $Y=\{0,1\}$, where $y=1$ is a ``success'' and $y=0$ is a ``failure.'' The agent chooses action $a_t\in A=\{1,2\}$ in each period $t$. Action $1$ is a \textit{routine} (or familiar) action and action $2$ is an \textit{innovative} (or experimental) action. Their costs are normalized as $c^1=0$ and $c^2=k>0$ so that innovation is more costly.

\medskip
\noindent\textbf{Routine action.}
Action $1$ induces one ``prior'' or ``structured'' model $Q^{1}=\{q^1\}\subset\Delta(Y)$ with $q^1(1)=:p\in(0,1)$. It is routine in the sense that it will not affect the agent's learning process, so it is the analogue of the ``conventional work method'' in \citet{manso2011}.

\medskip
\noindent\textbf{Innovative action.}
Action $2$ induces two structured models $Q^{2}=\{q^2_L,q^2_H\}\subset\Delta(Y)$ with
$q^2_L(1)=:\theta_L$, $q^2_H(1)=:\theta_H$, and $0<\theta_L<p<\theta_H<1$. Model $q^2_L$ represents a pessimistic assessment of innovation and model $q^2_H$ represents an optimistic assessment. The agent begins with a prior belief $\mu_1\in\Delta(Q^{2})$ with $\mu_1(q^2_H)=m\in(0,1)$. Thus, the innovative action is the action whose value the agent seeks to learn, so it is the analogue of the ``new work method'' in \citet{manso2011}. The multiplicity of models in $Q^{2}$ aims to capture \citeauthor{mgt11}'s (\citeyear{mgt11}) idea that innovation projects are highly uncertain. This formulation will allow us to separate two kinds of learning: learning about which structured model is more likely and learning about how well these models themselves explain the evidence.

\medskip
\noindent\textbf{Expected-utility benchmark.}
First consider an expected-utility (EU) agent who trusts the structured models. For any continuation payoff vector $x:Y\to\R$, the period-$t$ value of the routine action is $\E_{q^1}[x(y)]$, while the period-$t$ value of the innovative action is
$\sum_{q\in Q^{2}}\mu_t(q)\E_q[x(y)],$ where $\mu_t$ is a Bayesian posterior defined below.
Thus, the difference
$\sum_{q\in Q^{2}}\mu_t(q)\E_q[x(y)]-\E_{q^1}[x(y)]$ is the EU payoff wedge between innovation and routine activity.
This comparison is the main object in \citeauthor{manso2011}'s  (\citeyear{manso2011}) standard bandit problem. 

\medskip
\noindent\textbf{Histories.} Histories are $h_1=\varnothing$ and $h_2=y_1\in Y$.

\medskip
\noindent\textbf{Bayesian updating.}
If the routine action is chosen $a_1=1$, set $\mu_2(\cdot|y_1)=\mu_1$ for every $y_1\in Y$. If the innovative action is chosen $a_1=2$ and $y_1$ is observed, Bayes rule yields:
\begin{align}
\mu_2(q^2_H|y_1=1)\propto m\hspace{0.02in}\theta_H,\qquad
\mu_2(q^2_H|y_1=0)\propto m(1-\theta_H).
\label{eq:dynamic:mu-update}
\end{align}
Under EU, an early success is good news because it shifts posterior weight toward the optimistic model $q^2_H$, and this is the only way evidence affects the period-2 comparison. 

\medskip
\noindent\textbf{Misspecification concern.}
Our key novelty is that we allow the agent to worry that all the structured models used to evaluate actions are \textit{misspecified}. To disentangle his attitude toward models from their prior likelihoods, his one-step evaluation of continuation payoffs is represented by \citeauthor{hansenmiss25}'s (\citeyear{hansenmiss25}) \textit{average robust control} (ARC) criterion, which consists of two steps. First, within each structured model $q$, the agent applies \citeauthor{hansen01}'s (\citeyear{hansen01}) multiplier-robustness adjustment, which uses a CARA transformation of the EU benchmark to reflect his concern that $q$ is misspecified. Second, he averages these model-specific robust valuations using a prior belief $\mu$. Formalizing these two steps, for $q\in\Delta(Y)$, $\lambda>0$, and $x: Y\to\R$, the \textit{multiplier criterion} is  
\begin{align}
g_q(x;\lambda):=\min_{r\in\Delta(Y)}
\Big\{\E_r[x(y)]+\frac{1}{\lambda}\DKL(r\|q)\Big\}
=-\frac{1}{\lambda}\log\Big(\sum_{y\in Y} q(y)e^{-\lambda x(y)}\Big),
\label{eq:dynamic:gq}
\end{align}
where the equality follows from the variational formula of \citet[][Proposition 1.4.2]{dupuis97}, $\DKL(r\|q):=\sum_{y\in Y} r(y)\log\frac{r(y)}{q(y)}\ge 0$ is the Kullback-Leibler (KL) divergence, and the constant $\lambda>0$ captures the agent's distrust of model $q$. Then, 
for a prior $\mu\in\Delta(Q^{2})$, \citet[][eq. (35)]{hansenmiss25} define ARC as the $\mu$-average of  (\ref{eq:dynamic:gq}):
\begin{align}
\mathcal G(x;\mu,\lambda):=\sum_{q\in Q^{2}} \mu(q)\hspace{0.02in} g_q(x;\lambda).
\label{eq:dynamic:G}
\end{align}
Under ARC, the continuation evaluation of innovation is $\mathcal G(x;\mu_t,\lambda_t)$, while the continuation evaluation of the routine action is $g_{q^1}(x;\lambda_t)$. The difference
$\mathcal G(x;\mu_t,\lambda_t)-g_{q^1}(x;\lambda_t)$
is the ARC payoff wedge between innovation and routine activity.
Relative to EU, the new state variable is $\lambda_t$, which measures how much the agent distrusts the structured models used to evaluate innovation in period $t$. The evolution of this state is formalized below.

\medskip
\noindent\textbf{LLR updating.}
The agent updates his misspecification concern $\lambda$ endogenously given past histories using the \textit{average} LLR updating rule introduced in \citet[][eq. (2)]{lanzani2025}.

\begin{assumption}\label{ass:dynamic:llr}
Fix $\gamma>0$. If the routine action is chosen $a_1=1$, set $\lambda_2(y_1)=\lambda_1$ for every $y_1\in Y$. If the innovative action is chosen $a_1=2$ and $y_1\in Y$ is observed,\footnote{Since $Y$ is finite and all structured models have full support, we follow \citeauthor{lanzani2025}'s (\citeyear[][Section 2.2]{lanzani2025}) suggestion by choosing $\Delta(Y)$ as the set of ``alternative unstructured'' models. Thus, more formally, (\ref{eq:dynamic:lambda-update}) is $\lambda_2(y_1)=\text{LLR}^2(y_1;Q^{2})/\gamma$, where $\text{LLR}^2(y;Q^{2})
:=-\log\frac{\text{\normalfont max}_{q\in Q^{2}}q(y)}{\text{\normalfont max}_{p\in\Delta(Y)}p(y)}$; see \eqref{eq:dynamic:infty:LLR2} for the general form.} set
\begin{align}
\lambda_2(y_1)=-\frac{1}{\gamma}\log\underset{q\in Q^{2}}{\text{\normalfont max}}\hspace{0.04in} q(y_1).
\label{eq:dynamic:lambda-update}
\end{align}
\end{assumption}
The constant $\lambda_1\geq0$ denotes the agent's initial misspecification concern at $t=1$, where $\lambda_t\equiv0$ is the limit case corresponding to the EU benchmark.  The constant $\gamma$ captures the extent to which the agent's concern for misspecification is sensitive to unexplained evidence. A success is ``unexplained'' if it has low absolute likelihood even under the most favorable model. A large LLR indicates a poor fit of models $Q^{2}$ to data, which signals a high risk of misspecification. This updating rule aims to formalize the logic behind the LK-99 example: experts routinely reassess the validity or empirical fit of their models in light of new evidence. This LLR rule is standard in statistics for hypothesis testing and model selection under misspecification \citep[e.g.,][]{miss77,miss89}. \citet[][Theorem 1]{lanzani2025} provides normative justifications for the LLR rule by showing that it is a \textit{rationality} benchmark in repeated decision problems. We show in Appendix \ref{subsec:disc:generic-update} that more general fit-based rules can be used in our framework. We focus on LLR mostly because it yields simple closed-form expressions, as shown below.

\begin{lemma}\label{lem:dynamic:lambda-closed}
Under Assumption \ref{ass:dynamic:llr} and $Q^{2}=\{q^2_L,q^2_H\}$,
\begin{align}
\lambda_2(1)=\frac{-\log \theta_H}{\gamma},\quad \text{and}\quad
\lambda_2(0)=\frac{-\log(1-\theta_L)}{\gamma}.
\label{eq:dynamic:lambda-closed}
\end{align}
\end{lemma}
An early success therefore has two effects. First, by Bayes rule \eqref{eq:dynamic:mu-update}, it raises the posterior weight on $q^2_H$. Second, by LLR \eqref{eq:dynamic:lambda-closed}, it also affects the agent's future misspecification concern. Thus,  first-period innovation has an experimentation value under ARC  because it affects the agent's continuation decision through two state variables ($\mu_2(\cdot|\cdot),\lambda_2(\cdot)$).

\subsection{Agency problem}\label{subsec:dynamic:env}

We now introduce the agency problem in which a principal delegates the two-period task above to an agent whose action is not observed. This will allow us to study how the agent's misspecification concern affects both his choices and the set of continuation payoff wedges that a contract can induce between routine and innovation activities.

\medskip
\noindent\textbf{Actions and costs.}
The action and outcome structure is as in Section \ref{subsec:dynamic:single-agent}. Action $1$ is routine, action $2$ is innovative, and their costs are $c^1=0$ and $c^2=k>0$.\footnote{To focus on the tension between innovative and routine actions, we omit the possibility of ``shirking.'' Appendix \ref{app:dynamic:shirking} allows shirking and shows that it only reinforces our main results.}

\medskip
\noindent\textbf{Principal.}
The principal is a risk-neutral subjective EU maximizer whose discount rate for period-2 utility is normalized to 1. We allow her period-2 utility to be scaled by a constant $\xi>1$. Specifically, her period-1 utility is $y_1$ and period-2 utility is $R_2y_2$, where
\begin{align}\label{eq:multiplier}
    R_2 := 1 + (\xi-1)\1\{a_1=2, y_1=1, a_2=2\}.
\end{align}
That is, if the agent innovates early, succeeds, and continues to innovate, the period-2 utility is scaled by $R_2= \xi>1$. This is a reduced-form way to capture \textit{staged investment}, where early success or milestone raises the value of continuation investment in a project, such as venture-capital stage financing of risky projects \citep[e.g.,][]{staged12}.

\medskip
\noindent\textbf{Structured knowledge.}
For the rest of the paper, we interpret structured models as information communicated within an organization. The pair $(\mu_1,Q^{2})$ represents the principal's \textit{structured knowledge} about the innovative action: $Q^{2}$ contains plausible predictive scenarios for the innovation outcome, while $\mu_1$ assigns likelihoods to these scenarios.

\medskip
\noindent\textbf{Communication.}
We interpret $(\mu_1,Q^{2})$ as the principal's \emph{roadmap} for innovation that she communicates to the agent to offer guidance in a project. The roadmap is \textit{exogenous},\footnote{Section \ref{sec:disc} indicates that this exogenous communication approach is standard in the literature. 
Appendix \ref{sec:disc:roadmap-no-free-lunch} shows that the effect of endogenous roadmaps on our qualitative results is limited.} e.g., $(\mu_1,Q^{2})$ is determined ex ante by the organization's protocols or scientific evidence. 

\medskip
\noindent\textbf{Agent.}
The agent's discount rate for period-2 utility is also normalized to 1. He evaluates continuation payoffs with the ARC criterion in \eqref{eq:dynamic:G}. This criterion implies two behavioral assumptions: (i) the agent treats the innovation prior $\mu_1$ as a credible belief, perhaps because he recognizes the organization's track record in classifying new scenarios based on historical data.\footnote{The assumption that the agent trusts the principal's belief captures the idea of ``shared beliefs'' in the organizational-culture literature \citep[e.g.,][]{van2005,van10}. Appendix \ref{sec:distrust} shows that our main qualitative results continue to hold even when the agent is allowed to distrust the principal's belief.} (ii) He faces, however, a daunting reality: his specific task is novel; he is therefore concerned that all structured models are misspecified representations of the true distribution induced by innovation.  In Appendix \ref{subsec:dynamic:mmr}, we follow \citet{mmr06_JET} to derive a recursive representation of ARC \eqref{eq:dynamic:G} that is dynamically consistent (Proposition \ref{prop:dynamic:mmr-recursive}).\footnote{Focusing on dynamic consistency allows us to show that our impossibility result (Theorem \ref{thm:dynamic:trap-scaleup}) is not driven by the agent's behavioral biases. If dynamic consistency is dropped, then as \citet[][Section 5.1, p. 1351]{lanzani2025} remarks, it may also be technically challenging to analyze ``an
uncommitted, forward-looking, and sophisticated agent playing an intra-personal equilibrium
with their future selves.''} We should note that our baseline analysis focuses on ARC because it is tractable, but it is not essential for our qualitative results (Appendices \ref{sec:smooth}--\ref{sec:distrust}).

\medskip
\noindent\textbf{Contracts.}
A dynamic contract is a pair $(x_1,x_2)$, where $x_1:Y\to\R$ and $x_2:Y\times Y\to\R$ denote utility payments in periods 1 and 2, respectively. Let $u:\R\to\R$ be strictly increasing and onto, with inverse $v:=u^{-1}$. Then, the monetary wage that delivers utility $x$ is $v(x)$ and the principal incurs cost $v(x)$, so we work directly in utility space with $x$.

\medskip
\noindent\textbf{Incentive capacity.}
For any $(\mu,\lambda)$ and $x$, define the \textit{incentive capacity} as
\begin{align}
C(\mu,\lambda):=\sup_{x\in\R^Y}{M}(\mu,\lambda;x),
\label{eq:dynamic:Delta-C}
\end{align}
where ${M}(\mu,\lambda;x):=\mathcal G(x;\mu,\lambda)-g_{q^1}(x;\lambda)$ is the incentive gap.  $C(\mu,\lambda)$ is the maximum continuation wedge between innovation and routine utility that the principal can create given $(\mu,\lambda)$. This object will play a central role because given $(\mu_2(\cdot|y_1=1),\lambda_2(1))$, the innovative action $a_2=2$ is implementable after a success if and only if
$C(\mu_2(\cdot|y_1=1),\lambda_2(1))> k$ or with equality when the supremum in \eqref{eq:dynamic:Delta-C} is attained (Proposition \ref{prop:dynamic:terminal-impl}).

\medskip
\noindent\textbf{Dynamic game.}
Hereafter, ``equilibrium'' refers to a subgame-perfect Nash equilibrium (SPNE) of the induced two-period extensive-form game, where at each history $h_t$, the agent optimizes given continuation payoffs and state $(\mu_t(\cdot|h_t),\lambda_t(h_t))$, and the principal chooses a contract ex ante while anticipating this behavior. To keep the main analysis focused, the formal dynamic implementation analysis is relegated to Appendices \ref{subsec:dynamic:policies}--\ref{subsec:dynamic:optimal}.

\section{Breakthrough Trap}\label{sec:dynamic}

The key economic question of this paper is whether an early breakthrough makes continued innovation easier or harder to implement. The robust-control literature predicts that higher-powered incentives induce more pessimistic belief distortions. Our contribution is to show that when these distortions evolve endogenously with learning in agency problems, they can completely eliminate implementability of continued innovation.

\subsection{Impossibility result}\label{subsec:dynamic:trap}

A \textit{breakthrough trap} is said to occur when the post-success incentive capacity falls below the agent's cost of innovation. This concept is defined more formally below.  
\begin{definition}\normalfont
A \textit{breakthrough trap} occurs when $C(\mu_2(\cdot|y_1=1),\lambda_2(1))<k$. 
\end{definition}
We are now in position to present this paper's main impossibility result.

\begin{theorem}[Breakthrough trap]\label{thm:dynamic:trap-scaleup}
Fix $(p,\theta_H,m,k,\gamma,\xi)$ with $0<p<\theta_H<1$, $k>0$, and define $\lambda^*:=\frac{1}{k}\log \frac{1-p}{1-\theta_H}.$
If $\lambda_2(1)>\lambda^*$, then there exists $\bar\theta_L\in(0,p)$ such that for all
$\theta_L\in(0,\bar\theta_L)$, the post-success incentive capacity satisfies
$C(\mu_2(\cdot|y_1=1),\lambda_2(1))<k.$
Consequently, under any dynamic contract
and any induced subgame-perfect equilibrium:
\begin{enumerate}
    \item[(1)] After $(a_1=2,y_1=1)$, the continuation action $a_2 = 2$ is never implemented.
    \item[(2)] The principal's expected profit is independent of the scale-up multiplier $\xi$.
\end{enumerate}
\end{theorem}

Theorem \ref{thm:dynamic:trap-scaleup} highlights a fundamental tension between Bayesian learning and endogenous misspecification concerns. In standard exploration-exploitation frameworks such as \citet{manso2011}, a breakthrough is good news in the following sense: it shifts the posterior $\mu_2$ toward the model $q^2_H$, which increases the expected returns of continued innovation. In our framework, however, a breakthrough also reveals whether the models are flawed, so it is bad news when it triggers an endogenous shock in the agent's post-success misspecification concern $\lambda_2(1)=-\frac{1}{\gamma}\log\theta_H$. When $\lambda_2(1) > \lambda^*$, 
the misspecification shock dominates the improvement in the posterior belief from Bayesian updating. Then, the agent's shadow cost of robustness---the risk premium required to bear the uncertainty of innovation---rises faster than the project's expected surplus. As a result, the incentive capacity $C(\mu_2,\lambda_2)$---the maximum transferable utility the principal can promise per unit of risk---falls below the cost of innovation $k$. In this case, Theorem \ref{thm:dynamic:trap-scaleup}.(1) shows that it will become impossible to motivate the agent to innovate in period 2. This result reveals that a success---the very signal that makes a project more valuable---paradoxically renders innovation uncontractible in the future when it is poorly explained by models. \par The threshold $\lambda^*=\frac{1}{k}\log \frac{1-p}{1-\theta_H}$ is intuitive: it captures the project's signal-to-cost ratio or statistical discriminability per unit of innovation cost. The condition $\lambda_2(1)>\lambda^*$ in Theorem \ref{thm:dynamic:trap-scaleup} may be rewritten as $k\log \frac{1}{\theta_H}>\gamma \log \frac{1-p}{1-\theta_H}$, which shows that it can hold in two key cases: i) when $\theta_H$ is small enough such that early breakthrough is very surprising; ii) when $\gamma$ is small enough such that the agent is very sensitive to misspecification. 
\par  Theorem \ref{thm:dynamic:trap-scaleup}.(2) is the main economic consequence of breakthrough traps. It shows that in such a trap, the principal loses all incentive leverage over post-breakthrough innovation. Recall that the scale-up multiplier $\xi$ in \eqref{eq:multiplier} raises the principal's payoff only on histories in which the agent innovates early, succeeds, and then continues innovating. Notice, however, that Theorem \ref{thm:dynamic:trap-scaleup}.(1) rules out this continuation action after success. Thus, even high-powered post-success incentives can fail to sustain innovation precisely at the stage where the project becomes more valuable to an organization. In other words, unexplained success severs the usual link between project value and continuation effort. 
\begin{corollary}\label{cor:dynamic:success-failure-reversal}
Fix any finite initial misspecification concern $\lambda_1\in[0,\infty)$. Under the hypotheses of Theorem \ref{thm:dynamic:trap-scaleup}, there exists $\bar\theta_L\in(0,p)$ such that for all $\theta_L\in(0,\bar\theta_L)$:
\begin{enumerate}[label=(\roman*),leftmargin=18pt]
\item innovation is implementable at $t=1$ from the initial state $(\mu_1,\lambda_1)$;
\item continued innovation is not implementable at $t=2$ after $(a_1=2,y_1=1)$;
\item continued innovation is implementable at $t=2$ after $(a_1=2,y_1=0)$.
\end{enumerate}
\end{corollary}

Corollary \ref{cor:dynamic:success-failure-reversal} shows that the difficulty is not initiating innovation. For sufficiently exploratory projects, the principal can motivate innovation at $t=1$ for any bounded initial misspecification concern. The difficulty appears at the continuation stage: after an unexplained success, continued innovation can become uncontractible, while it remains implementable after failure. Thus, our framework delivers a reversal of implementability across histories: innovation can be launched, but an encouraging outcome can move the relationship into a state where  incentives lose power. We illustrate these dynamics below.

\subsubsection{Illustration: incentive curves}

A convenient way to analyze the breakthrough-trap mechanism is to plot the post-success continuation wedge as a function of the promised utility spread. Let $\varDelta_{\mathfrak{s}}:=x(1)-x(0)$ denote the utility spread and normalize $x(0)=0$ and $x(1)=\varDelta_{\mathfrak{s}}\in\R$. 
After a success, let
$\mu_{\mathfrak{s}}:=\mu_2(q^2_H|y_1=1)$ and $\lambda_{\mathfrak{s}}:=\lambda_2(1).$
Then, the post-success continuation wedge is
\begin{align*}
M_{\mathfrak{s}}(\varDelta_{\mathfrak{s}})
&:=M(\mu_{\mathfrak{s}},\lambda_{\mathfrak{s}};(0,\varDelta_{\mathfrak{s}}))\\
&=
\mu_{\mathfrak{s}}\Big[-\frac{1}{\lambda_{\mathfrak{s}}}
\log\frac{1-\theta_H+\theta_He^{-\lambda_{\mathfrak{s}}\varDelta_{\mathfrak{s}}}}{1-p+pe^{-\lambda_{\mathfrak{s}}\varDelta_{\mathfrak{s}}}}\Big]
+(1-\mu_{\mathfrak{s}})\Big[-\frac{1}{\lambda_{\mathfrak{s}}}
\log\frac{1-\theta_L+\theta_Le^{-\lambda_{\mathfrak{s}}\varDelta_{\mathfrak{s}}}}{1-p+pe^{-\lambda_{\mathfrak{s}}\varDelta_{\mathfrak{s}}}}\Big].
\end{align*}
This curve is the continuation incentive technology facing the principal after a breakthrough. To isolate the role of endogenous misspecification, let us compare $M_{\mathfrak{s}}(\varDelta_{\mathfrak{s}})$ above with the EU benchmark. Define
$\bar\theta_{\mathfrak{s}}:=\mu_{\mathfrak{s}}\theta_H+(1-\mu_{\mathfrak{s}})\theta_L.$
Then, the EU continuation wedge is simply the line
$M_{\mathfrak{s}}^{\text{EU}}(\varDelta_{\mathfrak{s}})=\big(\bar\theta_{\mathfrak{s}}-p\big)\varDelta_{\mathfrak{s}}.$
Thus, after success, Bayesian learning alone makes the implementable wedge grow linearly in the promised utility spread. In contrast, under ARC, the same spread is compressed by the misspecification shock $\lambda_{\mathfrak{s}}$. \par For concreteness, consider the following parameterization:
$$
p=0.40,\qquad \theta_H=0.45,\qquad \theta_L=0.02,\qquad m=0.50,\qquad \gamma=0.40,\qquad k=0.05.
$$
After a success,
$\mu_{\mathfrak{s}}=\frac{m\theta_H}{m\theta_H+(1-m)\theta_L}\approx 0.957,$
$\lambda_{\mathfrak{s}}=\frac{-\log\theta_H}{\gamma}
\approx 1.996,$ and
$\lambda^*=\frac{1}{k}\log\frac{1-p}{1-\theta_H}
\approx 1.74,$
so $\lambda_{\mathfrak{s}}>\lambda^*$ holds here because an early success is surprising even under the most optimistic model $q^2_H$. We can now verify whether continued innovation is implementable.
\par Under the EU benchmark, the posterior mean success probability becomes
$\bar\theta_{\mathfrak{s}}\approx 0.43,$
so the EU continuation wedge is $M_{\mathfrak{s}}^{\text{EU}}(\varDelta_{\mathfrak{s}})\approx 0.03\varDelta_{\mathfrak{s}}.$
Thus, any promised utility spread greater than or equal to
$\varDelta_{\mathfrak{s}}^{\text{EU}}:=\frac{k}{\bar\theta_{\mathfrak{s}}-p}\approx 1.58$
would implement continued innovation under Bayesian learning alone, where success affects only posterior beliefs. \par Under the ARC criterion, however, the incentive capacity \eqref{eq:dynamic:Delta-C} is bounded above:
\begin{align*}
C_{\mathfrak{s}}:=\sup_{\varDelta_{\mathfrak{s}}\in\R}M_{\mathfrak{s}}(\varDelta_{\mathfrak{s}})
=\lim_{\varDelta_{\mathfrak{s}}\to\infty}M_{\mathfrak{s}}(\varDelta_{\mathfrak{s}})=
\frac{\mu_{\mathfrak{s}}}{\lambda_{\mathfrak{s}}}\log\frac{1-p}{1-\theta_H}
+\frac{1-\mu_{\mathfrak{s}}}{\lambda_{\mathfrak{s}}}\log\frac{1-p}{1-\theta_L}
\approx 0.03
<0.05
=k.
\end{align*}

\begin{figure}[hbt!]
\centering
\begin{tikzpicture}
\begin{axis}[
width=0.9\textwidth,
height=0.58\textwidth,
xmin=0, xmax=5,
ymin=0, ymax=0.17,
xlabel={success-contingent continuation spread $\varDelta_{\mathfrak{s}}=x(1)-x(0)$},
ylabel={post-success continuation wedge},
axis lines=left,
samples=250,
domain=0:5,
scaled y ticks=false,
legend style={
font=\footnotesize,
draw=none,
fill=none,
at={(0.02,0.98)},
anchor=north west,
row sep=1pt,
/tikz/every even column/.append style={column sep=3pt}
},
legend cell align=left,
label style={font=\small},
tick label style={font=\scriptsize},
xlabel style={yshift=0.01em},
xtick={0,0.5,1,1.577,2,2.5,3,3.5,4,4.5},
xticklabels={0,0.5,1,$\varDelta_{\mathfrak{s}}^{\text{EU}}$,2,2.5,3,3.5,4,4.5},
yticklabel style={
font=\scriptsize,
/pgf/number format/fixed,
/pgf/number format/precision=2
},
extra y ticks={0.0312739, 0.05},
extra y tick labels={$C_{\mathfrak{s}}$, $k$},
]
\addplot[name path=eu, very thick, ForestGreen]
{0.0317021276595744*x};
\addlegendentry{$M_{\mathfrak{s}}^{\text{EU}}(\varDelta_{\mathfrak{s}})$}

\addplot[name path=arc, very thick, BrickRed]
{0.957446808510638*(-1/1.99626924054443*ln((0.55 + 0.45*exp(-1.99626924054443*x))/(0.6 + 0.4*exp(-1.99626924054443*x))))
+0.0425531914893617*(-1/1.99626924054443*ln((0.98 + 0.02*exp(-1.99626924054443*x))/(0.6 + 0.4*exp(-1.99626924054443*x))))};
\addlegendentry{$M_{\mathfrak{s}}(\varDelta_{\mathfrak{s}})$}
\addplot[name path=kline, black, dashed, thick, forget plot] {0.05};
\addplot[black, densely dotted, thick, forget plot] {0.0312739349092558};

\addplot[forget plot, ForestGreen!18, draw=none]
fill between[of=eu and kline, soft clip={domain=1.577:5}];
\addlegendimage{area legend, draw=ForestGreen!50!black, fill=ForestGreen!18}
\addlegendentry{incentive surplus under Bayesian learning alone}

\addplot[forget plot, BrickRed!18, draw=none]
fill between[of=eu and arc, soft clip={domain=0:1.577}];

\addplot[forget plot, BrickRed!18, draw=none]
fill between[of=kline and arc, soft clip={domain=1.577:5}];
\addlegendimage{area legend, draw=BrickRed!60!black, fill=BrickRed!18}
\addlegendentry{incentive deficit when success is unexplained}

\addplot[forget plot, purple!18, draw=none]
fill between[of=kline and eu, soft clip={domain=0:1.577}];
\addlegendimage{area legend, draw=purple!60!black, fill=purple!18}
\addlegendentry{incentive deficit under both paradigms}

\draw[densely dashed] (axis cs:1.577,0) -- (axis cs:1.577,0.05);
\end{axis}
\end{tikzpicture}
\caption{\footnotesize EU spread $\varDelta_{\mathfrak{s}}^{\text{EU}}\approx 1.58$, innovation cost $k=0.05$, and incentive capacity $C_{\mathfrak{s}}\approx 0.03$. The purple region captures the incentive deficit of implementing innovation that is common to both learning paradigms, which is associated with all utility spreads that are strictly smaller than $\varDelta_{\mathfrak{s}}^{\text{EU}}$.}
\label{fig:dynamic:capacity}
\end{figure}

Thus, after a breakthrough, even arbitrarily large utility spreads cannot restore implementability of continued innovation, i.e., a breakthrough trap. To help visualize the underlying mechanism, Figure \ref{fig:dynamic:capacity} plots the curves $M_{\mathfrak{s}}(\varDelta_{\mathfrak{s}})$ and $M^{\text{EU}}_{\mathfrak{s}}(\varDelta_{\mathfrak{s}})$ as functions of the utility spread $\varDelta_{\mathfrak{s}}\in\R$, using shading to highlight their distance from the cost $k$. The green region represents the incentive surplus associated with utility spreads that would make innovation implementable under Bayesian learning alone. The red region shows that, after the misspecification shock, the same post-success node is bounded strictly below $k$ for all utility spreads. In summary, this figure depicts the main economic tradeoff. On one hand, Bayesian learning \textit{raises} the slope of the continuation wedge because success shifts weight toward $q^2_H$. On the other hand, endogenous robustness adjustments \textit{bend} the entire curve downward because a larger $\lambda_{\mathfrak{s}}$ makes the agent value success-contingent continuation pay less: the more diagnostic the breakthrough is about roadmap failure, the less incentive leverage the principal can extract from any given promised utility spread.

\par The ``mirror image'' of Figure \ref{fig:dynamic:capacity} plots implementability after an early failure to clarify that the breakthrough trap is driven by the epistemic tension between explained and unexplained evidence. When innovation is more likely to generate failure than the routine action, a larger failure-contingent continuation payoff is what makes innovation incentive-compatible: failure is now the continuation outcome more indicative of innovative activity. 
\par The failure-contingent continuation spread becomes $\varDelta_{\mathfrak{f}}:=x(0)-x(1)$, so set $x(1)=0$ and $x(0)=\varDelta_{\mathfrak{f}}\in\R$. Define the posteriors
$\mu_{\mathfrak{f}}:=\mu_2(q^2_H| y_1=0),$
$\lambda_{\mathfrak{f}}:=\lambda_2(0),$
and
$\bar\theta_{\mathfrak{f}}:=\mu_{\mathfrak{f}}\theta_H+(1-\mu_{\mathfrak{f}})\theta_L.$
The continuation wedges are
$M_{\mathfrak{f}}^{\mathrm{EU}}(\varDelta_{\mathfrak{f}})
=(p-\bar\theta_{\mathfrak{f}})\varDelta_{\mathfrak{f}}$
and
$M_{\mathfrak{f}}(\varDelta_{\mathfrak{f}})
:=M(\mu_{\mathfrak{f}},\lambda_{\mathfrak{f}};(\varDelta_{\mathfrak{f}},0)).$
Using the parameterization in Figure \ref{fig:dynamic:capacity}, 
$\mu_{\mathfrak{f}}\approx 0.36,$
$\lambda_{\mathfrak{f}}\approx 0.05,$
and
$\bar\theta_{\mathfrak{f}}\approx 0.17.$ After failure, the misspecification threshold is $\lambda_*:=\frac{1}{k}\log \frac{p}{\theta_L}\approx60$, so $\lambda_{\mathfrak{f}}<\lambda_*$. 
\par  Unlike an early success, early failure induces a very negligible misspecification shock because it is well explained by the roadmap. Thus, both continuation wedges are increasing in the post-failure premium $\varDelta_{\mathfrak{f}}$, both cross the cost $k$ near $\varDelta^{\text{EU}}_{\mathfrak{f}}\approx 0.22$, and notably, the ARC wedge now lies above the EU benchmark. Figure \ref{fig:dynamic:capacity-failure} illustrates the epistemic asymmetry of Theorem \ref{thm:dynamic:trap-scaleup}: while unexplained success tends to reduce incentive leverage, explained failure can preserve it  (Corollary \ref{cor:dynamic:success-failure-reversal}) and even strengthen it relative to the EU benchmark, as shown here. Thus, what matters for continued innovation is not whether early evidence is good or bad, but whether it is \emph{explained} or \emph{unexplained} by the roadmap. 

\begin{figure}[hbt!]
\centering
\begin{tikzpicture}
\begin{axis}[
width=0.9\textwidth,
height=0.58\textwidth,
xmin=0, xmax=3,
ymin=0, ymax=0.705,
xlabel={failure-contingent continuation spread $\varDelta_{\mathfrak{f}}=x(0)-x(1)$},
ylabel={post-failure continuation wedge},
axis lines=left,
samples=250,
domain=0:5,
scaled y ticks=false,
legend style={
font=\footnotesize,
draw=none,
fill=none,
at={(0.02,0.98)},
anchor=north west,
row sep=1pt,
/tikz/every even column/.append style={column sep=3pt}
},
legend cell align=left,
label style={font=\small},
tick label style={font=\scriptsize},
xlabel style={yshift=0.01em},
xtick={0,0.2217,0.5,1,1.5,2,2.5,3},
xticklabels={0,$\varDelta^{\text{EU}}_{\mathfrak{f}}$,0.5,1,1.5,2,2.5,3},
ytick={0,0.05,0.25,0.5},
yticklabels={0,$k$,0.25,0.5},
]

\addplot[name path=eu, very thick, ForestGreen]
{0.225424836601307*x};
\addlegendentry{$M_{\mathfrak{f}}^{\mathrm{EU}}(\varDelta_{\mathfrak{f}})$}

\addplot[name path=arc, very thick, BrickRed]
{0.359477124183007*(-1/0.0505067682937987*ln((0.45 + 0.55*exp(-0.0505067682937987*x))/(0.4 + 0.6*exp(-0.0505067682937987*x))))
+0.640522875816993*(-1/0.0505067682937987*ln((0.02 + 0.98*exp(-0.0505067682937987*x))/(0.4 + 0.6*exp(-0.0505067682937987*x))))};
\addlegendentry{$M_{\mathfrak{f}}(\varDelta_{\mathfrak{f}})$}

\addplot[name path=kline, black, dashed, thick, forget plot] {0.05};

\addplot[forget plot, ForestGreen!18, draw=none]
fill between[of=eu and kline, soft clip={domain=0.2217:5}];
\addlegendimage{area legend, draw=ForestGreen!50!black, fill=ForestGreen!18}
\addlegendentry{incentive surplus under both paradigms}

\addplot[forget plot, yellow!40, draw=none]
fill between[of=arc and eu, soft clip={domain=0:5}];
\addlegendimage{area legend, draw=yellow!80!black, fill=yellow!40}
\addlegendentry{additional surplus when failure is explained}

\addplot[forget plot, purple!18, draw=none]
fill between[of=kline and arc, soft clip={domain=0:0.2210}];
\addlegendimage{area legend, draw=purple!60!black, fill=purple!18}
\addlegendentry{incentive deficit under both paradigms}

\draw[densely dashed] (axis cs:0.2217,0) -- (axis cs:0.2217,0.05);

\end{axis}
\end{tikzpicture}
\caption{\footnotesize Post-failure implementability of innovation under the same parameterization as Figure \ref{fig:dynamic:capacity}. 
}
\label{fig:dynamic:capacity-failure}
\end{figure}

\subsubsection{Illustration: expected-utility benchmark}\label{sec:eu}

We now clarify that the breakthrough trap in Theorem \ref{thm:dynamic:trap-scaleup} is driven entirely by endogenous misspecification concerns. We achieve this by considering an EU maximizing agent as in \citet{manso2011}. In our framework, this is the case where the agent has no concern for misspecification, i.e., $\lambda_t(\cdot)=0$ for all $t$. Since \eqref{eq:dynamic:gq}--\eqref{eq:dynamic:G} are defined for $\lambda>0$, we extend them at $\lambda=0$ by taking limits: for every $x\in\R^Y$ and $\mu\in \Delta(Q^{2})$, $\lim_{\lambda\downarrow 0}g_q(x;\lambda)=\E_q[x]$ and $\lim_{\lambda\downarrow 0}\mathcal{G}(x;\mu,\lambda)=\E_{\bar q_\mu}[x]$, where $\bar q_\mu:=\sum_q \mu(q)\hspace{0.02in}q$. Now, define $\zeta_x:=x(1)-x(0)$ and $\bar\theta(\mu):=\bar q_\mu(1)$. Then, at state $(\mu,0)$, the incentive gap becomes
$M(\mu,0;x)
= \mathcal{G}(x;\mu,0)-g_{q^1}(x;0)
 = (\bar\theta(\mu)-p)\zeta_x$ and the incentive capacity  becomes
$$
C(\mu,0)
 = \sup_{x:Y\to\R} M(\mu,0;x)
 =\
\begin{cases}
+\infty, & \text{if }\bar\theta(\mu)\neq p,\\
0, & \text{if }\bar\theta(\mu)=p.
\end{cases}
$$
Indeed, if $\bar\theta(\mu)>p$, choosing $\zeta_x\to+\infty$ makes $M(\mu,0;x)\to+\infty$; if $\bar\theta(\mu)<p$, choosing $\zeta_x\to-\infty$ makes $M(\mu,0;x)\to+\infty$; and if $\bar\theta(\mu)=p$, then $M(\mu,0;x)=0$ for all $x$. Since $\mu_2(\cdot|y_1)$ has full support on $Q^{2}=\{q^2_L,q^2_H\}$ whenever $m\in(0,1)$ and $0<\theta_L<\theta_H<1$, we have $\bar\theta(\mu_2(\cdot|y_1))\in(\theta_L,\theta_H)$. Thus, except in the knife-edge case where the updated belief leaves the agent indifferent between continued innovation and the routine alternative  (i.e., $\bar\theta(\mu_2(\cdot|y_1)) = p$), incentive capacity is always infinite:
$C(\mu_2(\cdot|y_1),0)=+\infty$, for all $y_1\in Y$.
\par To summarize, in the EU benchmark, the post-success incentive capacity condition required for existence of a breakthrough trap in Theorem \ref{thm:dynamic:trap-scaleup} does \textit{not} hold generically. This means that continued innovation under EU is generically implementable after a success. The breakthrough trap is therefore a misspecification-driven phenomenon.

\subsubsection{Illustration: state-space diagram}

Having established that breakthrough traps do not arise in the expected-utility benchmark, Figure \ref{fig:dynamic:state-space-feasibility} provides a state-space representation of the continuation problem by combining the insights of Figures \ref{fig:dynamic:capacity} and \ref{fig:dynamic:capacity-failure} in a single diagram. The horizontal axis $\lambda=0$ is the Bayesian paradigm: along this axis, evidence affects only posterior beliefs, so a breakthrough simply moves the state rightward from $m$ to $\mu_{\mathfrak s}$, which is the good-news effect. Under ARC, however, evidence also changes the agent's concern for misspecification, so the state space becomes two-dimensional. In particular, an early success induces not only the horizontal shift from $m$ to $\mu_{\mathfrak s}$ but also the vertical shift from an initial concern $\lambda_1=0$ to $\lambda_{\mathfrak s}$, which is the bad-news effect generated by unexplained evidence.
\par Under the parameterization in Figure \ref{fig:dynamic:capacity}, Corollary \ref{cor:dynamic:success-failure-reversal} implies that innovation is implementable at $t=1$ for the initial state $(m,\lambda_1)$. Thus, the key economic tension is not that innovation is hard to initiate, but that it may be harder to sustain precisely after the most encouraging outcome is realized. Following an early success, the state moves to $(\mu_{\mathfrak s},\lambda_{\mathfrak s})$, where the principal can no longer implement innovation, whereas an early failure moves the state to $(\mu_{\mathfrak f},\lambda_{\mathfrak f})$, where innovation remains implementable. Figure \ref{fig:dynamic:state-space-feasibility} therefore highlights the key tradeoff: success relaxes the belief side of the agency problem, but can simultaneously tighten the incentive side enough to make continued innovation infeasible.

\begin{figure}[hbt!]
\centering
\begin{tikzpicture}

\colorlet{DarkOrange}{BurntOrange!75!black}

\pgfmathsetmacro{\p}{0.40}
\pgfmathsetmacro{\thH}{0.45}
\pgfmathsetmacro{\thL}{0.02}
\pgfmathsetmacro{\mprior}{0.50}
\pgfmathsetmacro{\kcost}{0.05}
\pgfmathsetmacro{\gammaBase}{0.40}

\pgfmathsetmacro{\mus}{\mprior*\thH/(\mprior*\thH + (1-\mprior)*\thL)}
\pgfmathsetmacro{\muf}{\mprior*(1-\thH)/(\mprior*(1-\thH) + (1-\mprior)*(1-\thL))}

\pgfmathsetmacro{\lams}{-ln(\thH)/\gammaBase}
\pgfmathsetmacro{\lamf}{-ln(1-\thL)/\gammaBase}

\pgfmathsetmacro{\Aplus}{ln((1-\p)/(1-\thH))}
\pgfmathsetmacro{\Bplus}{ln((1-\p)/(1-\thL))}
\pgfmathsetmacro{\Aminus}{ln(\p/\thH)}
\pgfmathsetmacro{\Bminus}{ln(\p/\thL)}

\pgfmathsetmacro{\extraYtick}{\Aplus/\kcost}

\pgfmathsetmacro{\xMin}{0.30}
\pgfmathsetmacro{\xMax}{1.00}
\pgfmathsetmacro{\yMax}{2.40}

\begin{axis}[
width=0.9\textwidth,
height=0.5\textwidth,
xmin=\xMin, xmax=\xMax,
ymin=0, ymax=\yMax,
axis lines=left,
axis on top, 
xlabel={posterior weight on optimistic model $\mu=\mu(q^2_H)$},
ylabel={misspecification concern $\lambda$},
samples=250,
scaled y ticks=false,
legend style={
font=\scriptsize,
draw=none,
fill=none,
at={(0.02,0.995)}, 
anchor=north west,
row sep=1pt
},
legend cell align=left,
label style={font=\small},
tick label style={font=\scriptsize},
xlabel style={yshift=0.01em},
xtick={0.40,0.60,0.70,0.80,0.90}, 
extra x ticks={\muf,0.50,\mus},
extra x tick labels={$\mu_{\mathfrak f}$,$m$,$\mu_{\mathfrak s}$},
ytick={0.5,1.0,1.5}, 
extra y ticks={ \lams, \lamf}, 
extra y tick labels={ $\lambda_{\mathfrak s}$, $\lambda_{\mathfrak f}$}, 
clip=false,
]

\addplot[
name path=frontierclip,
draw=none,
domain=\xMin:\xMax
]
{min(\yMax, max((x*\Aplus + (1-x)*\Bplus)/\kcost, (x*\Aminus + (1-x)*\Bminus)/\kcost))};

\addplot[name path=zeroline, draw=none, domain=\xMin:\xMax, forget plot] {0};
\addplot[name path=topline, draw=none, domain=\xMin:\xMax, forget plot] {\yMax};

\addplot[ForestGreen!18, draw=none, forget plot]
fill between[of=zeroline and frontierclip];

\addplot[BrickRed!18, draw=none, area legend, legend image post style={fill=BrickRed!30}]
fill between[of=frontierclip and topline];
\addlegendentry{innovation-implementability frontier $\bar\lambda(\mu)$} 
\begin{scope}
\clip (axis cs:\xMin,0) rectangle (axis cs:\xMax,\yMax);
\addplot[
very thick,
black,
domain=0.25:1.05, 
legend image post style={very thick} 
]
{max((x*\Aplus + (1-x)*\Bplus)/\kcost, (x*\Aminus + (1-x)*\Bminus)/\kcost)};
\addlegendentry{breakthrough-trap region} 
\end{scope}

\addplot[
only marks,
mark=*,
mark size=1.5pt,
color=black,
forget plot
]
coordinates {(\mprior,0)};

\addplot[
only marks,
mark=*,
mark size=1.5pt,
color=DarkOrange, 
forget plot
]
coordinates {(\mus,\lams)};
 
\addplot[
only marks,
mark=*,
mark size=1.5pt,
color=NavyBlue,
forget plot
]
coordinates {(\muf,\lamf)};

\draw[->, densely dashed, thin, DarkOrange, shorten >=2pt]
(axis cs:\mprior,0) -- node[sloped, above, font=\scriptsize, text=DarkOrange] {success update} (axis cs:\mus,\lams);

\draw[->, densely dashed, thin, NavyBlue, shorten >=2pt]
(axis cs:\mprior,0) -- node[sloped, above, font=\scriptsize, text=NavyBlue] {failure update} (axis cs:\muf,\lamf);

\draw[densely dotted, black, very thin]
(axis cs:\xMin,\lams) -- (axis cs:\mus,\lams);

\draw[densely dotted, black, very thin]
(axis cs:\mus,0) -- (axis cs:\mus,\lams);

\draw[densely dotted, black, very thin]
(axis cs:\xMin,\lamf) -- (axis cs:\muf,\lamf);

\draw[densely dotted, black, very thin]
(axis cs:\muf,0) -- (axis cs:\muf,\lamf);

\node[left, font=\scriptsize, yshift=-5pt, xshift=0.5pt] at (axis cs:\xMin, 0) {0};

\end{axis}
\end{tikzpicture}
\caption{\footnotesize
The green region contains states $(\mu,\lambda)$ at which innovation is implementable at $t=2$; the red region contains states at which only the routine action is implementable, i.e., the breakthrough-trap region. 
The black curve is the upper envelope $\bar\lambda(\mu)=\max\{\lambda^+(\mu),\lambda^-(\mu)\}$,
which separates innovation-implementable states from routine-only states.
The arrows are the one-step update transitions from the initial state $(m,\lambda_1)$ after success or failure under the parametrization in Figure \ref{fig:dynamic:capacity}, where we set $\lambda_1=0$.}
\label{fig:dynamic:state-space-feasibility}
\end{figure}

\par More formally, Figure \ref{fig:dynamic:state-space-feasibility} summarizes the state evolution and contractibility through the innovation-implementability frontier $\bar\lambda(\mu):=\max\{\lambda^{+}(\mu),\lambda^{-}(\mu)\}$, where
$$
\lambda^{+}(\mu):=\frac{1}{k}\left[\mu\log\frac{1-p}{1-\theta_H}+(1-\mu)\log\frac{1-p}{1-\theta_L}\right], \quad 
\lambda^{-}(\mu):=\frac{1}{k}\left[\mu\log\frac{p}{\theta_H}+(1-\mu)\log\frac{p}{\theta_L}\right]
$$
are the maximal misspecification levels consistent with continued innovation when the continuation contract rewards success or failure, respectively, for any belief $\mu\in\Delta(Q^{2})$. Figure \ref{fig:dynamic:state-space-feasibility} should therefore be interpreted as a map of innovation implementability, not as a map of innovation desirability at a given posterior. The green region at low $\mu$ does not mean that innovation is attractive in the usual success-probability sense; it simply means  that there exists some continuation contract that can motivate innovation there. 
\subsection{Discussion of results and assumptions}\label{sec:disc}
\textbf{Scope.} Appendix \ref{sec:general} shows that the mechanism underlying the breakthrough trap is not tied to the functional form of the ARC criterion or the LLR updating rule. The mechanism holds under more general specifications, but we focus mainly on the ARC-LLR pair because it allows us to present this mechanism in its most tractable representation. Similarly, the mechanism does not depend on the shape of the agent's utility or his liability; however, note that assuming limited liability would only reinforce the mechanism.\footnote{To see this, suppose $w(y)\ge \underline w$. Then, attainable utility vectors satisfy
$x(y)=u(w(y))\ge u(\underline w)\hspace{0.03in} \forall y$, so the feasible set of continuation utility vectors shrinks from $\mathbb R^Y$ to
$X^{LL}:=\{x\in\mathbb R^Y: x(y)\ge u(\underline w)\hspace{0.03in}\forall y\}$.
Define the liability-constrained capacity
$C^{LL}(\mu,\lambda):=\sup_{x\in X^{LL}} M(\mu,\lambda;x)$.
Thus,
$C^{LL}(\mu,\lambda)\le C(\mu,\lambda)$.} Moreover, if  the routine arm is treated as a ``known arm'' in the strong sense that the agent is assumed to know its true distribution, then the mechanism would also be reinforced.\footnote{The reason is that the agent would now evaluate the routine arm using its expected utility instead of ARC, but this would only make the routine arm more valuable because $g_{q^1}(x,\lambda)\leq\E_{q^1}[x(\cdot)]$ $\forall\lambda>0$.}

\medskip
\noindent\textbf{Roadmap.} We borrow the term ``roadmap'' from the title of \citet[][Chapter 8.2.2]{mgt11}: ``A Roadmap into Unknown Terrain'' because they argue that it is a ``rational coordination device'' in organizations.  This is consistent with \citet[][p. 621]{van10} who argues: ``The manager's beliefs give \textit{direction} to the firm by influencing the employees’ decisions, and thus also lead to coordination.'' Thus, we interpret  a roadmap as a model-belief ``library'' that managers share with employees to guide them in the execution of new tasks. \citet[][]{hansenmiss25} justify roadmap \textit{exogeneity} by interpreting structured models as outcomes of scientific knowledge.\footnote{\citet[][p. 2]{hansenmiss25} remark: ``These structured models are ones that are explicitly motivated or featured, such as ones with substantive motivation or scientific underpinnings [...] They may be based on scientific knowledge relying on empirical evidence and theoretical arguments.''}
More broadly, this exogenous communication is standard in the literature.\footnote{For example, \citet[][p. 1154]{amb11} interpret a roadmap as an exogenous ``consistency'' requirement between the beliefs of the principal and agent. \citet[][p. 1411]{miao16} assume that both the principal and agent have exogenous access to the same  model. \citet[][p. 88]{kellner17} assumes that both the principal and
the agent exogenously share the set of models and have a common prior.} Appendix \ref{sec:disc:roadmap-no-free-lunch} allows \textit{endogenous} roadmaps and shows that they have very limited effect on the breakthrough-trap mechanism; for example, designing a roadmap to avoid breakthrough traps can backfire in the long run because it involves trading off local gains against global misspecification cost.

\medskip
\noindent\textbf{Misspecification.} Classical moral hazard assumes rational expectations---both the principal and agent know the true outcome distribution. \citet[][]{miao16} suggest a departure from this assumption by acknowledging concerns for model misspecification in agency problems, and they also capture these concerns using \citeauthor{hansen01}'s (\citeyear{hansen01}) multiplier robustness.  
This departure is consistent with innovation domains such as materials science, where scientists rely on their labs' models and equipment to predict new experimental outcomes. In these settings, scientists would face model uncertainty in their respective tasks, whereas the lab manager may not face such uncertainty because she can aggregate information across all her scientists' tasks. This information asymmetry is standard in the literature \citep[e.g.,][]{maskin90,maskin92,amb11,kellner15,kellner17,amb24}. For example, \citet[][p. 1149]{amb11} justify it as: ``The principal owns the production technology and can evaluate the relation between effort and output more precisely.'' For further details, see Section \ref{sec:lit}.

\section{Managing Breakthrough Traps}\label{sec:avoidtrap}
Thus far, we have shown that breakthrough traps create costly bottlenecks in innovation projects. We will now propose three managerial tools that can help mitigate these traps: (1) feedback design, (2) screening at hiring, and (3) post-success turnover.
These tools aim to complement incentive design by affecting the evolution of the agent's misspecification concerns, and therefore the set of continuation contracts that can be sustained. 
\subsection{Feedback design}\label{subsec:dynamic:feedback}

We study how the principal can use \emph{feedback design}---a choice of what
performance information is publicly observed after period 1---to prevent breakthrough traps. Feedback design is a well-established organizational tool, which is effective when the principal is better able to evaluate performance than the agent \citep[e.g.,][Section VII]{manso2011}.

\medskip
\noindent\textbf{Feedback policies.}
Let $\mathcal{S}=\{0,1\}$ denote a binary \emph{evaluation} signal. A feedback \emph{policy} is a Markov
kernel $\Pi: Y\to \Delta(\mathcal{S})$, written as $\Pi(s|y)$ for $s\in \mathcal{S}$ and $y\in Y$.
After $y_1$ is realized, the public signal $s_1$ is drawn according to $\Pi(\cdot|y_1)$ and is
the only period-1 outcome statistic observed by the agent, and hence the only statistic on which
the continuation decision at $t=2$ can be conditioned. Since failing to reward a genuine breakthrough would unravel the agent's incentives, we focus on policies that do not generate \emph{false negatives}, i.e., do not classify a true success as an unfavorable evaluation.

\medskip
\noindent\textbf{Signal likelihoods.}
For any  $q\in\Delta(Y)$, define the induced signal distribution as
$q^\Pi(s)\hspace{0.02in}:=\sum_{y\in Y} q(y)\hspace{0.02in}\Pi(s|y),\hspace{0.02in} \forall s\in \mathcal{S}.$
If $a_1=2$ and $s_1$ is observed at $t=1$, the agent updates the second-order prior $\mu_1$ by Bayes rule using the
likelihoods $\{q^\Pi(s_1): q\in Q^{2}\}$ as
$\mu_2(q|s_1)\propto \mu_1(q)\hspace{0.02in}q^\Pi(s_1)$ $\forall q\in Q^{2}.$
We maintain the LLR updating rule (Assumption \ref{ass:dynamic:llr}), but now applied to the
observed feedback, so \eqref{eq:dynamic:lambda-update} becomes $\lambda_2(s_1)=-\frac{1}{\gamma}\log{\text{\normalfont max}}_{q\in Q^{2}} q^\Pi(s_1).$

\medskip
\noindent\textbf{Coarse evaluations.}
For $r\in[0,1]$, the \emph{coarse evaluation} (type-I error) policy $\Pi^r$ is
\begin{align}
\Pi^r(1|1)=1,\qquad \Pi^r(1|0)=r,
\label{eq:dynamic:Pi-r}
\end{align}
so that true successes are always classified as favorable ($s=1$) while failures are classified as
favorable with probability $r\in[0,1]$. The parameter $r$ governs how coarse the
evaluation is: higher $r$ makes $s=1$ less statistically surprising under the structured models,
thereby mitigating the LLR shock in the agent's misspecification concern $\lambda$.

\medskip
\noindent\textbf{Design principle.}
For each coarse evaluation $\Pi^r$, define
$q_H^r:=\theta_H+(1-\theta_H)r,$ and $
q_L^r:=\theta_L+(1-\theta_L)r.$
After a favorable evaluation $s_1=1$, the posterior and misspecification concern are
$\mu_2^r(q_H^2|s_1=1)
=
\frac{m q_H^r}{m q_H^r+(1-m)q_L^r},$ $\lambda_2^r(1)
=
-\frac{1}{\gamma}\log q_H^r.$
Thus, coarsening affects both components of the continuation state. The relevant design object is therefore the post-evaluation capacity
$C\big(\mu_2^r(\cdot|s_1=1),\lambda_2^r(1)\big).$
Fix any small $\eta>0$ and consider the set $F_\eta:=\{r\in[0,1]:C(\mu_2^r(\cdot|s_1=1),\lambda_2^r(1))\ge k+\eta\}$. The next result shows that $F_\eta$ is nonempty and closed under the hypotheses of Corollary \ref{cor:dynamic:success-failure-reversal}, so $r_\eta:=\min F_\eta$ is well defined. 

\begin{corollary}\label{prop:dynamic:feedback-optimal}
Restrict attention to binary feedback policies $\Pi$ satisfying \emph{no false negatives}, i.e., $\Pi(1|1)=1$. Fix $\eta>0$. Under the hypotheses of Corollary \ref{cor:dynamic:success-failure-reversal}, $F_\eta$ is nonempty and closed. Moreover, $\Pi^{r_\eta}$ preserves continued innovation after a favorable evaluation with slack $\eta$, i.e., $C(\mu_2^{r_\eta}(\cdot|s_1=1),\lambda_2^{r_\eta}(1))\ge k+\eta$, and uniquely minimizes the false-positive rate $\Pi(1|0)$ among all no-false-negative policies satisfying this capacity inequality.
\end{corollary}

Policy $\Pi^{r_\eta}$ captures the minimal tolerance for early failure needed to preserve continuation incentives after a favorable evaluation. Under $\Pi^{r_\eta}$, true successes are always reported as favorable, while true failures are also reported as favorable with the smallest probability that keeps the favorable-evaluation state within the innovation-implementable region. This differs from \citeauthor{manso2011}'s (\citeyear{manso2011}) wage-based tolerance for early failure, which operates by reducing the agent's fear of wage punishment. Here, tolerance operates through performance evaluation: it makes a favorable evaluation less diagnostic of roadmap misspecification and hence preserves implementability of continued innovation.

\subsection{Screening at hiring}\label{subsec:dynamic:screening}

While feedback design is an effective tool to prevent breakthrough traps, its optimal implementation relies on the principal's knowledge of the agent's characteristics. Specifically, the optimal coarse evaluation policy in Corollary \ref{prop:dynamic:feedback-optimal} depends critically on the agent's update scale $\gamma$. Thus, when $\gamma$ is unknown, the principal cannot calibrate the performance evaluation to the agent's specific misspecification sensitivity, which would render the optimal feedback design infeasible. In such environments, a natural alternative is to \textit{screen} agents at hiring. This approach captures the key idea from the organizational-culture literature that shared beliefs and values in organizations are established primarily through screening at hiring rather than ex-post adaptation \citep[e.g.,][Section 2]{van10}. 
\par The need for screening arises naturally from the mechanism behind the breakthrough trap. Recall that the endogenous post-success misspecification intensity is
$\lambda_2(1)=-\frac{1}{\gamma}\log\theta_H$ (Lemma \ref{lem:dynamic:lambda-closed}). This expression reveals a
natural source of heterogeneity: agents may differ in the update
scale $\gamma>0$, which captures their sensitivity to misspecification.
We now show that such heterogeneity can be screened one-shot at hiring using a menu of
dynamic contracts whose key instrument is success-contingent continuation incentives.

\medskip
\noindent\textbf{Types and hiring menus.}
Assume $\gamma>0$ is \textit{privately} known to the agent at $t=1$ (hiring) and affects preferences only
through the LLR updating rule in Assumption \ref{ass:dynamic:llr}. A \emph{hiring menu} is a pair of dynamic contracts
$\{E,I\}$, where $\mathfrak{m}\in\{E,I\}$ specifies utility payments $(x_1^\mathfrak{m},x_2^\mathfrak{m})$. Given type $\gamma$, the
agent selects the contract that maximizes his ex-ante value. Let $V_1^\mathfrak{m}(\gamma)$ denote the
resulting equilibrium value under contract $\mathfrak{m}$ for type $\gamma$, where $V_1$ is defined formally in
\eqref{eq:dynamic:V2}--\eqref{eq:dynamic:V1} and the state $(\mu_t,\lambda_t)$ evolves according to
\eqref{eq:dynamic:mu-update} and \eqref{eq:dynamic:lambda-closed}. The key quantity here is the utility difference across contracts $E$ and $I$:
$$D(\gamma):=V_1^I(\gamma)-V_1^E(\gamma).$$

\medskip
\noindent\textbf{Routine vs. innovation tracks.}
We interpret $E$ as an \emph{exploitation track} and $I$ as an \emph{innovation track}. The screening argument below uses variation in the continuation utility at the post-success node, because this is the node where the type parameter $\gamma$ changes the misspecification concern $\lambda_2(1)=-\frac{1}{\gamma}\log\theta_H$. To capture this, we assume the continuation utility under the innovation track $I$ is nonconstant at the success history:
\begin{align}\label{eq:dynamic:screening:nondeg}
x_2^I(1,0)\neq x_2^I(1,1).
\end{align}

\begin{corollary}\label{thm:dynamic:screening}
Fix a hiring menu $\{E,I\}$. Suppose that under $E$, $a_1=1$ is chosen at $t=1$ for every
$\gamma>0$, and under $I$, $a_1=2$ is chosen at $t=1$ for every $\gamma>0$. If
\eqref{eq:dynamic:screening:nondeg} holds, then $D(\gamma)$ is strictly increasing in $\gamma$. Moreover, if $D(\gamma')<0<D(\gamma'')$ for some $\gamma'<\gamma''$, then there exists a unique
cutoff $\gamma^\ast\in(\gamma',\gamma'')$ such that $D(\gamma^\ast)=0$, where all types $\gamma>\gamma^\ast$
strictly prefer track $I$, and all types $\gamma<\gamma^\ast$ strictly prefer track $E$.
\end{corollary}

Corollary \ref{thm:dynamic:screening} characterizes a simple screening mechanism that is consistent with innovation problems: continuation incentives that reward success make the innovation track
strictly more attractive for types with larger $\gamma$, because a larger $\gamma$ dampens the
post-success rise in misspecification concerns $\lambda_2(1)=-\frac{1}{\gamma}\log\theta_H$ and therefore
raises the agent's valuation of any nonconstant success-contingent continuation pay. 
Notice also that this screening device does not require extraneous lotteries or side signals:
it relies only on the project's own  output. Combined with
Theorem \ref{thm:dynamic:trap-scaleup}, Corollary \ref{thm:dynamic:screening} clarifies which types of agents are profitable for
high-upside projects: when the scale-up value $\xi$ in \eqref{eq:multiplier} is high, the principal prefers (and can screen for)
types with larger $\gamma$, i.e., agents whose misspecification concerns do not become
overwhelming after early unexplained success.   
\subsubsection{Illustration: screening by update scale}

In this section, we visualize how the screening mechanism in Corollary \ref{thm:dynamic:screening} operates by varying the private update scale $\gamma$. Fix the parameterization in Figure \ref{fig:dynamic:capacity}. The type that makes the post-success misspecification concern equal the sufficient threshold $\lambda^*$ in Theorem \ref{thm:dynamic:trap-scaleup} is
$\gamma^{BT}:=\frac{-\log \theta_H}{\lambda^*}\approx 0.459.$
Thus, every type with $\gamma<\gamma^{BT}$ satisfies $\lambda_{\mathfrak{s}}(\gamma)>\lambda^*$ and is therefore in a breakthrough trap after success by Theorem \ref{thm:dynamic:trap-scaleup}.

To isolate the screening force, consider a hiring menu with two tracks: the exploitation track $E$ is flat, so its value does not depend on $\gamma$, whereas the innovation track $I$ instead uses a success-contingent continuation spread, as in \eqref{eq:dynamic:screening:nondeg}. The key object is the post-success continuation wedge
$M_{\mathfrak{s}}(\varDelta_I;\gamma)
:=M(\mu_{\mathfrak{s}},\lambda_{\mathfrak{s}}(\gamma);(0,\varDelta_I))$, where $\mu_{\mathfrak{s}}\approx 0.957$ and $\varDelta_I:=x_2^I(1,1)-x_2^I(1,0)\geq 0$ is the success premium on the innovation track.

\par Figure \ref{fig:dynamic:screening-example} compares two types: $\gamma_L=0.40<\gamma^{BT}$ and $\gamma_H=0.70>\gamma^{BT}$. The low type has misspecification concern  $\lambda^{\gamma_L}_{\mathfrak{s}}(1)\approx 1.996>\lambda^*\approx 1.740$ and incentive capacity $C_{\mathfrak{s}}(\gamma_L)\approx 0.031<k,$
so no success-contingent premium can implement continued innovation, i.e., he is in a breakthrough trap. In contrast, the high type has $\lambda^{\gamma_H}_{\mathfrak{s}}(1)\approx 1.141<\lambda^*$ and $
C_{\mathfrak{s}}(\gamma_H)\approx 0.055>k,$
so continued innovation is implementable for sufficiently large success-contingent premium. The smallest premium that implements innovation for high type is $\varDelta_I^{H}\approx 2.67$, so at $\varDelta_I=3$: $M_{\mathfrak{s}}(3;\gamma_L)\approx 0.031<k$ while $M_{\mathfrak{s}}(3;\gamma_H)\approx 0.051>k.$

\begin{figure}[hbt!]
\centering
\begin{tikzpicture}
\begin{axis}[
width=0.9\textwidth,
height=0.52\textwidth,
xmin=0, xmax=5,
ymin=0, ymax=0.065,
xlabel={success-contingent continuation spread $\varDelta_I=x_2^I(1,1)-x_2^I(1,0)$},
ylabel={post-success continuation wedge},
axis lines=left,
samples=250,
domain=0:5,
scaled y ticks=false,
legend style={
font=\scriptsize,
draw=none,
fill=none,
at={(1,0.05)}, 
anchor=south east,
row sep=1pt,
/tikz/every even column/.append style={column sep=3pt}
},
legend cell align=left,
label style={font=\small},
tick label style={font=\scriptsize},
xlabel style={yshift=0.01em},
xtick={0,1,2,2.667,3,4,5},
xticklabels={0,1,2,$\varDelta_I^{H}$,3,4,5},
ytick={0,0.01,0.02,0.04,0.06}, 
yticklabel style={
font=\scriptsize,
/pgf/number format/fixed,
/pgf/number format/precision=2
},
extra y ticks={0.0312739,0.05,0.0547294},
extra y tick labels={$C_{\mathfrak{s}}(\gamma_L)$,$k$,$C_{\mathfrak{s}}(\gamma_H)$},
]

\addplot[name path=high, very thick, ForestGreen]
{0.957446808510638*(-1/1.14072528031110*ln((0.55 + 0.45*exp(-1.14072528031110*x))/(0.6 + 0.4*exp(-1.14072528031110*x))))
+0.0425531914893617*(-1/1.14072528031110*ln((0.98 + 0.02*exp(-1.14072528031110*x))/(0.6 + 0.4*exp(-1.14072528031110*x))))};
\addlegendentry{$M_{\mathfrak{s}}(\varDelta_I;\gamma_H)$}

\addplot[name path=low, very thick, BrickRed]
{0.957446808510638*(-1/1.99626924054443*ln((0.55 + 0.45*exp(-1.99626924054443*x))/(0.6 + 0.4*exp(-1.99626924054443*x))))
+0.0425531914893617*(-1/1.99626924054443*ln((0.98 + 0.02*exp(-1.99626924054443*x))/(0.6 + 0.4*exp(-1.99626924054443*x))))};
\addlegendentry{$M_{\mathfrak{s}}(\varDelta_I;\gamma_L)$}

\addplot[name path=kline, black, dashed, thick, forget plot] {0.05};

\addplot[black, densely dotted, forget plot] {0.0312739}; 
\addplot[black, densely dotted, forget plot] {0.0547294};

\addplot[forget plot, purple!18, draw=none]
fill between[of=high and kline, soft clip={domain=0:2.667}];

\addplot[forget plot, ForestGreen!18, draw=none]
fill between[of=high and kline, soft clip={domain=2.667:5}];
\addlegendimage{area legend, draw=ForestGreen!50!black, fill=ForestGreen!18}
\addlegendentry{high-type incentive surplus}

\addplot[forget plot, BrickRed!18, draw=none]
fill between[of=low and kline, soft clip={domain=2.667:5}];
\addlegendimage{area legend, draw=BrickRed!60!black, fill=BrickRed!18}
\addlegendentry{low-type incentive deficit}

\addlegendimage{area legend, draw=purple!60!black, fill=purple!18}
\addlegendentry{incentive deficit for both types}

\draw[densely dashed] (axis cs:2.667,0) -- (axis cs:2.667,0.05);
\end{axis}
\end{tikzpicture}
\caption{\footnotesize The low type $\gamma_L=0.40$ remains trapped after success because his post-success continuation wedge is bounded below $k$. The high type $\gamma_H=0.70$ can still be induced to continue innovating after success once the continuation spread exceeds $\varDelta_I^{H}\approx 2.67$, which is where the two types can be separated.}
\label{fig:dynamic:screening-example}
\end{figure}

Figure \ref{fig:dynamic:screening-example} illustrates how the same success-contingent premium can fail to motivate the low type---whose trust in the roadmap collapses following an early unexplained breakthrough---yet can remain sufficient to motivate the high type. Thus, the innovation track becomes strictly more attractive as the update scale $\gamma$ rises, while the flat exploitation track does not. This insight summarizes the sorting force behind Corollary \ref{thm:dynamic:screening}: by adjusting the flat utility level on track $E$, the principal can generate a unique hiring cutoff that assigns low-$\gamma$ agents to exploitation track and high-$\gamma$ agents to innovation track.

\subsection{Post-success turnover}
If neither feedback design nor screening is feasible in practice, then what could a principal do to avoid breakthrough traps? A standard organizational tool is the policy recommending to replace the agent after early failure \citep[e.g.,][Section VI]{manso2011}. We show that an opposite policy can be effective here: it may be profitable to replace the agent after early success and hire a new agent. Building on the premise from Section \ref{subsec:dynamic:screening} that agents differ in their epistemic reactions to evidence, we treat the state variable $\lambda_t$ as a private, agent-specific sentiment toward the experimentation process. Here is the intuition:
\begin{itemize}
    \item \textit{Incumbent agent}: After early success, his posterior $\mu_2(\cdot|y_1=1)$ increases, and if he finds this success surprising, his post-success misspecification concern $\lambda_2(1)$ would increase. If $\lambda_2(1)$ is larger than the threshold $\lambda^*$, he would fall into a breakthrough trap, where post-success incentive capacity falls below his effort cost (Theorem \ref{thm:dynamic:trap-scaleup}).
    \item \textit{New agent}: A new agent is hired at $t=2$ to replace the incumbent. The principal shares with him the early success, so he inherits the high posterior $\mu_2(\cdot|y_1=1)$, but he does \textit{not} suffer the associated misspecification ``shock'' because $\lambda_2(1)$ is part of the incumbent's private experience during experimentation. When the new agent's baseline misspecification concern is low enough that incentive capacity is larger than his cost, it would be possible to motivate him to innovate in period 2.
\end{itemize}
The main assumption here is that the variable $\lambda_t$ is interpreted as an agent-specific robustness state rather than as a public sufficient statistic of raw data: it summarizes how the incumbent's own experimentation under the roadmap affects his concern for misspecification.\footnote{Here is a key distinction. The raw success $y_1=1$ is public evidence about the project, so it is used to update $\mu$. However, the diagnostic effect of having produced and interpreted that success can be agent-specific: hands-on experimentation, model audit, and personal exposure to anomaly can degrade the incumbent's trust in the roadmap in a way that is not summarized by the public outcome alone.} Turnover then works because it changes the bearer of that robustness state, not because it changes the public data. The next result formalizes this intuition.
\begin{corollary}\label{prop:dynamic:turnover}
Suppose $a_1=2$ is implemented at $t=1$.
Fix $\varphi\ge 0$.
After observing $y_1=1$, the principal may pay $\varphi$ and hire a new agent for $t=2$.
The new agent observes $y_1=1$ and has state $(\mu_2(\cdot|y_1=1),\lambda_1)$.
If $C\big(\mu_2(\cdot|y_1=1),\lambda_2(1)\big)<k< C\big(\mu_2(\cdot|y_1=1),\lambda_1\big),$
then
\begin{itemize}
    \item[(i)] A turnover contract implements $a_2=2$ after $(a_1=2,y_1=1)$.
    \item[(ii)] The principal's
net expected profit is strictly increasing in $\xi$. If the principal's net expected profit is bounded within the class of contracts that keep the incumbent after $y_1=1$, then, for every $\varphi$, turnover is strictly optimal for all sufficiently large
$\xi$.
\end{itemize}
\end{corollary}

Observe that the turnover in Corollary \ref{prop:dynamic:turnover} is not about skills: the innovation cost $k$ is a task primitive and is unchanged by replacement.
The gain comes from the fact that the new agent does not inherit the incumbent's misspecification shock. Notice that this turnover would fail if the incumbent was allowed to communicate his experience with the roadmap to the new agent because then the new agent would also inherit the state $\lambda_2(1)$. \par Thus, Corollary \ref{prop:dynamic:turnover} identifies a policy consistent with the ``paradoxes of success,'' where employees get replaced after achieving key milestones \citep[e.g.,][]{founder03}. The intuition is that the person who discovers a breakthrough may not be \textit{the} right person to scale it, not because of differences in skill, but because the act of discovery may make them too worried about model robustness to execute the scale up efficiently. 

\section{Endogenous Innovation Cycles}\label{sec:cycles}

We now analyze what happens when the breakthrough-trap mechanism operates repeatedly along an infinite-horizon path. We show that even when breakthrough traps are locally mitigated, sustained
innovation can face a \emph{speed limit}, which
implies that \textit{cycles} between routine and innovation actions may emerge endogenously in the long run.

\medskip
\noindent\textbf{Environment.}
We extend the two-period setting in Section \ref{subsec:dynamic:env} to infinite horizons. It is known in the dynamic principal-agent literature that the agent's discount rate $\beta$ must be strictly less than 1 to ensure that his continuation values are well-behaved in infinite-horizon settings, so let $\beta\in(0,1)$ \citep[e.g.,][]{radner85,spear87}. 

\par At each $t\geq1$, the agent chooses action $a_t\in A$. He updates his belief $\mu_t$ using Bayesian rule and updates his misspecification concern $\lambda_t$ using the average LLR rule in Assumption \ref{ass:dynamic:llr} extended for all $t\geq2$. Formally,
$\lambda_t(h_t)=\frac{\mathrm{LLR}^2(h_t,Q^{2})}{\gamma(t-1)} $ for $t\ge 2,$
where
\begin{align}
\mathrm{LLR}^2(h_t,Q^{2})
:=
-\log\left(
\frac{\max_{q\in Q^{2}}\prod_{\tau=1}^{t-1}q(y_\tau)^{\1\{a_\tau=2\}}}
     {\max_{p\in\Delta(Y)}\prod_{\tau=1}^{t-1}p(y_\tau)^{\1\{a_\tau=2\}}}
\right)
\qquad \forall t\ge 2,\forall h_t.
\label{eq:dynamic:infty:LLR2}
\end{align}

This rule maintains the premise that learning is generated only through innovation histories.\footnote{Setting the initial concern $\lambda_1$ to zero is without loss in this long-run analysis: allowing an initial shock $\ell_0:=\gamma\lambda_1$ would replace $\lambda_t(h_t)$ by
$\frac{\ell_0+\mathrm{LLR}^2(h_t,Q^{2})}{\gamma(t-1)}$, so the additional term $\frac{\ell_0}{\gamma(t-1)}$ vanishes as $t\to\infty$.} The product keeps those observations for which the innovation roadmap made predictions. The denominator converts this likelihood-ratio evidence into the agent's current average concern at the decision date. Thus, when the agent spends time on the routine method, he observes no new anomaly about the innovation roadmap; past unexplained innovation outcomes become less representative of his current evidence as the relationship continues without further roadmap failures. Appendix \ref{subsec:app:infty:innovation-llr} shows that this updating rule is a special case of \citet[][eq. (2)]{lanzani2025}, adapted to our innovation domain, so an agent who uses it can be viewed as a ``statistically sophisticated'' type.
\par The true data-generating process is standard: conditional on the realized action path, outcomes are independent over time, with action-dependent distribution $p^a_\ast$. This is a natural assumption in bandit or exploration-exploitation problems such as material and drug discovery because it captures the idea that an innovation attempt is a controlled experimental trial. That is, the outcome can be success or failure, its probability depends on whether the agent uses the routine or innovative method, and once the method is fixed, remaining variation reflects only trial-level noise or idiosyncratic measurement error. 
\par We focus on the case where the innovation roadmap is misspecified, so its best-fit KL discrepancy
$D_\ast^{2}:=\min_{q\in Q^{2}}\DKL(p_\ast^2\|q)$ satisfies $D_\ast^{2}>0$. Thus, $D_\ast^{2}$ measures how well structured models approximate the true outcome distribution induced by innovative activities. This is therefore a misspecified-Bayesian-learning environment. Appendix \ref{app:infinite} contains all the technical details that formalize the infinite-horizon environment.

\medskip
\noindent\textbf{Innovation frequency.}
Let
$\alpha_T:=\frac{1}{T}\sum_{t=1}^{T}\1\{a_t=2\}$
denote the realized innovation share up to horizon $T\geq1$. This object will play a central role in this section because it summarizes how frequently the agent chooses to innovate along an equilibrium path.

\medskip
\noindent\textbf{Robustness thresholds.}
 Recall the threshold $\lambda^*$ in Theorem \ref{thm:dynamic:trap-scaleup} and rename it as follows
\begin{align}
\bar\lambda^H_k:=\frac{1}{k}\log \frac{1-p}{1-\theta_H} \quad\text{and}\quad
\bar\lambda^L_k:=\frac{1}{k}\log \frac{p}{\theta_L},
\label{eq:cycles:lambda-bars-main}
\end{align}
which correspond to the maximal robustness levels at which innovation remains contractible when the
agent considers $q^2_H$ or $q^2_L$, respectively. Now, let ${Q}^{2}_*\subseteq Q^{2}$ denote the set of structured models that best approximate
the true innovation outcome distribution, i.e., ${Q}^{2}_*:=\text{arg min}_{q\in Q^{2}}\DKL(p_\ast^2\|q)$. Then, the following constant will play a key role:
\begin{align}
\alpha_\ast
:=
\frac{\gamma}{D_\ast^{2}}\hspace{0.03in} \max_{q\in {Q}^{2}_*}\bar\lambda_k(q),
\label{eq:cycles:alpha-bar}
\end{align}
where $\bar\lambda_k(q^2_H)=\bar\lambda^H_k$ and $\bar\lambda_k(q^2_L)=\bar\lambda^L_k.$ We focus on environments where innovation is economically \textit{valuable} to an organization, in the sense that it is implemented infinitely often along an equilibrium path. We also focus on the generic case where there is a unique KL minimizer---${Q}^{2}_*$ is a singleton---to abstract from the knife-edge case where the two distinct models $\{q^2_H,q^2_L\}$ happen to have the same asymptotic statistical performance.\footnote{This assumption can be relaxed, but at the expense of requiring a more abstract condition that we call ``local uniform convergence'' in \eqref{eq:dynamic:infty:state-unif}, which is needed to apply Proposition \ref{thm:dynamic:infty:bridge} in Appendix \ref{subsec:dynamic:infty}.} 

\begin{theorem}\label{thm:cycles:frontier}
Suppose innovation occurs infinitely often along an equilibrium path, the induced equilibrium contract is bounded, and ${Q}^{2}_*$ is a singleton. Then, $(\alpha_T)_{T\ge 1}$ satisfies 
$$
\limsup_{T\to\infty}\hspace{0.03in}\alpha_T\le \alpha_\ast \qquad \text{\normalfont a.s.},
$$
where $\alpha_\ast$ is in \eqref{eq:cycles:alpha-bar}. Moreover, if $\alpha_\ast<1$, then along any equilibrium path on which innovation occurs infinitely often, the
routine action must also occur infinitely often a.s.
\end{theorem}

Theorem \ref{thm:cycles:frontier} is the long-run counterpart of Theorem \ref{thm:dynamic:trap-scaleup}. Recall that Theorem \ref{thm:dynamic:trap-scaleup} gives a local implementability condition: after a history, continued innovation is feasible only if the continuation state $(\mu_t,\lambda_t)$ keeps incentive capacity above cost $k$. Theorem \ref{thm:cycles:frontier} then aggregates this same local constraint along an infinite-horizon path: when innovation is implemented very frequently, the induced accumulation of misspecification concerns eventually pushes $(\mu_t,\lambda_t)$ beyond the range in which Theorem \ref{thm:dynamic:trap-scaleup}’s local inequality can be satisfied. The upper bound $\alpha_\ast$ in \eqref{eq:cycles:alpha-bar} therefore constitutes a
\emph{speed limit}: it is the maximum long-run innovation intensity consistent with keeping the agent's
endogenous misspecification concern low enough that continued innovation remains implementable.

Notice that $\alpha_\ast$ is composed of the three key forces from our framework. (1) The numerator $\gamma$ captures the \emph{agent side}:
how strongly statistical evidence is translated into misspecification concerns. (2) The denominator $D_\ast^{2}$ captures the
\emph{roadmap side}: how well structured models can approximate the true innovation
outcome distribution. (3) $\bar\lambda^H_k$ and $\bar\lambda^L_k$ in (\ref{eq:cycles:lambda-bars-main}) capture the \emph{incentive side}:
how much leverage a contract can extract from outcome differences between innovation and routine action
before misspecification concerns dominate. Sustainable innovation intensity is therefore increasing in
$\gamma$ and in the roadmap's incentive leverage, and decreasing in $D_\ast^{2}$ and in 
innovation cost $k$. Thus, organizations can speed up innovation intensity using three distinct organizational levers: 
\begin{itemize}
    \item[(1)] hiring epistemically resilient agents (Corollary \ref{thm:dynamic:screening}), which would increase $\gamma$.
    \item[(2)] improving their roadmaps' fit  (e.g., tuning hyperparameters), which would reduce $D_\ast^{2}$. Observe that when the roadmap is correctly specified,  $D_\ast^{2}=0$, so $\alpha_\ast=+\infty$.
    \item[(3)] choosing projects with more distinguishable outcomes, which would increase $\bar \lambda(\cdot)$.
\end{itemize}

When $\alpha_\ast<1$, Theorem \ref{thm:cycles:frontier} shows that continued innovation cannot be implemented at full intensity. Since implementing innovation can degrade the agent's \textit{trust} in the roadmap, organizations must periodically revert to routine tasks to rebuild this trust and prevent misspecification concerns from reaching uncontractible levels. This tension arises because long-run misspecification concerns accumulate proportionally to the realized frequency of innovation (Lemma \ref{lem:cycles:llr-frequency}). Theorem \ref{thm:cycles:frontier} therefore isolates a new dynamic mechanism: when a sophisticated agent faces a high-uncertainty contracting environment, exploration-exploitation cycles emerge endogenously in the long run.\footnote{These cycles also differ from those in standard misspecified Bayesian learning which are belief-driven \citep[e.g.,][]{nyarko91,fuden17}. In our framework, even given a fixed belief, cycles will still arise because they are driven by the agent's endogenous misspecification concerns and not by his belief.} Importantly, these cycles are dictated by the internal tension of managing the agent's trust in the roadmap and its statistical performance, rather than by exogenous changes in technology or payoff. 

\subsection{Illustration: long-run innovation frontier}\label{subsec:cycles:frontier-example}

A simple way to visualize the mechanism underlying Theorem \ref{thm:cycles:frontier} is to convert its asymptotic bound into a one-dimensional innovation frontier. This exercise will capture the long-run analogue of Figure \ref{fig:dynamic:capacity}. In Figure \ref{fig:dynamic:capacity}, the relevant object is the post-success continuation wedge as a function of the promised utility spread. Here, the relevant object is the \emph{asymptotic local incentive capacity} as a function of the long-run innovation frequency.

\par Let the true innovation outcome (Bernoulli) distribution have success probability $\theta_\ast\in(0,1)$ and
$\DKL(\theta_\ast\|\theta_H)<\DKL(\theta_\ast\|\theta_L),$
so that the unique KL minimizer is $q^2_H$. Then, $Q_*^{2}=\{q^2_H\}$ and $D^2_\ast=\DKL(\theta_\ast\|\theta_H)
=
\theta_\ast\log\frac{\theta_\ast}{\theta_H}
+
(1-\theta_\ast)\log\frac{1-\theta_\ast}{1-\theta_H}.$ Now, suppose a realized path satisfies $\alpha_T\to\alpha\in(0,1]$. Under the average LLR rule \eqref{eq:dynamic:infty:LLR2}, the agent's misspecification concern converges to $\lambda_\infty(\alpha)=\frac{\alpha D^2_\ast}{\gamma}$ (Lemma \ref{lem:cycles:llr-frequency}).
Since Bayesian posteriors asymptotically concentrate on KL minimizers under misspecification (Lemma \ref{lem:cycles:posterior}), the posterior $\mu_t$ will concentrate on $q^2_H$ as $t\rightarrow\infty$ and the relevant incentive capacity becomes
$$
C_\infty(\alpha)
:=
C(\delta_{\{q^2_H\}},\lambda_\infty(\alpha))
=
\frac{1}{\lambda_\infty(\alpha)}
\log\frac{1-p}{1-\theta_H}
=
\frac{\gamma}{\alpha D^2_\ast}
\log\frac{1-p}{1-\theta_H}.
$$
Using $\bar\lambda_k^H=\frac{1}{k}\log\frac{1-p}{1-\theta_H}$ from \eqref{eq:cycles:lambda-bars-main}, the above can be rewritten as
$C_\infty(\alpha)=k\frac{\alpha_\ast}{\alpha},$ where $\alpha_\ast=\frac{\gamma}{D^2_\ast}\bar\lambda_k^H.$
Thus, the long-run implementability condition $C_\infty(\alpha)\ge k$ is equivalent to
$\alpha\le \alpha_\ast.$
This is the speed limit in Theorem \ref{thm:cycles:frontier}: the more frequently the organization innovates, the more misspecification concern accumulates, and the less continuation incentive capacity remains available at future innovation nodes. For a numerical illustration, consider
$$
p=0.10,\qquad \theta_H=0.50,\qquad \theta_L=0.02,\qquad \theta_\ast=0.19,\qquad \gamma=0.10,\qquad k=0.40
$$
and evaluate all the expressions: $D^2_\ast
=
\DKL(0.19\|0.50)
\approx 0.21$
and
$\bar\lambda_k^H\approx 1.47$, so $\alpha_\ast\approx 0.71$ and
$C_\infty(\alpha)\approx \frac{0.28}{\alpha}.$
Notice, for example, that full-intensity innovation is not sustainable in this setting because
$C_\infty(1)\approx 0.28<0.40=k,$
whereas a lower innovation frequency such as $\alpha=0.50$ remains sustainable because
$C_\infty(0.50)\approx 0.56>0.40=k.$

\begin{figure}[hbt!]
\centering
\begin{tikzpicture}
\begin{axis}[
width=0.90\textwidth,
height=0.5\textwidth,
xmin=0.20, xmax=1.00,
ymin=0.00, ymax=1.50,
xlabel={Long-run innovation frequency $\alpha$},
ylabel={Asymptotic local incentive capacity},
axis lines=left,
samples=250,
domain=0.20:1.00,
scaled y ticks=false,
legend style={
font=\footnotesize,
draw=none,
fill=none,
at={(0.98,0.98)},
anchor=north east,
row sep=1pt,
/tikz/every even column/.append style={column sep=3pt}
},
legend cell align=left,
label style={font=\small},
tick label style={font=\scriptsize},
xlabel style={yshift=0.01em},
xtick={0.2,0.4,0.6,0.7101472662,0.8,1.0},
xticklabels={0.2,0.4,0.6,$\alpha_\ast$,0.8,1},
ytick={0,0.2,0.6,0.8,1.0,1.2,1.4}, 
yticklabel style={
font=\scriptsize,
/pgf/number format/fixed,
/pgf/number format/precision=2
},
extra y ticks={0.40},
extra y tick labels={$k$},
]
\addplot[name path=frontier, very thick, BrickRed]
{0.2840589065/x};
\addlegendentry{$C_\infty(\alpha)=k\alpha_\ast/\alpha$}

\addplot[name path=kline, black, dashed, thick, forget plot] {0.40};

\addplot[black, densely dotted, thick, forget plot]
coordinates {(0.7101472662,0) (0.7101472662,0.40)};

\addplot[forget plot, ForestGreen!18, draw=none]
fill between[of=frontier and kline, soft clip={domain=0.20:0.7101472662}];
\addlegendimage{area legend, draw=ForestGreen!50!black, fill=ForestGreen!18}
\addlegendentry{long-run incentive surplus}

\addplot[forget plot, BrickRed!18, draw=none]
fill between[of=kline and frontier, soft clip={domain=0.7101472662:1.00}];
\addlegendimage{area legend, draw=BrickRed!60!black, fill=BrickRed!18}
\addlegendentry{long-run incentive deficit}
\end{axis}
\end{tikzpicture}
\caption{\footnotesize The frontier $C_\infty(\alpha)=k\alpha_\ast/\alpha$ intersects the cost threshold $k$ uniquely at
$\alpha_\ast\approx 0.71$. 
}
\label{fig:cycles:frontier-example}
\end{figure}

\par Figure \ref{fig:cycles:frontier-example} plots the frontier $C_\infty(\alpha)$ as a function of the innovation frequency $\alpha$. The green region is the asymptotic incentive \textit{slack} associated with long-run innovation frequencies ($\alpha \le \alpha_\ast$) that are sustainable---their induced misspecification concern stays low enough for innovation to remain locally implementable. The red region is the asymptotic incentive \textit{shortfall} associated with unsustainable innovation frequencies ($\alpha > \alpha_\ast$). That is, if an organization attempts to innovate at these higher frequencies, the accumulated misspecification evidence will eventually drive the local incentive capacity, $C_\infty(\alpha)$, strictly below cost $k$. To summarize, this example has illustrated why when $\alpha_\ast<1$, organizations in high-uncertainty domains cannot innovate all the time: they must intersperse routine phases often enough to keep the average innovation intensity below the frontier.

\section{Concluding Remarks}\label{sec:discussion}

\subsection{Related Literature}\label{sec:lit}
In classical moral hazard, the principal and the agent are each aware of the true distribution of output associated with each level of effort \citep[e.g.,][]{holmstrom79,gh83}. It is now recognized that this is a strong assumption, which leads to the question of how it can be relaxed. 
As \citet[][p. 1429]{miao16} note: ``Contracting problems involve at least two parties. Introducing ambiguity or robustness into such problems must consider which party faces ambiguity and what it is ambiguous about.'' Some papers focus on the case where the principal faces ambiguity, not the agent; see \citet[][]{carroll19} for a survey. Here are two papers in this literature. \citet{carroll15} considers a maxmin principal who is averse to the ambiguity about the available actions, whereas the agent faces no ambiguity. \citet{miao16} study a continuous-time framework where the principal has a multiplier preference because she is concerned about misspecification, whereas the agent trusts it. The main justification for this approach is that the principal contracts with an agent who ``owns'' a relevant technology.
\par Notice that the approach above is the opposite of our framework, where we assume the agent faces ambiguity, \textit{not} the principal. Our paper therefore fits in the literature pioneered by \citet{amb11}, where a justification is that the principal ``owns'' a relevant technology and contracts with a less informed agent \citep[e.g.,][]{amb11,kellner15,kellner17,amb24}. Unlike \citet{amb11} who assume the agent has an incomplete preference, our agent is more sophisticated in the sense that he can disentangle his attitude toward prior beliefs and models, which then allows him to learn about model likelihoods and model misspecification separately.  Thus, our framework is the first to introduce both concerns for misspecification and learning dynamics on the agent side of moral hazard. As noted throughout, we have in mind settings where the principal aims to motivate the agent to innovate in high-uncertainty environments. On one hand, the principal can be viewed as the designer or owner of the roadmap, so she likely has superior information about it, hence why she trusts it. In practice,  managers of large labs would trust their roadmaps because they oversee many scientists engaged in diverse innovative activities, giving them an aggregate sense of the uncertainty involved across different tasks. On the other hand, each scientist bears, individually, the epistemic uncertainty present within their respective tasks.  This form of asymmetric information between principal and agent is close in spirit to the ``informed-principal'' literature pioneered by \citet{maskin90,maskin92}, where the principal is allowed to have superior information about the technology or productivity environment faced by the agent.
\par In parallel to the two literatures above, there is a ``middle ground'' pioneered by \citet{ghir94}, who allows both the principal and agent to face ambiguity with Choquet preferences \citep{schmeidler89}. Future research may extend our framework by allowing the principal to have ambiguity-misspecification concerns, which would unify the three literatures above. However, allowing the principal to  have ambiguity-misspecification concerns would not overturn the breakthrough-trap mechanism. The reason is that Theorem \ref{thm:dynamic:trap-scaleup} is an implementability result on the agent's post-success incentive set: once the agent's post-success incentive capacity falls below innovation cost, no  contract can induce continued innovation. This feasibility failure is pinned down by the agent's criterion and updating rule, not by the principal's objective function. Thus, replacing the principal's preference by any ambiguity-misspecification-sensitive preference would change only the principal's ranking of feasible contracts, not restore feasibility when the agent's post-success incentive set is empty. If anything, additional principal-side pessimism would reinforce the breakthrough-trap mechanism by shrinking the set of available contracts.

\subsection{Conclusion}
We analyzed an agency problem where a sophisticated agent learns the value of innovation while being concerned that  models in the organization's ``roadmap'' are misspecified. We identified a ``breakthrough trap,'' where early success can trigger a loss of trust in the roadmap, to the point that it renders continued innovation uncontractible. Thus, long-run innovation intensity is subject to an endogenous ``speed limit,'' which may force organizations into ``exploration-exploitation'' cycles to manage both the agent's effort and his trust in the roadmap. We then showed that tolerance for early failure, screening at hiring, and post-success turnover can help mitigate breakthrough traps. 

\par Appendix \ref{sec:dynamic:optimal} derives a recursive  ARC and characterizes dynamic implementation. Appendix \ref{app:proofs} contains proofs of main results. Appendices \ref{sec:general}--\ref{app:dynamic:extensions} contain further extensions.

\appendix

\section{Appendix: Dynamic Analysis}\label{sec:dynamic:optimal}
We formalize the recursive representation of ARC in Appendix \ref{subsec:dynamic:mmr} and characterize dynamic implementation in Appendices \ref{subsec:dynamic:policies}--\ref{subsec:dynamic:optimal}. To align more closely with \citeauthor{mmr06_JET}'s (\citeyear{mmr06_JET}) recursive representations, let the agent's discount rate be denoted more generally as a constant $\beta\in(0,1]$. All proofs from this appendix are in Appendix \ref{sec:appproof}.

\subsection{Recursive ARC}\label{subsec:dynamic:mmr}

Following \citet{mmr06_JET}, we now characterize a recursive formulation of ARC that is appropriate for our innovation environment. Notably, this is new in the decision-theory literature because \citet{lanzani2025} assumes myopia. Recall that in Section \ref{subsec:dynamic:env}, the prior has full support on $Q^{2}$, and all models in $Q^{2}$ and $q^1$ have full support on $Y$.

\medskip
\noindent\textbf{Event tree and information.}
Fix the two-period tree with outcome paths $\omega=(y_1,y_2)\in\Omega:=Y\times Y$.\footnote{Relative to \citeauthor{mmr06_JET}'s (\citeyear{mmr06_JET}) node-payoff formulation, our $x_t$ is an edge payoff paid after $y_t$ is realized. Equivalently, one can relabel dates so that $x_t$ is the payoff at the successor information set; the recursive representation is unchanged because $x_t$ enters the one-step-ahead continuation vector.}
Let $(\mathcal F_t)_{t=1}^3$ be the filtration generated by realized histories:
$\mathcal F_1=\{\varnothing,\Omega\}$ (before observing $y_1$),
$\mathcal F_2=\sigma(y_1)$ (after observing $y_1$),
and $\mathcal F_3=\sigma(y_1,y_2)$ (after observing $y_2$). Formally, let $\kappa_t$ denote the one-step KL penalty defined as follows:
\begin{itemize}[]
\item If $a=1$, then for any $r\in\Delta(Y)$ and every $t$, set
$\kappa_t(h_t,1;r):=\frac{1}{\lambda_t(h_t)}\hspace{0.02in}\DKL(r\|q^1).$
\item If $a=2$, represent distortions as a joint distribution $\pi\in\Delta(Q^{2}\times  Y)$ with a benchmark
$\pi^0_t(q,y):=\mu_t(q|h_t)\hspace{0.02in}q(y)$ for all $t$, and set
\begin{align*}
\kappa_t(h_t,2;\pi):=
\begin{cases}
\displaystyle \frac{1}{\lambda_t(h_t)}\hspace{0.02in}\DKL(\pi\|\pi^0_t), & \text{if }\pi_{Q^{2}}=\mu_t(\cdot|h_t),\\
+\infty, & \text{otherwise,}
\end{cases}
\end{align*}
for every $t$, where $\pi_{Q^{2}}$ is the marginal of $\pi$ on $Q^{2}$.
\end{itemize} 

At $t=2$, the state depends on the first-period action. To ease notation, we write $\mu_2(\cdot|y_1)$, $\lambda_2(y_1)$, and $V_2(y_1;x_2)$ below to refer to the post-innovation history $(a_1=2,y_1)$. For the first-period comparison only, let $\widehat{\mathcal V}_2(y_1,a;z)$ denote the same expression as $\mathcal V_2(y_1,a;z)$, but evaluated at the unchanged state $(\mu_1,\lambda_1)$, as induced by first-period routine play. The agent's continuation values are then defined by backward induction:
\begin{align}
V_2(y_1;x_2) &:= \max_{a_2\in\{1,2\}}
 \mathcal V_2 \big(y_1,a_2; x_2(y_1,\cdot)\big),\label{eq:dynamic:V2}\\
\widehat V_2(y_1;x_2) &:= \max_{a_2\in\{1,2\}}
 \widehat{\mathcal V}_2 \big(y_1,a_2; x_2(y_1,\cdot)\big),\nonumber\\
V_1(x_1,x_2) &:=\max\left\{
 \mathcal V_1 \big(\varnothing,1; x_1(\cdot)+ \beta \widehat V_2(\cdot;x_2)\big),
 \mathcal V_1 \big(\varnothing,2; x_1(\cdot)+ \beta V_2(\cdot;x_2)\big)
 \right\}.\label{eq:dynamic:V1}
\end{align}
where for each $t,h_t,a$ and continuation utility vector $z:Y\to\R$,
\begin{align}
\mathcal V_t(h_t,1;z)&:=\min_{r\in\Delta(Y)}
\Big\{\sum_{y\in Y}r(y)z(y)+\kappa_t(h_t,1;r)\Big\},\label{eq:dynamic:V-routine}\\
\mathcal V_t(h_t,2;z)&:=\min_{\pi\in\Delta(Q^{2}\times Y)}
\Big\{\sum_{(q,y)\in Q^{2}\times Y}\pi(q,y)z(y)+\kappa_t(h_t,2;\pi)\Big\}-k.\label{eq:dynamic:V-innov}
\end{align}

\medskip
\noindent\textbf{Local misspecification.}
This representation encodes the behavioral restrictions on what the scientist distrusts that are necessary under ARC. For the routine action, $q^1$ is the benchmark scenario generated by an established protocol, and the distortion $r$ captures robustness concerns around that familiar scenario. For innovation, the roadmap is instead a collection of scenarios $Q^2$ together with posterior likelihoods $\mu_t(\cdot|h_t)$. In the large-lab interpretation, these likelihoods are informed by the organization's historical data, experimental protocols, and simulation tools, so a scientist may regard $\mu_t(\cdot|h_t)$ as a calibrated aggregate summary of which roadmap scenario best fits the evidence observed so far. This is why our recursive formulation does not allow Nature to distort the marginal on $Q^2$: this is operationalized as  $\kappa_t(h_t,2;\pi)=+\infty$ unless $\pi_{Q^2}=\mu_t(\cdot|h_t)$.

At the same time, the scientist may still doubt the predictions made conditional on any particular scenario $q\in Q^2$. In materials-science experiments, a scenario is only a simplified representation of the underlying mechanism; it may omit impurities, synthesis conditions, or other features of the experiment. Thus, even conditional on $q$, the scientist may fear that the outcome distribution is a nearby distribution rather than exactly $q$. This is why Nature may distort the $Y$-conditional distribution within each scenario, with misspecification intensity $\lambda_t(h_t)>0$.\footnote{When $\lambda=0$, define $\frac{1}{\lambda}\DKL(\cdot\|f):=\mathbb{I}_{f}(\cdot)$, where $\mathbb{I}_{f}(g)=0$ if $g=f$; $\mathbb{I}_{f}(g)=+\infty$ if $g\neq f$.
The recursive construction still works because the resulting ambiguity index is
grounded, closed, and convex.} Thus, misspecification concern is assumed to be ``local'' in the following sense: the scientist trusts the lab's posterior over roadmap scenarios, but worries that each scenario is only an approximation of the real experiment. At the end of this analysis, we show that this assumption is a defining feature of ARC. 

\medskip
\noindent\textbf{KL decomposition.}
Fix $h_t$ and write $\mu:=\mu_t(\cdot|h_t)$ and $\pi^0(q,y):=\mu(q)\hspace{0.02in}q(y)$.
If a joint distribution $\pi\in\Delta(Q^{2}\times Y)$ satisfies $\pi_{Q^{2}}=\mu$, then the KL cost in \eqref{eq:dynamic:V-innov}
decomposes as 
\begin{align}\label{eq:dynamic:KL-decomp}
\DKL(\pi\|\pi^0)=\sum_{q\in Q^{2}} \mu(q)\hspace{0.02in}\DKL \big(\pi(\cdot|q)\|q(\cdot)\big),
\end{align}
where $\pi(\cdot|q)$ denotes the $Y$-conditional of $\pi$ given $q$.
 \eqref{eq:dynamic:KL-decomp} is the formal statement that the misspecification penalty is local:
it is an average of within-model likelihood distortions.

\medskip
\noindent\textbf{One-period-ahead ambiguity indices.}
To connect more closely to \citet{mmr06_JET}, it is convenient to express the one-step problems
\eqref{eq:dynamic:V-routine}--\eqref{eq:dynamic:V-innov} as minimizations over the
one-period-ahead marginals on the next observed outcome $y\in Y$. For $t\in\{1,2\}$, history $h_t$, and $r\in\Delta(Y)$, define the ``one-period-ahead
ambiguity index'' $\tau_t$ as follows:
\begin{itemize}[leftmargin=18pt]
\item If $a=1$, set
$\tau_t(h_t,1;r)\hspace{0.02in}:=\hspace{0.02in}\kappa_t(h_t,1;r)\hspace{0.02in}=\hspace{0.02in}\frac{1}{\lambda_t(h_t)}\hspace{0.02in}\DKL(r\|q^1).$
\item If $a=2$, set
$\tau_t(h_t,2;r)\hspace{0.02in}:=\hspace{0.02in}\inf\big\{\kappa_t(h_t,2;\pi): \pi\in\Delta(Q^{2}\times Y), \pi_Y=r\big\},$
where $\pi_Y$ is the marginal of $\pi$ on $Y$ (and feasibility already enforces $\pi_{Q^{2}}=\mu_t(\cdot|h_t)$).
\end{itemize}
Since $z(\cdot)$ depends only on the realized $y$, 
\eqref{eq:dynamic:V-routine}--\eqref{eq:dynamic:V-innov} can be equivalently written as
\begin{align}
\mathcal V_t(h_t,1;z)
&=\min_{r\in\Delta(Y)}\Big\{\sum_{y\in Y} r(y)\hspace{0.02in}z(y)\hspace{0.02in}+\hspace{0.02in}\tau_t(h_t,1;r)\Big\},\label{eq:dynamic:mmr-one-step-routine}\\
\mathcal V_t(h_t,2;z)
&=\min_{r\in\Delta(Y)}\Big\{\sum_{y\in Y} r(y)\hspace{0.02in}z(y)\hspace{0.02in}+\hspace{0.02in}\tau_t(h_t,2;r)\Big\}-k.\label{eq:dynamic:mmr-one-step-innov}
\end{align}
\begin{definition}[Dynamic action plan]\label{def:dynamic:plan}
An \emph{action plan} is $\psi=(a_1,\sigma_2)$ where $a_1\in A$ and
$\sigma_2:Y\to A$ is a continuation policy. We say that a contract $(x_1,x_2)$
\emph{subgame-perfectly implements} $\psi$ if there exists a subgame-perfect equilibrium in
which the agent chooses $a_1$ at $t=1$ and, after each history $y_1\in Y$, chooses
$a_2=\sigma_2(y_1)$ at $t=2$.
\end{definition}
After a plan $\psi=(a_1,\sigma_2)$ is fixed, $\tau_2(y_1,\sigma_2(y_1);\cdot)$ is understood as the index evaluated at the period-2 state induced by $a_1$: if $a_1=2$, it is the post-innovation state $(\mu_2(\cdot|y_1),\lambda_2(y_1))$, while if $a_1=1$, it is the unchanged state $(\mu_1,\lambda_1)$. Thus, $\{\tau_t(h_t,\psi_t(h_t);\cdot)\}^2_{t=1}$ is the ``one-period-ahead ambiguity index'' in \citet[][Proposition 2]{mmr06_JET}.

\medskip
\noindent\textbf{Induced two-period dynamic ambiguity index.}
Fix any history-dependent action plan $\psi=(a_1,\sigma_2)$. We
now specialize the recursive variational representation in \citet[][Theorem 1, eqs. (10)--(12),
and Theorem 2]{mmr06_JET} to our two-period tree
$\Omega$ defined above. For a payoff vector $(x_1,x_2)$, define the plan-induced payoffs
$$
u_1^\psi(y_1):=x_1(y_1)-k\1\{a_1=2\},
\qquad
u_2^\psi(y_1,y_2):=x_2(y_1,y_2)-k\1\{\sigma_2(y_1)=2\}.
$$
For any $p\in\Delta(\Omega)$, define
$r_1(y_1):=\sum_{y_2\in Y}p(y_1,y_2),$ $
r_2(y_1)(y_2):=\frac{p(y_1,y_2)}{r_1(y_1)}$ if $r_1(y_1)>0$, with $r_2(y_1)$ arbitrary when $r_1(y_1)=0$. We now derive the terminal, period-2, and period-1 ambiguity indices. At the terminal date, after the full outcome path $\omega=(y_1,y_2)$ is observed,
there is no remaining uncertainty, so the terminal ambiguity index $c_3^\psi$ is
\begin{align}
c_3^\psi(\omega,p)
:=
\begin{cases}
0, & \text{if }p=\delta_{\{\omega\}},\\
+\infty, & \text{otherwise,}
\end{cases}
\qquad \omega\in\Omega,\quad p\in\Delta(\Omega),
\label{eq:dynamic:c3}
\end{align}
where $\delta_{\{\omega\}}$ is the Dirac measure on $\omega$. At a period-2 history $y_1$, the
continuation distribution is identified with $r_2\in\Delta(Y)$, and the period-2 ambiguity index $c_2^\psi$ is
\begin{align}
c_2^\psi(y_1,r_2)
&:=
\beta\sum_{y_2\in Y}r_2(y_2)c_3^\psi\big((y_1,y_2),d_{(y_1,y_2)}\big)
+\tau_2\big(y_1,\sigma_2(y_1);r_2\big) \nonumber\\
&=
\tau_2\big(y_1,\sigma_2(y_1);r_2\big).
\label{eq:dynamic:c2}
\end{align}
At the initial node, the period-1 ambiguity index $c_1^\psi$ over two-period distributions is
\begin{align}
c_1^\psi(p)
&:=
\beta
\sum_{\substack{y_1\in Y:\\ r_1(y_1)>0}}
r_1(y_1)c_2^\psi\big(y_1,r_2(y_1)\big)
+\tau_1(\varnothing,a_1;r_1) \nonumber\\
&=
\beta
\sum_{\substack{y_1\in Y:\\ r_1(y_1)>0}}
r_1(y_1)\tau_2\big(y_1,\sigma_2(y_1);r_2(y_1)\big)
+\tau_1(\varnothing,a_1;r_1).
\label{eq:dynamic:c1}
\end{align}

Notice that \eqref{eq:dynamic:c2}--\eqref{eq:dynamic:c1} are the two-period ``no-gain
condition'' in \citet[][eq. (11)]{mmr06_JET}: the cost of a distribution over future histories is the
current one-period-ahead cost plus the discounted expected continuation cost. The dynamic variational representation, corresponding to \citet[][eq. (10)]{mmr06_JET},
becomes
\begin{align}
W_2^\psi(y_1;x_2)
&:=
\min_{r_2\in\Delta(Y)}
\left\{
\sum_{y_2\in Y}r_2(y_2)u_2^\psi(y_1,y_2)
+c_2^\psi(y_1,r_2)
\right\},
\label{eq:dynamic:W2-full}\\
W_1^\psi(x_1,x_2)
&:=
\min_{p\in\Delta(\Omega)}
\left\{
\sum_{(y_1,y_2)\in\Omega}p(y_1,y_2)
\big[u_1^\psi(y_1)+\beta u_2^\psi(y_1,y_2)\big]
+c_1^\psi(p)
\right\}.
\label{eq:dynamic:W1-full}
\end{align}
Using \eqref{eq:dynamic:c1}, the period-1 representation \eqref{eq:dynamic:W1-full} can be
written recursively as
\begin{align}
W_1^\psi(x_1,x_2)
=
\min_{r_1\in\Delta(Y)}
\left\{
\sum_{y_1\in Y}r_1(y_1)
\big[u_1^\psi(y_1)+\beta W_2^\psi(y_1;x_2)\big]
+\tau_1(\varnothing,a_1;r_1)
\right\},
\label{eq:dynamic:W1-recursive}
\end{align}
which is the two-period version of \citet[][eq. (12)]{mmr06_JET}. Using \eqref{eq:dynamic:W2-full} and \eqref{eq:dynamic:W1-recursive}, and then plugging $u_1^\psi$ and $u_2^\psi$ gives
$W_1^\psi(x_1,x_2)
=
\mathcal V_1\big(\varnothing,a_1;x_1(\cdot)+\beta W_2^\psi(\cdot;x_2)\big)$ and
\begin{align*}
W_2^\psi(y_1;x_2)=\begin{cases}
    \mathcal V_2\big(y_1,\sigma_2(y_1);x_2(y_1,\cdot)\big) &\text{if }a_1=2\\
    \widehat{\mathcal V}_2\big(y_1,\sigma_2(y_1);x_2(y_1,\cdot)\big)&\text{if }a_1=1.
\end{cases}
\end{align*}
Thus, the recursive variational representation generated by
$c_3^\psi,c_2^\psi,c_1^\psi$ coincides with the node-by-node continuation values used in
\eqref{eq:dynamic:V2}--\eqref{eq:dynamic:V1}. This is formalized in the next result.

\begin{proposition}\label{prop:dynamic:mmr-recursive}
Fix any history-dependent action plan $\psi=(a_1,\sigma_2)$.
Consider the induced preference over two-period payoff processes $(x_1,x_2)$ obtained by evaluating continuation payoff vectors
node-by-node using: $\mathcal V_1(\varnothing,a_1;\cdot)$ at the initial node and, at each period-2 history $y_1$, using
$\mathcal V_2(y_1,\sigma_2(y_1);\cdot)$ if $a_1=2$ and
$\widehat{\mathcal V}_2(y_1,\sigma_2(y_1);\cdot)$ if $a_1=1$.
Then, this plan-induced preference is a recursive variational preference in the sense of \citet[][Theorem 2]{mmr06_JET},
with one-period-ahead ambiguity indices $\{\tau_t(h_t,\psi_t(h_t);\cdot)\}_{t=1}^2$ given in
\eqref{eq:dynamic:mmr-one-step-routine}--\eqref{eq:dynamic:mmr-one-step-innov}.
In particular, it is dynamically consistent in the sense of \citet[][Theorem 1]{mmr06_JET}.
Consequently, the optimal continuation values in \eqref{eq:dynamic:V2}--\eqref{eq:dynamic:V1} are obtained by backward induction
from these dynamically consistent conditional preferences.
\end{proposition}

The local-misspecification restriction $\pi_{Q^{2}}=\mu_t(\cdot|h_t)$ is needed for Proposition
\ref{prop:dynamic:mmr-recursive}. It ensures that the one-step innovation value is
the ARC value in \citet[][eq. (35)]{hansenmiss25}. To see this, fix $h_t$ and let
$\mu:=\mu_t(\cdot|h_t)$ and $\pi^0(q,y):=\mu(q)q(y)$, then
\begin{align*}
\min_{\substack{\pi\in\Delta(Q^{2}\times Y):\\ \pi_{Q^{2}}=\mu}}
\Big\{
\sum_{(q,y)\in Q^{2}\times Y}\pi(q,y)z(y)
+\frac{1}{\lambda_t(h_t)}\DKL(\pi\|\pi^0)
\Big\}
=
\sum_{q\in Q^{2}}\mu(q)g_q(z;\lambda_t(h_t))
=
\mathcal G(z;\mu,\lambda_t(h_t)),
\end{align*}
which is $\mathcal G$ in \eqref{eq:dynamic:G}. The equality follows from the KL decomposition in \eqref{eq:dynamic:KL-decomp} and the
pointwise minimization over the $Y$-conditional distribution within each model $q$. If Nature were allowed to distort the marginal on $Q^{2}$, then the
one-step value would become
\begin{align*}
\min_{\pi\in\Delta(Q^{2}\times Y)}
\Big\{
\sum_{(q,y)\in Q^{2}\times Y}\pi(q,y)z(y)
+\frac{1}{\lambda_t(h_t)}\DKL(\pi\|\pi^0)
\Big\}
&=
-\frac{1}{\lambda_t(h_t)}
\log\Big(
\sum_{q\in Q^{2}}\mu(q)\sum_{y\in Y}q(y)e^{-\lambda_t(h_t)z(y)}
\Big)\\
&=
g_{\bar q_\mu}(z;\lambda_t(h_t)),
\end{align*}
where $\bar q_\mu:=\sum_{q\in Q^{2}}\mu(q)q$. However, in general,
$\sum_{q\in Q^{2}}\mu(q)g_q(z;\lambda_t(h_t))
\neq
g_{\bar q_\mu}(z;\lambda_t(h_t)),$
because of log. Thus, allowing Nature to distort $\mu$ would collapse
the two-stage ARC evaluation into a single multiplier-like preference around the
mixture $\bar q_\mu$. Such a preference may be variational, but it is no longer represented by the ARC
criterion in \eqref{eq:dynamic:G}. Thus, the one-period-ahead ambiguity indices
and the no-gain recursion used in Proposition \ref{prop:dynamic:mmr-recursive} would
represent a different dynamic preference. Local misspecification is therefore the restriction
that keeps Bayesian updating over models separate from within-model
misspecification and makes the recursive variational representation coincide with ARC.

\subsection{Subgame-perfect implementation}\label{subsec:dynamic:policies}

Since the contract can condition on the intermediate outcome $y_1$, the $t=2$ decision problems
separate across histories. The relevant period-2 state, however, depends on the first-period action
induced by the plan. Fix an action plan $\psi=(a_1,\sigma_2)$ and define its plan-induced state after
$y_1$, $\big(\mu_2^\psi(\cdot|y_1),\lambda_2^\psi(y_1)\big)$, as follows
\begin{align*}
\big(\mu_2^\psi(\cdot|y_1),\lambda_2^\psi(y_1)\big):=
\begin{cases}
\big(\mu_2(\cdot|y_1),\lambda_2(y_1)\big), & \text{if }a_1=2,\\
(\mu_1,\lambda_1), & \text{if }a_1=1.
\end{cases}
\end{align*}
All statements below use the EU-limit convention when $\lambda_2^\psi(y_1)=0$.

\begin{proposition}\label{prop:dynamic:terminal-impl}
Fix a plan $\psi=(a_1,\sigma_2)$ and a dynamic contract $(x_1,x_2)$. At $t=2$, after
$y_1$, $a_2=2$ is weakly optimal at the plan-induced state
$\big(\mu_2^\psi(\cdot|y_1),\lambda_2^\psi(y_1)\big)$ if and only if
${M} \big(\mu_2^\psi(\cdot|y_1),\lambda_2^\psi(y_1); x_2(y_1,\cdot)\big)\ge k.$
Thus, conditional on $\psi$, $a_2=2$ is implementable after $y_1$ if and only if there exists a
finite vector $x:Y\to\R$ such that
${M}(\mu_2^\psi(\cdot|y_1),\lambda_2^\psi(y_1);x)\ge k.$
Equivalently, either
$C(\mu_2^\psi(\cdot|y_1),\lambda_2^\psi(y_1))>k,$
or
$C(\mu_2^\psi(\cdot|y_1),\lambda_2^\psi(y_1))=k$
and the supremum defining $C(\cdot,\cdot)$ in \eqref{eq:dynamic:Delta-C} is attained.
\end{proposition}

The next bound and limit are convenient sufficient conditions for node-by-node feasibility in
Proposition \ref{prop:dynamic:terminal-impl}.

\begin{obs}\label{prop:dynamic:capacity-lb}
Fix $\lambda>0$ and $\mu\in\Delta(Q^{2})$, where $\mu(q^2_H)=m$ and $\mu(q^2_L)=1-m$. Then,
$C(\mu,\lambda) \ge \frac{1}{\lambda}\big(m\log \frac{p}{\theta_H}+(1-m)\log \frac{p}{\theta_L}\big).$
\end{obs}
The next result shows that continued innovation is implementable after a failure. It constitutes the mirror image of Theorem \ref{thm:dynamic:trap-scaleup} that is illustrated in Figure \ref{fig:dynamic:capacity-failure}.
\begin{corollary}\label{cor:dynamic:failure-impl}
Let $(\mu_{\mathfrak f},\lambda_{\mathfrak f})
:=(\mu_2(\cdot|y_1=0),\lambda_2(0))$. Then,
$\lim_{\theta_L\downarrow0}C(\mu_{\mathfrak f},\lambda_{\mathfrak f})=+\infty.$
In particular, for any fixed $k>0$, there exists $\bar\theta_L\in(0,p)$ such that for every
$\theta_L\in(0,\bar\theta_L)$,
$C(\mu_{\mathfrak f},\lambda_{\mathfrak f})>k,$
so $a_2=2$ is implementable after $y_1=0$.
\end{corollary}

\subsection{First-period implementation}\label{subsec:dynamic:optimal}

Here, we complete the two-period implementation analysis by characterizing the agent's $t=1$ best response given an arbitrary continuation contract. 

\begin{proposition}\label{prop:dynamic:first-period-ic}
Fix a dynamic contract $(x_1,x_2)$ and the agent's discount rate $\beta\in(0,1]$. At $t=1$, $a_1=2$ is weakly optimal if and only if
\begin{align}\label{eq:dynamic:first-period-IC}
\mathcal G\big(x_1(\cdot)+\beta V_2(\cdot;x_2);\mu_1,\lambda_1\big)-g_{q^1}\big(x_1(\cdot)+\beta \widehat V_2(\cdot;x_2);\lambda_1\big)
\ge
k.
\end{align}
\end{proposition}

\section{Appendix: Proofs for Main Text}\label{app:proofs}

This appendix collects proofs of main results; see Appendix \ref{sec:appproof} for all other proofs.

\subsection{Proof of Lemma \ref{lem:dynamic:lambda-closed}}
\begin{proof} By Assumption \ref{ass:dynamic:llr},
$\lambda_2(y_1)=-\frac{1}{\gamma}\log\underset{q\in Q^{2}}{\text{\normalfont max}}\hspace{0.04in}q(y_1).$
If $y_1=1$, then $\underset{q\in Q^{2}}{\text{\normalfont max}}\hspace{0.04in}q(1)=\max\{\theta_L,\theta_H\}=\theta_H$, hence
$\lambda_2(1)=-\log(\theta_H)/\gamma$.
If $y_1=0$, then $\underset{q\in Q^{2}}{\text{\normalfont max}}\hspace{0.04in}q(0)=\max\{1-\theta_L,1-\theta_H\}=1-\theta_L$, hence
$\lambda_2(0)=-\log(1-\theta_L)/\gamma$.
\end{proof}

\subsection{Proof of Theorem \ref{thm:dynamic:trap-scaleup} and Corollary \ref{cor:dynamic:success-failure-reversal}}
\begin{lemma}[Translation invariance]\label{lem:translation}
For $\zeta\in\R$, define $T_\zeta(x):=x-\zeta\cdot\mathbf{1}$, where $\mathbf{1}$ is the vector of
ones in $\R^Y$. Then, $\forall x\in\R^Y,\mu\in\Delta(Q^{2}),\lambda>0$, $\mathcal{G}(T_\zeta(x);\mu,\lambda) = \mathcal{G}(x;\mu,\lambda) - \zeta.$
\end{lemma}

\begin{lemma}\label{lem:dynamic:two-point}
Fix $\lambda>0$ and Bernoulli distributions $q(\zeta),q(\vartheta)\in\Delta( Y)$ with $\zeta,\vartheta\in(0,1)$.
Then,
$\sup_{x: Y\to\R}\big(g_{q(\zeta)}(x;\lambda)-g_{q(\vartheta)}(x;\lambda)\big)
=\frac{1}{\lambda}\log \max\big\{\frac{\vartheta}{\zeta},\frac{1-\vartheta}{1-\zeta}\big\}.$
\end{lemma}
\begin{proof}[Proof of Lemma \ref{lem:dynamic:two-point}]
Let $q(\zeta)$ and $q(\vartheta)$ be Bernoulli distributions on $ Y=\{0,1\}$ with success probabilities
$\zeta,\vartheta\in(0,1)$, and fix $\lambda>0$.

\medskip
\noindent\textit{Step 1: normalization.}
For any $\eta\in\R$, translation invariance of $g_q$ implies
$g_q(x-\eta\mathbf 1;\lambda)=g_q(x;\lambda)-\eta$ (Lemma \ref{lem:translation}) and therefore
$g_{q(\zeta)}(x;\lambda)-g_{q(\vartheta)}(x;\lambda)
=
g_{q(\zeta)}(x-\eta\mathbf 1;\lambda)-g_{q(\vartheta)}(x-\eta\mathbf 1;\lambda)$.
Hence, without loss of generality, set $x(0)=0$ and write $x(1)=s\in\R$.

\medskip
\noindent\textit{Step 2: reparameterization.}
Let $r:=e^{-\lambda s}>0$. Then,
$g_{q(\zeta)}(x;\lambda)=-\frac{1}{\lambda}\log((1-\zeta)+\zeta r)$ and
$g_{q(\vartheta)}(x;\lambda)=-\frac{1}{\lambda}\log((1-\vartheta)+\vartheta r),$
so $g_{q(\zeta)}(x;\lambda)-g_{q(\vartheta)}(x;\lambda)
=
\frac{1}{\lambda}\log
\frac{(1-\vartheta)+\vartheta r}{(1-\zeta)+\zeta r}
=\frac{1}{\lambda}\log \phi(r),$
where $\phi(r):=\frac{(1-\vartheta)+\vartheta r}{(1-\zeta)+\zeta r}$.

\medskip
\noindent\textit{Step 3: maximization over $r>0$.}
The derivative is
$\phi'(r)=\frac{\vartheta(1-\zeta)-\zeta(1-\vartheta)}{((1-\zeta)+\zeta r)^2}
=\frac{\vartheta-\zeta}{((1-\zeta)+\zeta r)^2},$
so $\phi$ is monotone in $r$. Therefore, $\sup_{r>0}\phi(r)$ equals the larger of the two boundary limits:
$\lim_{r\downarrow 0}\phi(r)=\frac{1-\vartheta}{1-\zeta}$ and $\lim_{r\uparrow\infty}\phi(r)=\frac{\vartheta}{\zeta}.$
Taking the larger boundary value and applying $\log(\cdot)/\lambda$ yields
$\sup_{x: Y\to\R}\big(g_{q(\zeta)}(x;\lambda)-g_{q(\vartheta)}(x;\lambda)\big)
=
\frac{1}{\lambda}\log\max\big\{\frac{\vartheta}{\zeta},\frac{1-\vartheta}{1-\zeta}\big\}.$\end{proof}

\begin{proof}[Proof of Theorem \ref{thm:dynamic:trap-scaleup}] We prove parts (1) and (2) separately, starting with the former.

\medskip
\par\noindent\textit{--- Theorem \ref{thm:dynamic:trap-scaleup}.(1).} Let $\lambda_{\mathfrak{s}}(\cdot):=\lambda_2(1)=-\log(\theta_H)/\gamma$ (Lemma \ref{lem:dynamic:lambda-closed}) and $\mu_{\mathfrak{s}}(\cdot):=\mu_2(\cdot|y_1=1)$ with
$\mu_{\mathfrak{s}}(q^2_H)=m_{\mathfrak{s}}(\theta_L):=\frac{m\theta_H}{m\theta_H+(1-m)\theta_L}$ in \eqref{eq:dynamic:mu-update}.
Fix any utility vector $x: Y\to\R$. Then,
$$
{M}(\mu_{\mathfrak{s}},\lambda_{\mathfrak{s}};x)
=
m_{\mathfrak{s}}(\theta_L)\big(g_{q^2_H}(x;\lambda_{\mathfrak{s}})-g_{q^1}(x;\lambda_{\mathfrak{s}})\big)
+
(1-m_{\mathfrak{s}}(\theta_L))\big(g_{q^2_L}(x;\lambda_{\mathfrak{s}})-g_{q^1}(x;\lambda_{\mathfrak{s}})\big).
$$
Since $\theta_H>p>\theta_L$, Lemma \ref{lem:dynamic:two-point} yields $\forall x$:
$g_{q^2_H}(x;\lambda_{\mathfrak{s}})-g_{q^1}(x;\lambda_{\mathfrak{s}})
\le \frac{1}{\lambda_{\mathfrak{s}}}\log\max\big\{\frac{p}{\theta_H},\frac{1-p}{1-\theta_H}\big\}
=\frac{1}{\lambda_{\mathfrak{s}}}\log \frac{1-p}{1-\theta_H}$
 and 
$g_{q^2_L}(x;\lambda_{\mathfrak{s}})-g_{q^1}(x;\lambda_{\mathfrak{s}})
\le \frac{1}{\lambda_{\mathfrak{s}}}\log\max\big\{\frac{p}{\theta_L},\frac{1-p}{1-\theta_L}\big\}
=\frac{1}{\lambda_{\mathfrak{s}}}\log \frac{p}{\theta_L},$
where the equalities use $\theta_H>p$ (so $\frac{1-p}{1-\theta_H}>\frac{p}{\theta_H}$) and
$p>\theta_L$ (so $\frac{p}{\theta_L}>\frac{1-p}{1-\theta_L}$). Combining them yields:
\begin{align}
{M}(\mu_{\mathfrak{s}},\lambda_{\mathfrak{s}};x)
\le
\frac{1}{\lambda_{\mathfrak{s}}}\log \frac{1-p}{1-\theta_H}
+
\frac{1-m_{\mathfrak{s}}(\theta_L)}{\lambda_{\mathfrak{s}}}\log \frac{p}{\theta_L}
\quad\forall x.
\label{eq:succ:Delta-ub}
\end{align}

Now, note that $1-m_{\mathfrak{s}}(\theta_L)
=
\frac{(1-m)\theta_L}{m\theta_H+(1-m)\theta_L}
\le
\frac{1-m}{m\theta_H}\hspace{0.02in}\theta_L.$
Hence,
$0\le (1-m_{\mathfrak{s}}(\theta_L))\log \frac{p}{\theta_L}
\le \frac{1-m}{m\theta_H}\hspace{0.02in}\theta_L\log \frac{p}{\theta_L}
\xrightarrow[\theta_L\downarrow 0]{}0,$
since $\theta\log(1/\theta)\to 0$ as $\theta\downarrow 0$. Assume $\lambda_{\mathfrak{s}}>\lambda^*=\frac{1}{k}\log \frac{1-p}{1-\theta_H}$.
Equivalently,
$\frac{1}{\lambda_{\mathfrak{s}}}\log \frac{1-p}{1-\theta_H}<k.$
Let $\eta:=k-\frac{1}{\lambda_{\mathfrak{s}}}\log\big(\frac{1-p}{1-\theta_H}\big)>0$.
Choose $\bar\theta_L\in(0,p)$ such that for all $\theta_L\in(0,\bar\theta_L)$,
$\frac{1-m_{\mathfrak{s}}(\theta_L)}{\lambda_{\mathfrak{s}}}\log \frac{p}{\theta_L}<\frac{\eta}{2}.$
Then, \eqref{eq:succ:Delta-ub} implies ${M}(\mu_{\mathfrak{s}},\lambda_{\mathfrak{s}};x)\le k-\eta/2$ for all $x$, and
therefore
$C(\mu_{\mathfrak{s}},\lambda_{\mathfrak{s}})=\sup_x {M}(\mu_{\mathfrak{s}},\lambda_{\mathfrak{s}};x)\le k-\eta/2<k.$
Finally, if $C(\mu_{\mathfrak{s}},\lambda_{\mathfrak{s}})<k$, there does not exist any continuation utility schedule
$x_2(1,\cdot)$ with ${M}(\mu_{\mathfrak{s}},\lambda_{\mathfrak{s}};x_2(1,\cdot))\ge k$, so by
Proposition \ref{prop:dynamic:terminal-impl}, action $a_2=2$ is not implementable after history
$(a_1=2,y_1=1)$.

\par\noindent\textit{--- Theorem \ref{thm:dynamic:trap-scaleup}.(2).} Given $C(\mu_2(\cdot|y_1=1),\lambda_2(1))<k$, 
fix any dynamic contract $(x_1,x_2)$ and consider any induced SPNE.
On the history $(a_1=2,y_1=1)$, the state entering $t=2$ is $(\mu_2(\cdot|1),\lambda_2(1))$.
Since $C(\mu_2(\cdot|1),\lambda_2(1))<k$, by definition of incentive capacity (\ref{eq:dynamic:Delta-C}), we have
${M}(\mu_2(\cdot|1),\lambda_2(1);x)<k$ for every continuation utility vector $x$.
In particular, $M\big(\mu_2(\cdot|1),\lambda_2(1); x_2(1,\cdot)\big)<k.$
Applying Proposition \ref{prop:dynamic:terminal-impl} to any plan with $a_1=2$ implies that after the post-innovation history $(a_1=2,y_1=1)$, the agent strictly prefers
$a_2=1$ to $a_2=2$. Thus, in any SPNE, $\Proba(a_1=2, y_1=1, a_2=2)=0.$
Recall that the principal's period-2 scale-up utility is $R_2 = 1 + (\xi-1)\hspace{0.02in}\1\{a_1=2, y_1=1, a_2=2\}$ (see, \eqref{eq:multiplier}),
so $R_2=1$ almost surely under any SPNE whenever the above event has probability zero. Consequently, the principal's
expected profit is independent of $\xi$.\end{proof}

\begin{proof}[Proof of Corollary \ref{cor:dynamic:success-failure-reversal}]
Parts (ii) and (iii) follow from Theorem \ref{thm:dynamic:trap-scaleup} and Corollary \ref{cor:dynamic:failure-impl}, after shrinking $\bar\theta_L$ so that the post-failure capacity is strictly larger than $k$.

It remains to prove part (i). Choose a finite continuation vector $x_2(0,\cdot)$ that makes $a_2=2$ optimal after $y_1=0$; this exists by part (iii), since the post-failure capacity is strictly larger than $k$. Set $x_2(1,0)=x_2(1,1)=0$. After $y_1=1$, both actions then deliver the same continuation utility, but $a_2=2$ additionally incurs the cost $k>0$, so $a_2=1$ is strictly optimal. Thus, there is a finite $x_2$ such that $a_2=1$ is optimal after $y_1=1$ and $a_2=2$ is optimal after $y_1=0$. Let $V_2(\cdot;x_2)$ denote the period-2 value after first-period innovation and let $\widehat V_2(\cdot;x_2)$ denote the period-2 value after first-period routine play (see, \eqref{eq:dynamic:V2}--\eqref{eq:dynamic:V1}). Since $x_1$ is unrestricted in utility space, it is enough to find a finite vector $z:Y\to\mathbb R$ such that $\mathcal G(z;\mu_1,\lambda_1)-k\ge g_{q^1}(z-\beta(V_2-\widehat V_2);\lambda_1)$; then setting $x_1=z-\beta V_2(\cdot;x_2)$ makes first-period innovation optimal. 
\par Now, suppose $\lambda_1>0$. For $T>0$, set $z_T(0)=T$ and $z_T(1)=0$. As $T\to+\infty$, the routine value $g_{q^1}(z_T-\beta(V_2-\widehat V_2);\lambda_1)$ converges to a finite number because $V_2-\widehat V_2$ is finite (since $x_2$ is finite), while $\mathcal G(z_T;\mu_1,\lambda_1)$ converges to $-\lambda_1^{-1}\{m\log\theta_H+(1-m)\log\theta_L\}$, which diverges to $+\infty$ as $\theta_L\downarrow0$. Hence, after shrinking $\bar\theta_L$ if necessary, some finite $T$ satisfies the first-period incentive constraint.  Then, suppose $\lambda_1=0$. For all sufficiently small $\theta_L>0$, the initial mean success probability $m\theta_H+(1-m)\theta_L$ differs from $p$. If this mean is above $p$, choose $z_T(1)=T$ and $z_T(0)=0$; if it is below $p$, choose $z_T(0)=T$ and $z_T(1)=0$. Under the limiting EU criterion, $\mathcal G(z_T;\mu_1,0)-g_{q^1}(z_T-\beta(V_2-\widehat V_2);0)$ is affine in $T$ with strictly positive slope, so the first-period incentive constraint holds for $T$ large enough. Thus, first-period innovation is implementable.
\end{proof}

\begin{proof}[Proof of Corollary \ref{prop:dynamic:feedback-optimal}]
To avoid notational conflicts throughout this proof, write $g_H(\cdot;\lambda):=g_{q_H^2}(\cdot;\lambda)$, $g_L(\cdot;\lambda):=g_{q_L^2}(\cdot;\lambda)$, and $g_1(\cdot;\lambda):=g_{q^1}(\cdot;\lambda)$. Fix any binary policy $\Pi$ satisfying $\Pi(1|1)=1$. Then, $\Pi$ is determined by its false-positive rate $r:=\Pi(1|0)\in[0,1]$, so $\Pi=\Pi^r$. Under $\Pi^r$, the favorable-signal likelihoods are $q_H^r=\theta_H+(1-\theta_H)r$ and $q_L^r=\theta_L+(1-\theta_L)r$. Let $m^r:=m q_H^r/(m q_H^r+(1-m)q_L^r)$ denote the posterior weight on the high innovation model after $s_1=1$. Proposition \ref{prop:dynamic:terminal-impl} implies that continued innovation after $s_1=1$ is implementable with slack $\eta$ whenever $C(\mu_2^r(\cdot|s_1=1),\lambda_2^r(1))\ge k+\eta$.

We first show that $F_\eta$ is nonempty. When $r=1$, both favorable-signal likelihoods are equal to one, so the favorable evaluation is completely uninformative: $\mu_2^r(\cdot|s_1=1)=\mu_1$ and $\lambda_2^r(1)=0$. Under the hypotheses of Corollary \ref{cor:dynamic:success-failure-reversal}, the EU-limit calculation used there gives $C(\mu_1,0)=+\infty$. Thus, $1\in F_\eta$. It remains to show that $F_\eta$ is closed. Let $\{r_n\}_{n\ge1}\subseteq F_\eta$ and suppose $r_n\to r\in[0,1]$. If $r=1$, then $r\in F_\eta$ by the previous paragraph. Therefore, suppose $r<1$. Choose $\bar r\in(r,1)$. Since $r_n\to r$, there exists $N$ such that $r_n\in[0,\bar r]$ for every $n\ge N$; omitting the first $N-1$ terms preserves both convergence to $r$ and membership in $F_\eta$, so it is enough to work on $[0,\bar r]$. On this interval, $r\mapsto\lambda_2^r(1)=-\gamma^{-1}\log q_H^r$ is continuous and bounded away from zero because $q_H^r\le q_H^{\bar r}<1$ for every $r\in[0,\bar r]$. By translation invariance (Lemma \ref{lem:translation}), the supremum defining capacity can be normalized to payoff vectors with $x(0)=0$. Thus, writing $x(1)=\tau$, we have $C(\mu_2^r(\cdot|s_1=1),\lambda_2^r(1))=\sup_{\tau\in\R}\Phi(r,\tau)$, where $\Phi(r,\tau):=m^r g_H((0,\tau);\lambda_2^r(1))+(1-m^r)g_L((0,\tau);\lambda_2^r(1))-g_1((0,\tau);\lambda_2^r(1))$. The map $\Phi$ is continuous on $[0,\bar r]\times\R$. Moreover, because $\lambda_2^r(1)$ is bounded away from zero on $[0,\bar r]$, $\Phi$ extends continuously to $[0,\bar r]\times\overline{\R}$, where $\overline{\R}:=\R\cup\{-\infty,+\infty\}$, with endpoint values obtained by taking $\tau\to+\infty$ and $\tau\to-\infty$ in the preceding expression. Therefore, by the maximum theorem, the function $r\mapsto C(\mu_2^r(\cdot|s_1=1),\lambda_2^r(1))=\max_{\tau\in\overline{\R}}\Phi(r,\tau)$ is continuous on $[0,\bar r]$. Hence, $C(\mu_2^r(\cdot|s_1=1),\lambda_2^r(1))=\lim_{n\to\infty}C(\mu_2^{r_n}(\cdot|s_1=1),\lambda_2^{r_n}(1))\ge k+\eta$, so $r\in F_\eta$. As a result, $F_\eta$ is closed. Since $F_\eta$ is a nonempty closed subset of $[0,1]$, it is compact, so $r_\eta=\min F_\eta$ is well defined. 

Among no-false-negative binary policies, minimizing false positives is minimizing $r=\Pi(1|0)$. By definition of $r_\eta$, no smaller $r$ satisfies the capacity inequality. Since every binary policy is uniquely determined by its false-positive rate, $\Pi^{r_\eta}$ uniquely minimizes $\Pi(1|0)$ among all no-false-negative policies that preserve innovation with slack $\eta$.
\end{proof}

\subsection{Proof of Corollary \ref{thm:dynamic:screening}}\label{app:proof:thm-dynamic-screening}
\begin{lemma}\label{lem:F-lambda}
Fix $x\in\R^Y$. If $0<\lambda_1<\lambda_2$, then $\mathcal{G}(x;\mu,\lambda_1) \ge \mathcal{G}(x;\mu,\lambda_2).$
Moreover, if $x$ is non-constant and every $q\in Q^{2}$ has full support, the inequality is strict.
\end{lemma}

\begin{lemma}\label{lem:app:screening:decomposition}
Fix any posterior $\mu\in\Delta(Q^{2})$ and define the benchmark joint distribution
$\pi^0(q,y):=\mu(q) q(y)$ on $Q^{2}\times Y$. Fix any $\lambda\in(0,\infty)$ and any payoff 
$z:Y\to\R$. Then,
\begin{align}\label{eq:app:screening:decomposition}
\min_{\pi\in\Delta(Q^{2}\times Y):\ \pi_{Q^{2}}=\mu}
\Bigg\{\sum_{(q,y)\in Q^{2}\times Y}\pi(q,y) z(y)
+\frac{1}{\lambda} \DKL\!\big(\pi\big\|\pi^0\big)\Bigg\}
=\sum_{q\in Q^{2}}\mu(q) g_q(z;\lambda),
\end{align}
for $g_q$ in \eqref{eq:dynamic:gq} and the minimizer is
$\pi^\ast(q,y)=\mu(q)
\frac{q(y)\exp\{-\lambda z(y)\}}{\sum_{y'\in Y}q(y')\exp\{-\lambda z(y')\}}$ $\forall(q,y)\in Q^{2}\times Y.$
\end{lemma}

\begin{proof}[Proof of Corollary \ref{thm:dynamic:screening}]
Fix $\{E,I\}$ and $\gamma>0$. Throughout, $\mu_2(\cdot|y_1)$ is given by
\eqref{eq:dynamic:mu-update} and $\lambda_2(\cdot)$ evolves by Assumption \ref{ass:dynamic:llr} and
Lemma \ref{lem:dynamic:lambda-closed}. Recall $g_q$ and $\mathcal G$ from
\eqref{eq:dynamic:gq}--\eqref{eq:dynamic:G}.

\medskip
\noindent\emph{Step 1: Terminal reduction and monotonicity in $\gamma$.}
Fix any history $y_1\in Y$ and any continuation utility vector $z:Y\to\R$. Applying \citet[][Proposition 1.4.2]{dupuis97} to \eqref{eq:dynamic:V-routine} yields $\mathcal V_2(y_1,1;z)=g_{q^1}(z;\lambda_2(y_1)).$
Similarly, Lemma \ref{lem:app:screening:decomposition} (applied with $\mu=\mu_2(\cdot|y_1)$ and
$\lambda=\lambda_2(y_1)$) implies from \eqref{eq:dynamic:V-innov} that
$\mathcal V_2(y_1,2;z)=\sum_{q\in Q^{2}}\mu_2(q|y_1)\hspace{0.02in}g_q(z;\lambda_2(y_1)) - k
=\mathcal G\big(z;\mu_2(\cdot|y_1),\lambda_2(y_1)\big)-k.$
Therefore, by \eqref{eq:dynamic:V2},
\begin{align}\label{eq:app:screening:V2-reduction}
V_2(y_1;x_2)=
\max\Big\{
g_{q^1} \big(x_2(y_1,\cdot);\lambda_2(y_1)\big),\hspace{0.02in}
\mathcal G \big(x_2(y_1,\cdot);\mu_2(\cdot|y_1),\lambda_2(y_1)\big)-k
\Big\}.
\end{align}

Now specialize to contract $I$. Under the given hypothesis, the agent chooses $a_1=2$
under $I$ for every $\gamma$, so $\mu_2(\cdot|y_1)$ is independent of $\gamma$ and
$\lambda_2^\gamma(1)=\frac{-\log\theta_H}{\gamma}$, $\lambda_2^\gamma(0)=\frac{-\log(1-\theta_L)}{\gamma}$
by Lemma \ref{lem:dynamic:lambda-closed}. In particular, for each $y_1\in Y$,
$\gamma' >\gamma$ implies $\lambda_2^{\gamma'}(y_1)<\lambda_2^\gamma(y_1)$.

By Lemma \ref{lem:F-lambda}, for any $q$ and any fixed $x$,
$\lambda\mapsto g_q(x;\lambda)$ is weakly decreasing on $(0,\infty)$, and it is strictly
decreasing whenever $x$ is non-constant. Hence, for each $y_1$,
$\gamma' >\gamma\Longrightarrow
g_{q^1} \big(x_2^I(y_1,\cdot);\lambda_2^{\gamma'}(y_1)\big) \ge
g_{q^1} \big(x_2^I(y_1,\cdot);\lambda_2^{\gamma}(y_1)\big),$
and similarly for each $q\in Q^{2}$, $g_{q} (x_2^I(y_1,\cdot);\lambda_2^{\gamma'}(y_1)) \ge
g_{q} (x_2^I(y_1,\cdot);\lambda_2^{\gamma}(y_1)).$
Since $\mu_2(\cdot|y_1)$ has full support on $Q^{2}$, summing across $q$ yields
$\mathcal G \big(x_2^I(y_1,\cdot);\mu_2(\cdot|y_1),\lambda_2^{\gamma'}(y_1)\big) \ge\
\mathcal G \big(x_2^I(y_1,\cdot);\mu_2(\cdot|y_1),\lambda_2^{\gamma}(y_1)\big).$
If, in addition, $x_2^I(y_1,\cdot)$ is non-constant, then all of the above weak inequalities are
strict by the strict part of Lemma \ref{lem:F-lambda}. Applying these monotonicities to \eqref{eq:app:screening:V2-reduction} shows that for each $y_1$,
$V_2^I(y_1;\gamma)$ is weakly increasing in $\gamma$, and it is strictly increasing at
$y_1=1$ because \eqref{eq:dynamic:screening:nondeg} implies $x_2^I(1,\cdot)$ is non-constant and
therefore both terms inside the max in \eqref{eq:app:screening:V2-reduction} strictly increase in
$\gamma$. Formally, if $\gamma'>\gamma$, then both components inside the max are strictly larger at
$\gamma'$ than at $\gamma$, hence the max is strictly larger as well:
\begin{align}\label{eq:app:screening:V2-strict}
V_2^I(1;\gamma') > V_2^I(1;\gamma),
\qquad
V_2^I(0;\gamma') \ge V_2^I(0;\gamma).
\end{align}

\medskip
\noindent\emph{Step 2: $V_1^I(\gamma)$ is strictly increasing in $\gamma$.}
Define the first-period payoff vector induced by contract $I$:
$z_\gamma^I(y_1):=x_1^I(y_1)+\beta\hspace{0.02in}V_2^I(y_1;\gamma),$ for $y_1\in Y$. By \eqref{eq:app:screening:V2-strict} and $\beta>0$,
\begin{align}\label{eq:app:screening:z-mon}
z_{\gamma'}^I(1) > z_\gamma^I(1),
\qquad
z_{\gamma'}^I(0) \ge z_\gamma^I(0)
\quad\text{whenever }\gamma'>\gamma.
\end{align}

To compute the period-1 innovation evaluation, apply Lemma \ref{lem:app:screening:decomposition} to 
\eqref{eq:dynamic:V-innov} at $t=1$ when $\lambda_1>0$ (with $\mu=\mu_1$ and $\lambda=\lambda_1$, where $\pi^0_1(q,y)=\mu_1(q)q(y)$ and $\pi_{Q^{2}}=\mu_1$). The boundary case $\lambda_1=0$ follows from the EU-limit convention. Hence, for any $z:Y\to\R$,
\begin{align}\label{eq:app:screening:V1-innov-reduction}
\mathcal V_1(\varnothing,2;z)=\mathcal G(z;\mu_1,\lambda_1)-k.
\end{align}

We claim $\mathcal G(\cdot;\mu_1,\lambda_1)$ is strictly increasing in each coordinate. If $\lambda_1>0$, let $z,z':Y\to\R$ satisfy $z'(y)\ge z(y)$ for every $y$ and $z'(y_0)>z(y_0)$ for some $y_0\in Y$. Since every $q\in Q^{2}$ has full support,
$\sum_{y\in Y}q(y)e^{-\lambda_1 z'(y)}<\sum_{y\in Y}q(y)e^{-\lambda_1 z(y)}$ for every $q\in Q^{2}$. Applying $-(1/\lambda_1)\log(\cdot)$ gives $g_q(z';\lambda_1)>g_q(z;\lambda_1)$ for each $q$, so
$\mathcal G(z';\mu_1,\lambda_1)>\mathcal G(z;\mu_1,\lambda_1)$. If $\lambda_1=0$, then by the EU-limit convention, 
$\mathcal G(z;\mu_1,0)=\sum_{y\in Y}\bar q_1(y)z(y)$, where
$\bar q_1:=\sum_{q\in Q^{2}}\mu_1(q)q$ has full support. Thus,
$\mathcal G(z';\mu_1,0)>\mathcal G(z;\mu_1,0)$ as well.

\medskip
\noindent\emph{Step 3: $V_1^E(\gamma)$ is constant in $\gamma$.}
Under contract $E$, the hypothesis says that $a_1=1$ is chosen for every $\gamma>0$. Therefore, the relevant period-2 continuation value in the period-1 routine branch is the post-routine object
$\widehat V_2^E(\cdot)$ (see, Appendix \ref{subsec:dynamic:mmr}). By the definition of
$\widehat{\mathcal V}_2$, first-period routine play leaves the state equal to the initial state
$(\mu_1,\lambda_1)$. Hence, for each $y_1\in Y$,
$\widehat V_2^E(y_1)
=
\max\left\{
g_{q^1}\big(x_2^E(y_1,\cdot);\lambda_1\big),
\mathcal G\big(x_2^E(y_1,\cdot);\mu_1,\lambda_1\big)-k
\right\}.$
The contract $E$, the initial belief $\mu_1$, and the initial misspecification concern $\lambda_1$ are fixed as $\gamma$ varies. Thus, $\widehat V_2^E(y_1)$ is independent of $\gamma$ for every $y_1\in Y$. The period-1 payoff vector relevant for choosing the routine track is
$x_1^E(\cdot)+\beta\widehat V_2^E(\cdot),$
so it is also independent of $\gamma$. As a result,
$V_1^E(\gamma)
=
\mathcal V_1\big(\varnothing,1;x_1^E(\cdot)+\beta\widehat V_2^E(\cdot)\big)$
is constant in $\gamma$.

\smallskip
\noindent\textit{Step 4: Continuity of $D(\gamma)$.}
Under contract $I$, $\mu_2(\cdot|y_1)$ is independent of $\gamma$ by 
hypothesis that $a_1=2$ is chosen for every $\gamma$. Moreover, by Lemma
\ref{lem:dynamic:lambda-closed}, for each $y_1\in Y$ we have $\lambda_2^\gamma(y_1)=Z(y_1)/\gamma$
for constant $Z(y_1)>0$, so $\gamma\mapsto \lambda_2^\gamma(y_1)$ is continuous on $(0,\infty)$.
Since $Y$ is finite and each $q$ has full support on $Y$, 
$\lambda\mapsto g_q(x;\lambda)=-(1/\lambda)\log\big(\sum_{y\in Y}q(y)e^{-\lambda x(y)}\big)$ is
continuous on $(0,\infty)$ for any fixed payoff vector $x$. It follows from
\eqref{eq:app:screening:V2-reduction} that for each $y_1$, the continuation value
$\gamma\mapsto V_2^I(y_1;\gamma)$ is continuous as the pointwise maximum of two continuous
functions. Thus, $\gamma\mapsto z_\gamma^I(y_1)=x_1^I(y_1)+\beta V_2^I(y_1;\gamma)$ is continuous,
and therefore $\gamma\mapsto V_1^I(\gamma)=\mathcal V_1(\varnothing,2;z_\gamma^I)$ is continuous
because $\mathcal G(\cdot;\mu_1,\lambda_1)$ is continuous on $\R^{Y}$.
Under contract $E$, Step 3 shows $V_1^E(\gamma)$ is constant, so $D(\gamma)$ is continuous.

\medskip
\noindent\emph{Step 5.}
Since $V_1^I(\gamma)$ is strictly increasing and $V_1^E(\gamma)$ is constant, the difference
$D(\gamma)=V_1^I(\gamma)-V_1^E(\gamma)$ is strictly increasing. By the continuity argument above,
$D(\cdot)$ is also continuous on $(0,\infty)$. Fix $\gamma'<\gamma''$ with $D(\gamma')<0<D(\gamma'')$.
By continuity there exists $\gamma^\ast\in(\gamma',\gamma'')$ such that $D(\gamma^\ast)=0$, and by
strict monotonicity this $\gamma^\ast$ is unique. Hence, $D(\gamma)<0$ for all $\gamma<\gamma^\ast$
and $D(\gamma)>0$ for all $\gamma>\gamma^\ast$, which delivers the cutoff sorting.\end{proof}

\subsection{Proof of Corollary \ref{prop:dynamic:turnover}}

\begin{proof}
\noindent\emph{Step 1: Implementing innovation.}
The condition $C(\mu_2(\cdot|y_1=1),\lambda_1)> k$ means that there exists a finite vector
$x:Y\to\R$ such that $M(\mu_2(\cdot|y_1=1),\lambda_1;x)> k,$
where $M(\mu,\lambda;x)=\mathcal G(x;\mu,\lambda)-g_{q^1}(x;\lambda)$ is defined in
\eqref{eq:dynamic:Delta-C}. Consider a turnover policy that replaces the incumbent after $y_1=1$ and
offers the new agent the continuation contract $x_2(1,\cdot)=x$. Since the new agent's
state at $t=2$ is $(\mu_2(\cdot|y_1=1),\lambda_1)$, Proposition \ref{prop:dynamic:terminal-impl} implies
that $a_2=2$ is optimal for him after $y_1=1$. This proves part (i).

\medskip
\noindent\emph{Step 2: Net profit is strictly increasing in $\xi$.}
Recall the scale-up multiplier $\xi$ in \eqref{eq:multiplier}. Fix any turnover contract that induces $a_2=2$ after the first-period innovation success $(a_1=2,y_1=1)$. Under this contract, the principal's net expected profit can be written as
$\varPi^{\mathrm{turn}}(\xi)
=
\E\left[y_1+y_2+(\xi-1)\1\{a_1=2,y_1=1,a_2=2\}y_2\right]
-\E\left[v(x_1(y_1))+v(x_2(y_1,y_2))\right]
-\varphi\Proba(y_1=1).$
All wage costs and the hiring cost $\varphi\Proba(y_1=1)$ are independent of $\xi$. Since the first-period action is innovation and the turnover contract induces $a_2=2$ after $y_1=1$, we have
$\1\{a_1=2,y_1=1,a_2=2\}=\1\{y_1=1\}$ along the induced equilibrium path. Thus,
$\frac{\partial \varPi^{\mathrm{turn}}(\xi)}{\partial \xi}=
\E\left[\1\{y_1=1\}y_2\right]=
\Proba(y_1=1)\E[y_2\mid y_1=1,a_2=2]=
\big(m\theta_H+(1-m)\theta_L\big)
\sum_{q\in Q^{2}}\mu_2(q|y_1=1)q(1)
>0,$
because $m\theta_H+(1-m)\theta_L>0$ and
$\sum_{q\in Q^{2}}\mu_2(q|y_1=1)q(1)\in(\theta_L,\theta_H)\subset(0,1)$.
Thus, $\varPi^{\mathrm{turn}}(\xi)$ is affine in $\xi$ with strictly positive slope, and hence is strictly increasing in $\xi$.

\medskip
\noindent\emph{Step 3: Turnover is optimal.}
If the incumbent is kept after $y_1=1$,
$C(\mu_2(\cdot|y_1=1),\lambda_2(1))<k$ implies that, under any dynamic contract and any induced SPNE,
$a_2=2$ is never implemented after $(a_1=2,y_1=1)$ and the principal's expected profit is independent
of $\xi$ (Theorem \ref{thm:dynamic:trap-scaleup}). Let $\varPi^{\text{keep}}$ be the supremum of net expected profit among all dynamic
contracts that keep the incumbent after $y_1=1$ (which is independent of $\xi$):
$$
\varPi^{\text{keep}}
:=\sup\big\{\text{net expected profit from }(x_1,x_2)\ :\ (x_1,x_2)\text{ keeps incumbent after }y_1=1\big\}.
$$
Fix the turnover contract constructed in Steps 1--2 and let $\varPi^{\text{turn}}(\xi)$ denote its net profit.
By Step 2, $\varPi^{\text{turn}}(\xi)=\varPi^{\text{turn}}(1)+(\xi-1)s$ for some $s>0$.
Since $\varPi^{\text{keep}}<+\infty$ is assumed, let
$\bar \xi:=\max\big\{1,\ 1+\frac{\varPi^{\text{keep}}-\varPi^{\text{turn}}(1)}{s}\big\}.$
Then, $\varPi^{\text{turn}}(\xi)>\varPi^{\text{keep}}$ for all $\xi>\bar \xi$, so the turnover contract above yields strictly higher net expected profit than every contract that keeps the incumbent after $y_1=1$. This proves part (ii).
\end{proof}

\subsection{Proof of Theorem \ref{thm:cycles:frontier}}\label{app:proof:cycles}

\begin{lemma}\label{lem:cycles:capacity}
Fix $\lambda>0$, $p,\theta\in(0,1)$, $p\neq\theta$, and suppose $\mu=\delta_{\{q_\theta\}}$ with $q_\theta(1)=\theta$.
Then,
$$
C(\delta_{\{q_\theta\}},\lambda)=
\begin{cases}
\displaystyle \frac{1}{\lambda}\log \frac{1-p}{1-\theta} & \text{if }\theta>p,\\
\displaystyle \frac{1}{\lambda}\log \frac{p}{\theta} & \text{if }\theta<p.
\end{cases}
$$
\end{lemma}

\begin{proof}[Proof of Lemma \ref{lem:cycles:capacity}]
This is an immediate specialization of Lemma \ref{lem:dynamic:two-point}. Applying Lemma \ref{lem:dynamic:two-point} with $(\zeta,\vartheta)=(\theta,p)$ yields
$C(\delta_{\{q_\theta\}},\lambda)
=\frac{1}{\lambda}\log \max\left\{\frac{p}{\theta},\frac{1-p}{1-\theta}\right\}.$
If $\theta>p$, then $(1-p)/(1-\theta)>\hspace{0.02in}p/\theta$, so the maximum equals $(1-p)/(1-\theta)$; if
$\theta<p$ the maximum equals $p/\theta$.
\end{proof}
Theorem \ref{thm:cycles:frontier} relies on some technical results in Appendix \ref{app:dynamic:verification}:  Lemmas \ref{lem:dynamic:infty:mu-local}--\ref{lem:cycles:Wlimit-public}, and Corollary \ref{cor:dynamic:infty:bridge-subseq-tree}.  Here, we state two results that describe the long-run behavior of $(\mu_t,\lambda_t)$. See Appendix \ref{sec:appproof} for all their proofs. The first result below is a variation of \citet[][Lemma 6]{lanzani2025} and the second is that Bayesian posteriors concentrate on KL minimizers. 
\begin{lemma}\label{lem:cycles:llr-frequency}
Let $(T_n)_{n\ge 1}$ be a deterministic or random sequence with $T_n\to\infty$ such that
$$
\alpha_{T_n}:=\frac{1}{T_n}\sum_{t=1}^{T_n}\1\{a_t=2\} \longrightarrow \alpha\in[0,1]
\qquad\text{a.s.}
$$
Then, under the true distribution induced by $p_\ast$, the following hold a.s.
$$
\frac{\mathrm{LLR}^2(h_{T_n+1},Q^{2})}{T_n} \longrightarrow \alpha D_\ast^{2},
\qquad\text{and hence}\qquad
\lambda_{T_n+1}
=
\frac{\mathrm{LLR}^2(h_{T_n+1},Q^{2})}{\gamma T_n}
\longrightarrow
\frac{\alpha D_\ast^{2}}{\gamma}.
$$
\end{lemma}
 
\begin{lemma}\label{lem:cycles:posterior}
Let $\sum_{t=1}^{T}\1\{a_t=2\}\to\infty$ along a realized path and 
${Q}^{2}_*:=\arg\min_{q\in Q^{2}}\DKL(p^2_\ast\|q)$.
Then, under the true distribution induced by $p_\ast$,
$\mu_{T+1}(q|h_{T+1})\to 0$ a.s. $\forall q\notin {Q}^{2}_*$. 
\end{lemma}

\subsubsection{Proof of Theorem \ref{thm:cycles:frontier}}\label{app:proof:thm-cycles-frontier}

\begin{proof}[Proof of Theorem \ref{thm:cycles:frontier}]
 Define
$N_{2,T}:=\sum_{t=1}^{T}\1\{a_t=2\}$ and $\alpha_T:=N_{2,T}/T$.

\medskip
\noindent\emph{Step 1: reduce $\limsup_T\alpha_T$ to innovation times.}
Let $t_1<t_2<\dots$ be the (random) times at which the realized action is innovative, $a_{t_m}=2$.
For any $m\ge 1$ and any $T$ such that $t_m\le T<t_{m+1}$, we have $N_{2,T}=m$, hence $\alpha_T=\frac{m}{T}.$
Since $t_m\le T<t_{m+1}$, this yields 
\begin{align}\label{eq:app:cycles:alpha-two-sided}
\frac{m}{t_{m+1}-1} \le \alpha_T = \frac{m}{T} \le \frac{m}{t_m} = \alpha_{t_m}.
\end{align}
Then, for each $T$ with $N_{2,T}=m$, we have $\alpha_T\le \alpha_{t_m}$, so
$\limsup_{T\to\infty}\alpha_T\le \limsup_{m\to\infty}\alpha_{t_m}$.
Thus, $(\alpha_{t_m})_{m\ge 1}$ is a subsequence of $(\alpha_T)_{T\ge 1}$, hence
$\limsup_{m\to\infty}\alpha_{t_m}\le \limsup_{T\to\infty}\alpha_T$.
Therefore,
$\limsup_{T\to\infty}\alpha_T=\limsup_{m\to\infty}\alpha_{t_m}=: \alpha\in[0,1].$
By the definition of $\limsup$, there exists an increasing (possibly random) subsequence $(m_n)_{n\ge 1}$
such that $\alpha_{t_{m_n}}\to \alpha$ and $t_{m_n}\to\infty$. To apply Lemma \ref{lem:cycles:llr-frequency} at times $t_{m_n}$ (rather than $t_{m_n}+1$), 
set $T_n:=t_{m_n}-1$. For $n$ large enough, $t_{m_n}\ge 2$ so that $T_n\ge 1$ and $\alpha_{T_n}$ is well-defined.
Using $N_{2,t_m}=m$ and $N_{2,t_m-1}=m-1$ for $t_m\ge 2$, we have
$\left|\alpha_{t_m}-\alpha_{t_m-1}\right|
=
\big|\frac{m}{t_m}-\frac{m-1}{t_m-1}\big|
=
\frac{t_m-m}{t_m(t_m-1)}
\le
\frac{1}{t_m-1}.$
Since $t_{m_n}\to\infty$, $\alpha_{t_{m_n}-1}=\alpha_{T_n}\to \alpha$ almost surely.
Thus, Lemma \ref{lem:cycles:llr-frequency} applied to $(T_n)$ yields
\begin{align}\label{eq:app:cycles:lambda-subseq}
\lambda_{t_{m_n}} = \lambda_{T_n+1} \longrightarrow \bar\lambda:=\frac{\alpha D_\ast^{2}}{\gamma}
\qquad\text{a.s.}
\end{align}
where it is assumed that the roadmap is misspecified, i.e., $D_\ast^{2}>0$.

\medskip
\noindent\emph{Step 2: incentive capacity must asymptotically cover $k$ at innovation nodes.}
Fix an innovation time $t_m$ and the realized private history $h_{t_m}$. Let $(W_t)_{t\ge1}$ be the equilibrium continuation values defined by \eqref{eq:dynamic:infty:W}. For each candidate current action $a\in\{1,2\}$, define the immediate successor continuation-utility vector
$z_{t_m}^{a}(y)
:=
x_{t_m}(\chi_{t_m}(h_{t_m}),y)
+\beta W_{t_m+1}(h_{t_m},a,y),$ for $y\in Y.$
Let $(\mu_{t_m},\lambda_{t_m}):=(\mu_{t_m}(h_{t_m}),\lambda_{t_m}(h_{t_m}))$. Since $a_{t_m}=2$ is optimal at $h_{t_m}$ under the equilibrium contract, \eqref{eq:dynamic:infty:W} implies
$\mathcal G(z_{t_m}^2;\mu_{t_m},\lambda_{t_m})-k
\ge
g_{q^1}(z_{t_m}^1;\lambda_{t_m}).$
Set
$\eta_m:=
\beta\sup_{y\in Y}
\left|W_{t_m+1}(h_{t_m},2,y)-W_{t_m+1}(h_{t_m},1,y)\right|.$ Then, for every $y\in Y$,
$\big|z_{t_m}^1(y)-z_{t_m}^2(y)\big|
=\beta \big|W_{t_m+1}(h_{t_m},1,y)-W_{t_m+1}(h_{t_m},2,y)\big|
\le \eta_m,$
or equivalently,
$z_{t_m}^2-\eta_m\mathbf{1}\le z_{t_m}^1\le z_{t_m}^2+\eta_m\mathbf{1}.$ By monotonicity of $g_{q^1}(\cdot;\lambda_{t_m})$ and translation invariance (Lemma \ref{lem:translation}),
$$
g_{q^1}(z_{t_m}^2;\lambda_{t_m})-\eta_m
=
g_{q^1}(z_{t_m}^2-\eta_m\mathbf{1};\lambda_{t_m})
\le
g_{q^1}(z_{t_m}^1;\lambda_{t_m})
\le
g_{q^1}(z_{t_m}^2+\eta_m\mathbf{1};\lambda_{t_m})
=
g_{q^1}(z_{t_m}^2;\lambda_{t_m})+\eta_m.
$$
In particular, $g_{q^1}(z_{t_m}^1;\lambda_{t_m})\ge g_{q^1}(z_{t_m}^2;\lambda_{t_m})-\eta_m$. Combining with
$\mathcal G(z_{t_m}^2;\mu_{t_m},\lambda_{t_m})-k \ge g_{q^1}(z_{t_m}^1;\lambda_{t_m})$ yields $\mathcal G(z_{t_m}^2;\mu_{t_m},\lambda_{t_m})-k
\ge
g_{q^1}(z_{t_m}^2;\lambda_{t_m})-\eta_m,$
hence $M(\mu_{t_m},\lambda_{t_m};z_{t_m}^2)\ge k-\eta_m$, so $C(\mu_{t_m},\lambda_{t_m})\ge k-\eta_m$. We now show $\eta_{m_n}\to 0$ a.s. along the subsequence $(t_{m_n})$ from Step 1.
Since $Q_*^{2}=\{\hat q\}$ and innovation occurs infinitely often, Lemma \ref{lem:cycles:posterior} implies
$\mu_t(\hat q|h_t)\to 1$ a.s.; in particular, $\mu_{t_{m_n}}(\hat q|h_{t_{m_n}})\to 1$ a.s.
Fix any $N\ge 1$. Applying Lemma \ref{lem:dynamic:infty:mu-local} with $q_\star=\hat q$, $t=t_{m_n}$, and $h_t=h_{t_{m_n}}$ yields
$$
\sup_{\tilde h\in\mathcal H_{t_{m_n},N}(h_{t_{m_n}})}
\big\|\mu_{|\tilde h|}(\cdot|\tilde h)-\delta_{\{\hat q\}}\big\|_1
\le
2K^{N}\frac{1-\mu_{t_{m_n}}(\hat q|h_{t_{m_n}})}{\mu_{t_{m_n}}(\hat q|h_{t_{m_n}})}
\longrightarrow 0
\qquad\text{a.s.}
$$
for $K$ in \eqref{eq:app:infty:K}. Since $\lambda_t(h_t)=\mathrm{LLR}^2(h_t,Q^{2})/(\gamma(t-1))$ and $t_{m_n}\to\infty$, Lemma \ref{lem:dynamic:infty:llr-local} implies
$$
\sup_{\tilde h\in\mathcal H_{t_{m_n},N}(h_{t_{m_n}})}
\left|\lambda_{|\tilde h|}(\tilde h)-\lambda_{t_{m_n}}(h_{t_{m_n}})\right|
\le
\frac{B}{\gamma}\frac{1+\log(t_{m_n}+N)}{t_{m_n}}
\longrightarrow 0
\qquad\text{a.s.}
$$
for a constant $B<\infty$. By property \textup{(ii)} of the metric $d$ in Appendix \ref{subsec:dynamic:infty}, the above implies
$\sup_{\tilde h\in\mathcal H_{t_{m_n},N}(h_{t_{m_n}})}
d\big(\lambda_{|\tilde h|}(\tilde h),\lambda_{t_{m_n}}(h_{t_{m_n}})\big)\rightarrow 0$ a.s.
Combining with \eqref{eq:app:cycles:lambda-subseq} and triangle inequality for $d$ gives $\sup_{\tilde h\in\mathcal H_{t_{m_n},N}(h_{t_{m_n}})}
d\big(\lambda_{|\tilde h|}(\tilde h),\bar\lambda\big)\rightarrow 0$ a.s.
Thus, for each fixed $N$, the two suprema in \eqref{eq:dynamic:infty:state-unif} evaluated at $t=t_{m_n}$ converge to $0$
with $(\mu_\infty,\lambda_\infty)=(\delta_{\{\hat q\}},\bar\lambda)$. Since $\beta\in(0,1)$ and the contract is bounded,
Corollary \ref{cor:dynamic:infty:bridge-subseq-tree}, applied with
$\tau_n:=t_{m_n}$ and $L=1$, yields
$\sup_{\tilde h\in\mathcal H_{t_{m_n},1}(h_{t_{m_n}})}
\big|W_{|\tilde h|}(\tilde h)-W_{|\tilde h|}^\infty(\tilde h)\big|
\rightarrow 0$ $\text{a.s.,}$
where $(W_t^\infty)_{t\ge 1}$ solves \eqref{eq:dynamic:infty:Wlimit} for
$(\mu_\infty,\lambda_\infty)=(\delta_{\{\hat q\}},\bar\lambda)$.
Since each successor $(h_{t_{m_n}},a,y)$ belongs to
$\mathcal H_{t_{m_n},1}(h_{t_{m_n}})$, then
$\sup_{(a,y)\in\{1,2\}\times Y}
\big|W_{t_{m_n}+1}(h_{t_{m_n}},a,y)
-
W_{t_{m_n}+1}^\infty(h_{t_{m_n}},a,y)\big|
\rightarrow 0$ $\text{a.s.};$ $(W_t^\infty)_{t\ge 1}$ solves \eqref{eq:dynamic:infty:Wlimit} for
$(\mu_\infty,\lambda_\infty)=(\delta_{\{\hat q\}},\bar\lambda)$. By Lemma \ref{lem:cycles:Wlimit-public}, $W_{t_{m_n}+1}^\infty(h_{t_{m_n}},2,y)=
W_{t_{m_n}+1}^\infty(h_{t_{m_n}},1,y)$ $\forall y\in Y$. Thus,
$\eta_{m_n}\le
2\beta\sup_{(a,y)\in\{1,2\}\times Y}
\big|W_{t_{m_n}+1}(h_{t_{m_n}},a,y)-W_{t_{m_n}+1}^\infty(h_{t_{m_n}},a,y)\big|
\rightarrow 0$ $\text{a.s.}$, so
\begin{align}\label{eq:app:cycles:C-ge-k-subseq}
C(\mu_{t_{m_n}},\lambda_{t_{m_n}})\ge k-\eta_{m_n}\qquad\text{with }\eta_{m_n}\to 0\quad\text{a.s.}
\end{align}

\medskip
\noindent\emph{--- An upper bound.}
Fix $(\mu,\lambda)$ and any $x:Y\to\R$. Then, 
$$
M(\mu,\lambda;x)
=\sum_{q\in Q^{2}}\mu(q)\Big(g_q(x;\lambda)-g_{q^1}(x;\lambda)\Big)
=\sum_{q\in Q^{2}}\mu(q)\hspace{0.02in}M(\delta_{\{q\}},\lambda;x)
\le \sum_{q\in Q^{2}}\mu(q)\hspace{0.02in}C(\delta_{\{q\}},\lambda),
$$
where the inequality uses $M(\delta_{\{q\}},\lambda;x)\le \sup_{x'}M(\delta_{\{q\}},\lambda;x')=C(\delta_{\{q\}},\lambda)$ and we recall that $\mathcal G(x;\mu,\lambda)=\sum_{q\in Q^{2}}\mu(q)g_q(x;\lambda)$.
Taking the supremum over $x$ yields
\begin{align}\label{eq:app:cycles:C-mix-ub}
C(\mu,\lambda)\le \sum_{q\in Q^{2}}\mu(q)\hspace{0.02in}C(\delta_{\{q\}},\lambda).
\end{align}

\medskip
\noindent\emph{--- Reduction to KL-best fits.}
Let ${Q}^{2}_*:=\arg\min_{q\in Q^{2}}\DKL(p^2_\ast\|q)$.
Since innovation occurs infinitely often, $N_{2,T}\to\infty$, so Lemma \ref{lem:cycles:posterior} implies
$\mu_T(q)\to 0$ for all $q\notin{Q}^{2}_*$ almost surely. In particular, along $(t_{m_n})$,
$\varepsilon_n:=\sum_{q\notin {Q}^{2}_*}\mu_{t_{m_n}}(q) \rightarrow 0$ a.s.

If $\alpha=0$, $\limsup_{T\to\infty}\alpha_T=\alpha=0\le \alpha_\ast$, so the first claim holds.
Thus, assume $\alpha>0$, so $\bar\lambda>0$. By \eqref{eq:app:cycles:lambda-subseq}, there exists $n_0$ such that
 $\forall n\ge n_0$,
$\lambda_{t_{m_n}} \ge \bar\lambda/2.$
Since $Q^{2}$ is finite and Lemma \ref{lem:cycles:capacity} gives $C(\delta_{\{q\}},\lambda)<\infty$ for $\lambda>0$,
the constant $B:=\max_{q\in Q^{2}}C(\delta_{\{q\}},\bar\lambda/2)<\infty$
is well-defined. Lemma \ref{lem:cycles:capacity} also implies $C(\delta_{\{q\}},\lambda)$ is decreasing in $\lambda$
(for fixed $q$), so for $n\ge n_0$ and all $q\in Q^{2}$ we have $C(\delta_{\{q\}},\lambda_{t_{m_n}})\le B$.
Applying \eqref{eq:app:cycles:C-mix-ub} at $(\mu_{t_{m_n}},\lambda_{t_{m_n}})$ and splitting the sum into
$q\in{Q}_*^{2}$ and $q\notin{Q}_*^{2}$ gives, for $n\ge n_0$,
$C(\mu_{t_{m_n}},\lambda_{t_{m_n}})
\le
\max_{q\in{Q}^{2}_*} C(\delta_{\{q\}},\lambda_{t_{m_n}})+\varepsilon_n B.$
Combine with \eqref{eq:app:cycles:C-ge-k-subseq} and let $n\to\infty$ to obtain
\begin{align}\label{eq:app:cycles:k-ub}
k \le \max_{q\in{Q}^{2}_*} C(\delta_{\{q\}},\bar\lambda),
\end{align}
where we used $\varepsilon_n\to 0$, $\lambda_{t_{m_n}}\to \bar\lambda$, and continuity of
$C(\delta_{\{q\}},\lambda)$ in $\lambda>0$ (hence continuity of $\max_{q\in{Q}^{2}_*} C(\delta_{\{q\}},\lambda)$)
from Lemma \ref{lem:cycles:capacity}.

\medskip
\noindent\emph{--- Solve for the innovation frequency bound.}
By Lemma \ref{lem:cycles:capacity}, for each $q\in Q^{2}$ the inequality $C(\delta_{\{q\}},\lambda)\ge k$ is equivalent
to $\lambda\le \bar\lambda_k(q)$, where $\bar\lambda_k(q^2_H)=\bar\lambda^H_k$ and $\bar\lambda_k(q^2_L)=\bar\lambda^L_k$
as defined in \eqref{eq:cycles:lambda-bars-main}. Hence, \eqref{eq:app:cycles:k-ub} implies $\bar\lambda \le \max_{q\in{Q}^{2}_*}\bar\lambda_k(q).$ 
Since $\bar\lambda=\alpha D_\ast^{2}/\gamma$ \eqref{eq:app:cycles:lambda-subseq}, we conclude
$\limsup_{T\to\infty}\alpha_T=\alpha
 \le\
\frac{\gamma}{D_\ast^{2}}\hspace{0.03in}\max_{q\in{Q}^{2}_*}\bar\lambda_k(q)
 =: \alpha_\ast,$
which proves the first claim.

\medskip
\noindent\emph{--- Cycles.}
When $\alpha_\ast<1$ and innovation occurs infinitely often, then $\limsup_{T\to\infty}\alpha_T\le \alpha_\ast<1$.
If the routine action occurred only finitely many times, then eventually $a_t=2$ for all $t$, implying
$\alpha_T\to 1$, contradicting $\limsup_{T\to\infty}\alpha_T<1$. Thus, the routine action must occur infinitely often, which proves the second claim.
\end{proof}

\bibliography{ref}

\begin{thebibliography}{}

\bibitem[Anderson and Tushman, 2018]{innov90}
Anderson, P. and Tushman, M.~L. (2018).
\newblock Technological discontinuities and dominant designs: A cyclical model of technological change.
\newblock In {\em Organizational innovation}, pages 373--402. Routledge.

\bibitem[Berk, 1966]{berk66}
Berk, R.~H. (1966).
\newblock Limiting behavior of posterior distributions when the model is incorrect.
\newblock {\em The Annals of Mathematical Statistics}, 37(1):51--58.

\bibitem[Bewley, 2002]{bew02}
Bewley, T.~F. (2002).
\newblock Knightian decision theory. part i.
\newblock {\em Decisions in economics and finance}, 25:79--110.

\bibitem[Brier, 1950]{brier50}
Brier, G.~W. (1950).
\newblock Verification of forecasts expressed in terms of probability.
\newblock {\em Monthly weather review}, 78(1):1--3.

\bibitem[Carroll, 2015]{carroll15}
Carroll, G. (2015).
\newblock Robustness and linear contracts.
\newblock {\em American Economic Review}, 105(2):536--563.

\bibitem[Carroll, 2019]{carroll19}
Carroll, G. (2019).
\newblock Robustness in mechanism design and contracting.
\newblock {\em Annual Review of Economics}, 11(1):139--166.

\bibitem[Castelvecchi, 2023]{naturef23}
Castelvecchi, D. (2023).
\newblock How would room-temperature superconductors change science?
\newblock {\em Nature}, 621(7977):18--19.
\newblock News.

\bibitem[Cerreia-Vioglio et~al., 2025]{hansenmiss25}
Cerreia-Vioglio, S., Hansen, L.~P., Maccheroni, F., and Marinacci, M. (2025).
\newblock Making decisions under model misspecification.
\newblock {\em Review of Economic Studies}, page rdaf046.

\bibitem[Dahiya and Ray, 2012]{staged12}
Dahiya, S. and Ray, K. (2012).
\newblock Staged investments in entrepreneurial financing.
\newblock {\em Journal of Corporate Finance}, 18(5):1193--1216.

\bibitem[Dupuis and Ellis, 1997]{dupuis97}
Dupuis, P. and Ellis, R.~S. (1997).
\newblock {\em A Weak Convergence Approach to the Theory of Large Deviations}, volume 313.
\newblock John Wiley \& Sons.

\bibitem[D{\"u}tting et~al., 2024]{amb24}
D{\"u}tting, P., Feldman, M., Peretz, D., and Samuelson, L. (2024).
\newblock Ambiguous contracts.
\newblock {\em Econometrica}, 92(6):1967--1992.

\bibitem[Ederer and Manso, 2013]{manso13}
Ederer, F. and Manso, G. (2013).
\newblock Is pay for performance detrimental to innovation?
\newblock {\em Management Science}, 59(7):1496--1513.

\bibitem[Ewens and Marx, 2018]{founder18}
Ewens, M. and Marx, M. (2018).
\newblock Founder replacement and startup performance.
\newblock {\em The Review of Financial Studies}, 31(4):1532--1565.

\bibitem[Foutz and Srivastava, 1977]{miss77}
Foutz, R.~V. and Srivastava, R. (1977).
\newblock The performance of the likelihood ratio test when the model is incorrect.
\newblock {\em The annals of Statistics}, pages 1183--1194.

\bibitem[Fudenberg et~al., 2017]{fuden17}
Fudenberg, D., Romanyuk, G., and Strack, P. (2017).
\newblock Active learning with a misspecified prior.
\newblock {\em Theoretical Economics}, 12(3):1155--1189.

\bibitem[Garisto, 2023a]{naturelk23}
Garisto, D. (2023a).
\newblock Claimed superconductor lk-99 is an online sensation—but replication efforts fall short.
\newblock {\em Nature}, 620(7973):253--253.

\bibitem[Garisto, 2023b]{nature23}
Garisto, D. (2023b).
\newblock Lk-99 isn’t a superconductor—how science sleuths solved the mystery.
\newblock {\em Nature}, 620(7975):705--706.

\bibitem[Ghirardato, 1994]{ghir94}
Ghirardato, P. (1994).
\newblock Agency theory with non-additive uncertainty.
\newblock Technical report, mimeo, http://web. econ. unito. it/gma/paolo/age. pdf.

\bibitem[Gilboa and Schmeidler, 1989]{gilboa89}
Gilboa, I. and Schmeidler, D. (1989).
\newblock Maxmin expected utility with non-unique prior.
\newblock {\em Journal of mathematical economics}, 18(2):141--153.

\bibitem[Grossman and Hart, 1983]{gh83}
Grossman, S.~J. and Hart, O.~D. (1983).
\newblock An analysis of the principal-agent problem.
\newblock {\em Econometrica: Journal of the Econometric Society}, pages 7--45.

\bibitem[Hansen and Sargent, 2001]{hansen01}
Hansen, L. and Sargent, T.~J. (2001).
\newblock Robust control and model uncertainty.
\newblock {\em American Economic Review}, 91(2):60--66.

\bibitem[Holmstr{\"o}m, 1979]{holmstrom79}
Holmstr{\"o}m, B. (1979).
\newblock Moral hazard and observability.
\newblock {\em The Bell journal of economics}, pages 74--91.

\bibitem[Kellner, 2015]{kellner15}
Kellner, C. (2015).
\newblock Tournaments as a response to ambiguity aversion in incentive contracts.
\newblock {\em Journal of Economic Theory}, 159:627--655.

\bibitem[Kellner, 2017]{kellner17}
Kellner, C. (2017).
\newblock The principal-agent problem with smooth ambiguity.
\newblock {\em Review of Economic Design}, 21(2):83--119.

\bibitem[Klibanoff et~al., 2005]{smooth05}
Klibanoff, P., Marinacci, M., and Mukerji, S. (2005).
\newblock A smooth model of decision making under ambiguity.
\newblock {\em Econometrica}, 73(6):1849--1892.

\bibitem[Kuhn, 1970]{kuhn70}
Kuhn, T.~S. (1970).
\newblock {\em The structure of scientific revolutions}, volume~2.
\newblock University of Chicago press Chicago.

\bibitem[Lanzani, 2025]{lanzani2025}
Lanzani, G. (2025).
\newblock Dynamic concern for misspecification.
\newblock {\em Econometrica}, 93(4):1333--1370.

\bibitem[Lee et~al., 2023]{lk23}
Lee, S., Kim, J.-H., and Kwon, Y.-W. (2023).
\newblock The first room-temperature ambient-pressure superconductor.
\newblock {\em arXiv preprint arXiv:2307.12008}.

\bibitem[Loch et~al., 2011]{mgt11}
Loch, C.~H., DeMeyer, A., and Pich, M. (2011).
\newblock {\em Managing the unknown: A new approach to managing high uncertainty and risk in projects}.
\newblock John Wiley \& Sons.

\bibitem[Lopomo et~al., 2011]{amb11}
Lopomo, G., Rigotti, L., and Shannon, C. (2011).
\newblock Knightian uncertainty and moral hazard.
\newblock {\em Journal of Economic Theory}, 146(3):1148--1172.

\bibitem[Maccheroni et~al., 2006]{mmr06_JET}
Maccheroni, F., Marinacci, M., and Rustichini, A. (2006).
\newblock Dynamic variational preferences.
\newblock {\em Journal of Economic Theory}, 128(1):4--44.

\bibitem[Manso, 2011]{manso2011}
Manso, G. (2011).
\newblock Motivating innovation.
\newblock {\em The journal of finance}, 66(5):1823--1860.

\bibitem[Maselli, 2025]{maselli25}
Maselli, A. (2025).
\newblock Misspecification averse preferences.
\newblock Technical report, Penn Institute for Economic Research.

\bibitem[Maskin and Tirole, 1990]{maskin90}
Maskin, E. and Tirole, J. (1990).
\newblock The principal-agent relationship with an informed principal: The case of private values.
\newblock {\em Econometrica: Journal of the Econometric Society}, pages 379--409.

\bibitem[Maskin and Tirole, 1992]{maskin92}
Maskin, E. and Tirole, J. (1992).
\newblock The principal-agent relationship with an informed principal, ii: Common values.
\newblock {\em Econometrica: Journal of the Econometric Society}, pages 1--42.

\bibitem[Miao and Rivera, 2016]{miao16}
Miao, J. and Rivera, A. (2016).
\newblock Robust contracts in continuous time.
\newblock {\em Econometrica}, 84(4):1405--1440.

\bibitem[Nyarko, 1991]{nyarko91}
Nyarko, Y. (1991).
\newblock Learning in mis-specified models and the possibility of cycles.
\newblock {\em Journal of Economic Theory}, 55(2):416--427.

\bibitem[Radner, 1985]{radner85}
Radner, R. (1985).
\newblock Repeated principal-agent games with discounting.
\newblock {\em Econometrica: Journal of the Econometric Society}, pages 1173--1198.

\bibitem[Schmeidler, 1989]{schmeidler89}
Schmeidler, D. (1989).
\newblock Subjective probability and expected utility without additivity.
\newblock {\em Econometrica: Journal of the Econometric Society}, pages 571--587.

\bibitem[Spear and Srivastava, 1987]{spear87}
Spear, S.~E. and Srivastava, S. (1987).
\newblock On repeated moral hazard with discounting.
\newblock {\em The Review of Economic Studies}, 54(4):599--617.

\bibitem[Tian and Wang, 2014]{innov14}
Tian, X. and Wang, T.~Y. (2014).
\newblock Tolerance for failure and corporate innovation.
\newblock {\em The Review of Financial Studies}, 27(1):211--255.

\bibitem[Van~den Steen, 2005]{van2005}
Van~den Steen, E. (2005).
\newblock Organizational beliefs and managerial vision.
\newblock {\em Journal of Law, Economics, and organization}, 21(1):256--283.

\bibitem[Van~den Steen, 2010]{van10}
Van~den Steen, E. (2010).
\newblock On the origin of shared beliefs (and corporate culture).
\newblock {\em The RAND Journal of Economics}, 41(4):617--648.

\bibitem[Vuong, 1989]{miss89}
Vuong, Q.~H. (1989).
\newblock Likelihood ratio tests for model selection and non-nested hypotheses.
\newblock {\em Econometrica: journal of the Econometric Society}, pages 307--333.

\bibitem[Wasserman, 2003]{founder03}
Wasserman, N. (2003).
\newblock Founder-ceo succession and the paradox of entrepreneurial success.
\newblock {\em Organization science}, 14(2):149--172.

\end{thebibliography}
\bibliographystyle{apalike}

\clearpage

\renewcommand{\thepage}{OA-\arabic{page}} \setcounter{page}{1}
{\noindent\LARGE\bf For Online Publication}

\section{Appendix: General Extensions}\label{sec:general}
This appendix highlights that our qualitative results do not depend on the ARC functional form (Appendix \ref{sec:smooth}), whether the agent trusts the prior (Appendix \ref{sec:distrust}), or the LLR updating rule (Appendix \ref{subsec:disc:generic-update}). Appendix \ref{sec:disc:scope} then discusses the scope of Theorem \ref{thm:cycles:frontier}. All the proofs of results from this appendix are located in Appendix \ref{sec:appproof}.
\subsection{Smooth Bayesian criteria}\label{sec:smooth}
The recursive construction in Appendix \ref{subsec:dynamic:mmr} is tailored to ARC. However, the key logic behind the breakthrough trap is local: after a breakthrough, what matters is only the continuation ranking at the success node $(a_1=2,y_1=1)$. This makes it possible to extend the mechanism cleanly to the \textit{smooth Bayesian} criterion in \citet[][eq. (32)]{hansenmiss25} without re-solving the full dynamic program for every such preference.

Fix the post-success node $(a_1=2,y_1=1)$. Let $c:\Delta(Y)\times(Q^{2}\cup\{q^1\})\to [0,+\infty]$ be a valid misspecification penalty and $\phi:\R\to\R$ be continuous and strictly increasing. For every $q\in Q^{2}\cup\{q^1\}$ and $x\in\R^Y$, define the model-specific variational certainty equivalent
$g_q^c(x):=\min_{p\in\Delta(Y)}
\big\{\E_p[x(y)]+c(p,q)\big\}.$
Following \citet[][eq. (32)]{hansenmiss25}, define the post-success smooth Bayesian valuation of innovation by
\begin{align}\label{eq:disc:smooth:G}
\mathcal G_c^\phi\big(x;\mu_2(\cdot|y_1=1)\big)
:=
\phi^{-1}\Big(
\sum_{q\in Q^{2}}\mu_2(q|y_1=1)\hspace{0.02in}\phi\big(g_q^c(x)\big)
\Big).
\end{align}
The routine action is evaluated by $g_{q^1}^c(x)$. Hence, the post-success incentive gap generated by a continuation utility vector $x$ is
$M_c^\phi\big(\mu_2(\cdot|y_1=1);x\big)
:=
\mathcal G_c^\phi\big(x;\mu_2(\cdot|y_1=1)\big)-g_{q^1}^c(x)$
and the post-success incentive capacity is
$C_c^\phi\big(\mu_2(\cdot|y_1=1)\big)
:=
\sup_{x\in\R^Y}
M_c^\phi\big(\mu_2(\cdot|y_1=1);x\big).$
Define the \emph{optimistic within-model capacity}
$\overline C^{c}
:=
\sup_{x\in\R^Y}
\max_{q\in Q^{2}}
\big\{
g_q^c(x)-g_{q^1}^c(x)
\big\}.$

\begin{obs}\label{prop:disc:smooth-bayes-trap}
If $\overline C^{c}<k$, then
$C_c^\phi\big(\mu_2(\cdot|y_1=1)\big)<k.$
Thus, after $(a_1=2,y_1=1)$, no  contract can implement $a_2=2$ under the smooth Bayesian criterion \eqref{eq:disc:smooth:G}.
\end{obs}

Observation \ref{prop:disc:smooth-bayes-trap} shows that the ARC structure is not needed for the breakthrough trap. What matters is whether, after a breakthrough, there remains enough within-model incentive leverage to cover the effort cost $k$. A smooth Bayesian criterion can reweight or smooth the model-specific continuation values $\{g_q^c(x)\}_{q\in Q^{2}}$, but since $\phi$ is increasing, it can never raise the continuation value above the most favorable one-model valuation. In other words, if even the most optimistic  structured model cannot justify paying for continued innovation after success, then averaging across them cannot save innovation.  
\par The smooth Bayesian criterion in \citet[][eq. (32)]{hansenmiss25} was axiomatized by \citet[][Theorem 4]{maselli25}. When the agent has no concern for misspecification, i.e., $c(p,q)=\delta_{\{q\}}(p)$ in \eqref{eq:disc:smooth:G}, we get \citeauthor{smooth05}'s (\citeyear{smooth05}) smooth ambiguity criterion.

\subsection{Variational robust Bayesian criteria}\label{sec:distrust}
Appendix \ref{sec:smooth} kept the communicated posterior $\mu_2(\cdot|y_1=1)$ fixed and only changed how the model-specific continuation values are aggregated across $Q^{2}$. We now allow the agent to \textit{distrust} that posterior itself. Notice that this is stronger than the smooth Bayesian criterion: the agent now distrusts not only the structured models in the roadmap, but also the roadmap's weighting across those models. This attitude is relevant when an organization does not have a strong track record in the new task, so that a breakthrough can make the agent question the entire roadmap $(\mu_2(\cdot|y_1=1),Q^{2})$, not just its models.

Fix the post-success node $(a_1=2,y_1=1)$ and keep the notation from Appendix \ref{sec:smooth}, i.e., $c$ is the model-misspecification penalty, $g_q^c(x)$ is the model-specific variational certainty equivalent for each $q\in Q^{2}\cup\{q^1\}$ and $x\in\R^Y$, and $\overline C^{c}$ is the optimistic within-model capacity. Let
$d(\cdot;\mu_2(\cdot|y_1=1)):\Delta(Q^{2})\to[0,+\infty]$
be any belief-misspecification penalty, where $\inf_{\nu\in\Delta(Q^{2})}d(\nu;\mu_2(\cdot|y_1=1))=0.$
Following \citet[][eq. (37)]{hansenmiss25}, define the post-success \textit{variational robust Bayesian} valuation of innovation by
\begin{align}\label{eq:disc:varrobust:G}
\mathcal G_c^d\big(x;\mu_2(\cdot|y_1=1)\big)
:=
\min_{\nu\in\Delta(Q^{2})}
\Bigg\{
\sum_{q\in Q^{2}}\nu(q)\hspace{0.02in}g_q^c(x)
+
d\big(\nu;\mu_2(\cdot|y_1=1)\big)
\Bigg\}.
\end{align}
The routine action is evaluated by $g_{q^1}^c(x)$. The post-success incentive gap generated by a continuation utility vector $x$ is
$M_c^d\big(\mu_2(\cdot|y_1=1);x\big)
:=
\mathcal G_c^d\big(x;\mu_2(\cdot|y_1=1)\big)-g_{q^1}^c(x)$
and the post-success incentive capacity is
$C_c^d\big(\mu_2(\cdot|y_1=1)\big)
:=
\sup_{x\in\R^Y}
M_c^d\big(\mu_2(\cdot|y_1=1);x\big).$

\begin{obs}\label{prop:disc:varrobust-trap}
If $\overline C^{c}<k$, then
$C_c^d\big(\mu_2(\cdot|y_1=1)\big)<k.$
Thus, after $(a_1=2,y_1=1)$, no contract can implement $a_2=2$ under the variational robust Bayesian criterion \eqref{eq:disc:varrobust:G}.
\end{obs}

Observation \ref{prop:disc:varrobust-trap} shows that allowing the agent to distrust the prior cannot save innovation once every structured model is already too weak to justify it after success. The key difference from Appendix \ref{sec:smooth} is that the agent can now also choose a pessimistic distortion of the communicated posterior. However, this extra layer of distrust still cannot create incentive leverage that is absent within every single model: since $d\ge 0$ and $\nu$ only averages the collection $\{g_q^c(x)\}_{q\in Q^{2}}$, the value $\mathcal G_c^d$ can never exceed the most favorable one-model continuation value. Thus, if even the best structured model cannot support a post-success wedge of size $k$, then distrust of the entire roadmap only reinforces the breakthrough trap. Lastly, we should note that the variational robust Bayesian criterion in  \citet[][eq. (37)]{hansenmiss25} was axiomatized by \citet[][Theorem 2]{maselli25}.

\subsection{General updating rules}\label{subsec:disc:generic-update}

The LLR updating rule in Assumption \ref{ass:dynamic:llr} is only one way to update misspecification concerns. In our framework, the breakthrough-trap mechanism requires only three elements: (i) the update must be driven by a statistic of the roadmap's \emph{absolute} fit to the realized signal, not merely by posterior odds within $Q^{2}$; (ii) worse fit must weakly raise misspecification concern; and (iii) any imperfect fit must induce a strictly positive alarm. The next result shows that these primitive conditions suffice to generate breakthrough traps.

More formally, for each signal $y\in Y$, let
$b(y):=\max_{q\in Q^{2}} q(y)$
denote the roadmap's best attainable fit to $y$. Consider any updating rule of the form
$\lambda_2(y)=\Psi\big(\varLambda(b(y))\big),$
where $\varLambda:(0,1]\to\R_+$ is continuous and strictly decreasing with $\varLambda(1)=0$, and $\Psi:\R_+\to\R_+$ is continuous, weakly increasing, and satisfies $\Psi(0)=0<\Psi(z)$ for all $z>0$.

\begin{obs}\label{prop:disc:generic-update}
 Under the updating rule $\lambda_2(\cdot)=\Psi(\varLambda(b(\cdot)))$, there exists $\bar\theta_H\in(p,1)$ such that, for every $\theta_H\in(p,\bar\theta_H)$, there is a $\bar\theta_L\in(0,p)$ with the following property: for all $\theta_L\in(0,\bar\theta_L)$, the post-success incentive capacity satisfies
$C\big(\mu_2(\cdot|y_1=1),\lambda_2(1)\big)<k.$
Consequently, under any dynamic contract and any induced subgame-perfect equilibrium, the continuation action $a_2=2$ is not implementable after $(a_1=2,y_1=1)$.
\end{obs}

Observation \ref{prop:disc:generic-update} shows that the breakthrough-trap mechanism does not rely on the LLR rule. What matters is more basic: an updating rule must detect poor absolute fit of the roadmap and translate that poor fit into a strictly positive increase in misspecification concern. When $\theta_H$ is close to $p$, the Bayesian improvement in continuation incentives is arbitrarily weak because success barely distinguishes innovation from the routine action, whereas any nondegenerate fit-based alarm remains strictly positive as long as $\theta_H<1$. This is why the breakthrough trap is a generic phenomenon for fit-based updating rules.

A natural way to construct the updating rule $\varLambda$ is from a forecast-loss function for the realized success event. For example, if $\ell(1,m)$ is continuous and strictly decreasing in the forecast probability $m$, then
$\varLambda(m):=\ell(1,m)-\ell(1,1)$
satisfies the required conditions. More concretely, this approach nests several familiar ``scoring rules'' such as
$$
\ell_{\log}(1,m)=-\log m,\qquad
\ell_{\mathrm B}(1,m)=(1-m)^2,\qquad
\ell_{\mathrm{sph}}(1,m)=1-\frac{m}{\sqrt{m^2+(1-m)^2}},
$$
where the LLR rule is the logarithmic-loss case $\ell_{\log}$ with $\Psi(z)=z/\gamma$, the Brier score is the quadratic-loss case $\ell_{\mathrm B}$ \citep{brier50}, and the spherical score is based on the loss $\ell_{\mathrm{sph}}$.

\subsection{Discussion: scope of the mechanism}\label{sec:disc:scope}

Appendices \ref{sec:smooth}--\ref{subsec:disc:generic-update} clarify the scope of the breakthrough-trap mechanism. \ref{sec:smooth} and \ref{sec:distrust} show that the post-success implementability failure is not tied to ARC, to linear averaging across benchmark models, or to trusting the communicated prior.  \ref{subsec:disc:generic-update} complements this by showing that the dynamic trigger is also not tied to the LLR rule.  

\par These extensions also clarify the scope of Theorem \ref{thm:cycles:frontier}. Its closed-form frontier
$\alpha_\ast=\frac{\gamma}{D_\ast^{2}}\max_{q\in{Q}_*^{2}}\bar\lambda_k(q)$ in \eqref{eq:cycles:alpha-bar}
is special to the ARC-LLR pair because, in this case, both components are available in closed form: the largest misspecification level compatible with local implementability and the long-run rate at which innovation generates misspecification concern. Under a more general  decision criterion and a more general fit-based updating rule, the same speed-limit logic still applies, but the frontier becomes implicit. Specifically, let $\mathscr C(\tau)$ denote the criterion-specific incentive capacity as a function of the relevant distrust index $\tau$, and let $\tau(\alpha)$ denote the asymptotic distrust level generated by long-run innovation frequency $\alpha$. Then, any equilibrium path with innovation frequency $\alpha$ must satisfy $\mathscr C(\tau(\alpha))\ge k$; otherwise, innovation eventually becomes locally uncontractible. Thus, the long-run speed limit is characterized by the cutoff
$\alpha^{\mathscr C}:=\sup\{\alpha\in[0,1]: \mathscr C(\tau(\alpha))\ge k\}.$
 What is special in Theorem \ref{thm:cycles:frontier} is therefore the closed-form expression of the frontier, not the existence of the frontier itself. 
\par Moreover, if $\alpha^{\mathscr C}<1$ and innovation occurs infinitely often along an equilibrium path, then the routine action must also occur infinitely often. Otherwise, the routine action would be played only finitely many times, so the long-run innovation frequency would converge to 1, contradicting the bound $\alpha\leq \alpha^{\mathscr C}<1$. Thus, the implicit frontier $\alpha^{\mathscr C}$ still generates exploration-exploitation cycles even when it is not available in closed form.

\section{Appendix: Roadmap design}\label{sec:disc:roadmap-no-free-lunch}

It is natural to ask whether the principal can avoid breakthrough traps by designing the roadmap itself. We show that the answer is very limited. When the agent's misspecification-sensitivity parameter $\gamma$ is known, roadmap design can only help locally by making a success less surprising under the optimistic model. However, this same distortion can be costly because it can tighten the long-run speed limit. When $\gamma$ is unknown, the problem is even sharper: under full-support restriction, no roadmap can rule out breakthrough traps uniformly over all types $\gamma>0$. More broadly, we can also imagine that the principal may be constrained or disciplined by organization's protocols and credibility when designing a roadmap, e.g., a roadmap that departs too far from established protocol or scientific evidence is costly to defend and may not be trusted. Adding such constraints suggests that roadmap design is not an effective approach to address breakthrough traps. 
All proofs from this appendix are in Appendix \ref{sec:appproof}.
\subsection{When $\gamma$ is known}
Observation \ref{obs:roadmap-known-gamma} shows that, when $\gamma$ is known, the only direct roadmap lever against the success-node misspecification shock is the success probability of the optimistic model: $\theta_H$. Thus, prior design cannot directly help here because $\mu_1$ does not enter $\lambda_2(1)$.
\begin{obs}\label{obs:roadmap-known-gamma}
Fix $\gamma>0$. There exists a unique $\hat\theta_H(\gamma)\in(p,1)$ such that
$\lambda_2(1)\leq \lambda^\ast$ if and only if $\theta_H\geq \hat\theta_H(\gamma)$. Moreover, $\lambda_2(1)$ is independent of $\mu_1$ and $\theta_L$.
\end{obs}

Observation \ref{obs:roadmap-no-free-lunch-alpha} establishes a sharp no-free-lunch result. If the principal raises $\theta_H$ to make an early success look less surprising, then she simultaneously lowers the long-run frontier in Theorem \ref{thm:cycles:frontier}, provided the optimistic model already overstates the success probability relative to the true innovation process. Thus, roadmap design can improve local implementability only by risking lower long-run innovation frequency. 

\begin{obs}[No free lunch]\label{obs:roadmap-no-free-lunch-alpha}
Suppose $q^2_H$ is the unique KL minimizer in Theorem \ref{thm:cycles:frontier} and $p<\theta_\ast<\theta_H$. Restrict attention to  $\theta_H$'s in a neighborhood where $q^2_H$ remains the unique KL minimizer. Then, $\lambda_2(1)$ is strictly decreasing in $\theta_H$, but
$\alpha_\ast=\frac{\gamma}{k}\frac{\log\frac{1-p}{1-\theta_H}}{\DKL(\theta_\ast\|\theta_H)}$ in (\ref{eq:cycles:alpha-bar})
is also strictly decreasing in $\theta_H$. Thus, making the roadmap more success-friendly relaxes the success-node misspecification shock but also tightens the long-run speed limit.
\end{obs}
In summary, when 
$\gamma$ is known, an endogenous roadmap is only a local calibration device: it can reduce the success-node shock by making success less surprising to the agent. However, when it does so by moving the roadmap away from the true innovation process, Observation \ref{obs:roadmap-no-free-lunch-alpha} shows that it creates a global misspecification cost, and Theorem \ref{thm:cycles:frontier} shows that this eventually lowers long-run innovation frequency.

\subsection{When $\gamma$ is unknown}
Observation \ref{obs:roadmap-unknown-gamma} shows that roadmap design is even less compelling when $\gamma$ is unknown. Under  full-support restriction on structured models in $Q^{2}$, the problem has no uniform solution across all types $\gamma$. If one drops full support, the only uniform way to kill the success-node shock is to set $\theta_H=1$, which is economically vacuous: success is then fully explained by construction, so the roadmap ceases to be a meaningful diagnostic object.

\begin{obs}\label{obs:roadmap-unknown-gamma}
Fix any full-support roadmap with $0<\theta_L<p<\theta_H<1$. Then, there exists $\bar\gamma>0$ such that
$C(\mu_2(\cdot|y_1=1),\lambda_2(1))<k$ for all $\gamma\in(0,\bar\gamma)$. Consequently, no full-support roadmap avoids breakthrough traps uniformly over all $\gamma>0$. If full support is dropped, uniform elimination of the success-node misspecification shock requires $\theta_H=1$.
\end{obs}

Taken together, Observations \ref{obs:roadmap-known-gamma}--\ref{obs:roadmap-unknown-gamma} indicate that roadmap design is not the right remedy for breakthrough traps. When $\gamma$ is known, it can only trade off a local gain against the global misspecification cost in Theorem \ref{thm:cycles:frontier}. When $\gamma$ is unknown, it cannot eliminate the trap without collapsing to a trivial roadmap. This suggests that the more natural remedies are those in Section \ref{sec:avoidtrap} that affect the evolution of misspecification concerns.

\section{Appendix: Infinite-Horizon Contracts}\label{app:infinite}

This appendix explores the infinite-horizon extension of 
Section \ref{sec:dynamic}: the \emph{history-dependent continuation operator} that governs whether
innovation can be sustained by incentives after informative histories. The two-period analysis
identifies a breakthrough trap when the post-success misspecification concern tightens continuation incentives
enough to make the policy ``continue innovating after success'' uncontractible. The purpose of this appendix is to
(i) provide a long-run benchmark for that continuation incentive environment under sustained play, and
(ii) give a formal bridge from convergence of the state $(\mu_t,\lambda_t)$ to convergence of
the entire  continuation evaluation.

\medskip
\noindent\textbf{Outline.}
Appendix \ref{subsec:app:infty:limit} characterizes the long-run limit of the one-step ARC operator
$x\mapsto \mathcal G^a(x;\mu_t,\lambda_t)$ along histories with an eventual constant action. The results
pin down whether long-run continuation incentives converge to \citeauthor{manso2011}'s
(\citeyear{manso2011}) EU benchmark, to a worst-case criterion, or to an interior ARC limit. These results follow from results in \citet{lanzani2025} and are simpler because of our finiteness assumptions on outcomes and models. The novelty is in
Appendix \ref{subsec:dynamic:infty}, where we show that when the state stabilizes locally along continuation
contingencies, the induced infinite-horizon recursion converges along the realized history to the
recursion associated with the limiting operator. This justifies using the limiting operator as a
long-run approximation for incentive provision.

\subsection{Long-run continuation incentives}\label{subsec:app:infty:limit}

This appendix formalizes the infinite-horizon version of the learning environment in
Section \ref{subsec:dynamic:env}. Throughout, $Y$ and $A$ are finite. The main object of interest is the
\emph{continuation operator} $x\mapsto \mathcal G^a(x;\mu_t,\lambda_t)$, because it is the
node-by-node object that enters the implementability inequalities through $M(\cdot)$ and the
incentive capacity $C(\mu,\lambda)$ in (\ref{eq:dynamic:Delta-C}).

\medskip
\noindent\textbf{Global structured models.}
A \emph{(global) structured model} is a vector $q=(q^a)_{a\in A}\in\Delta(Y)^A$, where $q^a$ is the
(one-step) outcome distribution under action $a$. Let $Q\subset\Delta(Y)^A$ be finite and assume that
$q^a(y)>0$ for all $q\in Q$, $a\in A$, and $y\in Y$. The agent has an initial second-order belief
$\mu_1\in\Delta(Q)$ with full support.

\medskip
\noindent\textbf{Objective data-generating process.}
Let $p_\ast=(p_\ast^a)_{a\in A}\in\Delta(Y)^A$ denote the true model. Conditional on the action path $(a_t)_{t\ge1}$, outcomes are independent over time. 

\medskip
\noindent\textbf{Public vs. private histories.}
The principal observes only outcomes, while the agent observes both outcomes and his own past actions.
Accordingly, we distinguish:

\begin{itemize}[leftmargin=18pt]
\item \emph{Public history:} $s_1=\varnothing$, $S_t:=Y^{t-1}$ for $t\ge 2$,
$s_t=(y_1,\dots,y_{t-1})\in S_t$.
\item \emph{Private history:} $h_1=\varnothing$, $H_t:=(A\times Y)^{t-1}$ for $t\ge 2$,
$h_t=(a_1,y_1,\dots,a_{t-1},y_{t-1})\in H_t$.
\end{itemize}

Let $\chi_t:H_t\to S_t$ be the public projection $\chi_t(h_t):=(y_1,\dots,y_{t-1})$. Contracts must be
measurable w.r.t. $S_t$ (not $H_t$), whereas Bayesian updating of $(\mu_t,\lambda_t)$ is measurable
w.r.t. $H_t$. When no confusion arises, we write $\chi(h_t)$ instead of $\chi_t(h_t)$. At each $t\geq1$, the action $a_t$ is chosen after observing the private history $h_t=(a_1,y_1,\dots,a_{t-1},y_{t-1})\in H_t$ but before observing $y_t$. Thus, $a_t$ is $\mathcal{F}_{t-1}$-measurable, where $\mathcal{F}_t := \sigma(a_1, y_1, \ldots, a_t, y_t)$.

\medskip
\noindent\textbf{Bayesian updating of $\mu$.}
Given a \emph{private} history $h_t=(a_1,y_1,\dots,a_{t-1},y_{t-1})\in H_t$, define the likelihood of $q\in Q$, $L_t(q;h_t):=\prod_{\tau=1}^{t-1} q^{a_\tau}(y_\tau).$
For all $t$, $h_t$, and $q\in Q$, the posterior belief is given by Bayes rule:
\begin{align}
\mu_t(q|h_t)\propto\mu_1(q)L_t(q;h_t).
\label{eq:dynamic:infty:mu-update}
\end{align}

\medskip
\noindent\textbf{One-step ARC aggregator.}
For any action $a\in A$, posterior $\mu\in\Delta(Q)$, intensity $\lambda>0$, and continuation payoff vector
$x:Y\to\R$ (already in utility units), define
$$
\mathcal G^a(x;\mu,\lambda):=\sum_{q\in Q}\mu(q)\hspace{0.02in}g_{q^a}(x;\lambda),
\qquad
g_{q^a}(x;\lambda):=-\frac{1}{\lambda}\log\Big(\sum_{y\in Y}q^a(y)e^{-\lambda x(y)}\Big).
$$
Thus, $\mathcal G^a(x;\mu,\lambda)$ is the agent's valuation of the continuation payoff $x$
if he chooses action $a$ at a node with robustness state $(\mu,\lambda)$. In the contracting problem,
the principal selects continuation vectors $x$ to satisfy incentive constraints, so the long-run behavior
of the operator $x\mapsto \mathcal G^a(x;\mu_t,\lambda_t)$ determines the long-run feasibility of
continued innovation.

\paragraph{Best-fitting models.}
For each action $a\in A$, define the set of KL-best-fit models
\begin{align}
{Q}^{a}_*:=\arg\min_{q\in Q}\DKL \big(p_\ast^a\|q^a\big),
\label{eq:dynamic:infty:Q-bestfit}
\end{align}
which is nonempty since $Q$ is finite. The next result is the standard result showing that Bayesian posteriors concentrate on KL minimizers, which builds on the logic of \citet{berk66}.  \citet{lanzani2025} already proves more general versions of it, so we omit its proof. 

\begin{obs}[Posterior concentration on KL-best fits]\label{prop:dynamic:infty:posterior}
Fix $a\in A$ and suppose that along a realized path there exists $T<\infty$ such that $a_t=a$ for all
$t\ge T$. Then, under the true distribution induced by $p_\ast$, the following hold almost surely:
\begin{enumerate}[label=(\alph*),leftmargin=18pt]
\item For any $q\notin {Q}^{a}_*$ and any $\hat{q}\in {Q}^{a}_*$,
$\lim_{t\to\infty}\frac{1}{t}\log\frac{\mu_t(q|h_t)}{\mu_t(\hat{q}|h_t)}
=
-(\DKL(p_\ast^a\|q^a)-\DKL(p_\ast^a\|\hat{q}^a))<0.$
In particular, $\mu_t(q|h_t)\to 0$ for every $q\notin {Q}^{a}_*$.
\item Every subsequential limit $\nu$ of $(\mu_t(\cdot|h_t))_{t\ge 1}$ satisfies
$\mathrm{supp}(\nu)\subseteq {Q}^{a}_*$. If ${Q}^{a}_*=\{\hat{q}\}$ is a singleton, then
$\mu_t(\cdot|h_t)\to \delta_{\hat{q}}$.
\item If there exists $q\in Q$ with $q^a=p_\ast^a$, then
${Q}^{a}_*=\{q\in Q:q^a=p_\ast^a\}$ and every subsequential limit $\nu$ satisfies
$\mathrm{supp}(\nu)\subseteq \{q\in Q:q^a=p_\ast^a\}$.
\end{enumerate}
\end{obs}

Under an eventually constant action $a$, the long-run posterior does not ``select the truth'' unless the
model class $Q$ is correctly specified. Instead, beliefs concentrate on the subset ${Q}^{a}_*$ of
structured models that best approximate the data in terms of KL divergence. Thus, even with unlimited data,
the relevant ``long-run roadmap'' is the set of best-fitting structured models, not necessarily the
true model. This matters for contracts because continuation incentives depend on the operator
averaging $g_{q^a}$ under $\mu_t$, and the long-run averaging is over ${Q}^{a}_*$.

\paragraph{Misspecification evidence.}
As in the main text, we continue to assume that the set of unstructured alternatives is the full simplex $\Delta(Y)^A$ \citep[see,][Section 5.1, p. 1351]{lanzani2025}.
Define the log-likelihood ratio (LLR) against all unstructured alternatives by
\begin{align}
\mathrm{LLR}(h_t,Q):=
-\log\left(
\frac{\max_{q\in Q}\prod_{\tau=1}^{t-1} q^{a_\tau}(y_\tau)}
     {\max_{p\in\Delta(Y)^A}\prod_{\tau=1}^{t-1} p^{a_\tau}(y_\tau)}
\right) \quad\forall t, h_t\in H_t.
\label{eq:dynamic:infty:LLR}
\end{align}
Appendix \ref{subsec:app:infty:innovation-llr} shows that $\text{LLR}^2$ in \eqref{eq:dynamic:infty:LLR2} is a special case of \eqref{eq:dynamic:infty:LLR} in our innovation settings. 
\begin{obs}\label{prop:dynamic:infty:llr-rate}
Fix $a\in A$ and suppose that along a realized path, there exists $T<\infty$ such that $a_t=a$ for all
$t\ge T$. Then, under the true distribution induced by $p_\ast$,
\begin{align}
\lim_{t\to\infty}\frac{\mathrm{LLR}(h_t,Q)}{t}
=
\min_{q\in Q}\DKL \big(p_\ast^a\|q^a\big).
\label{eq:dynamic:infty:llr-rate}
\end{align}
\end{obs}

Define the limit above as $D_\ast^{a}:=\min_{q\in Q}\DKL(p_\ast^a\|q^a)$.
This KL discrepancy is a primitive measure of \emph{roadmap misspecification} under action $a$:
it is the best achievable approximation error within the organization's structured models. In particular,
when $D_\ast^{a}>0$, misspecification evidence grows linearly in time along sustained play of $a$, so the
long-run robustness state need not ``wash out.'' This quantity will play a key role below.

\paragraph{Limits of the one-step operator.}
The next result presents the two operator limits that will anchor the long-run continuation environment. This is a simple limit of the ARC criterion, which is used in \citet{lanzani2025}, so we omit its proof.

\begin{obs}\label{prop:dynamic:infty:limits}
Fix $a\in A$, $\mu\in\Delta(Q)$, and $x:Y\to\R$.
\begin{enumerate}[label=(\alph*),leftmargin=18pt]
\item \textup{Vanishing concern}. $\displaystyle \lim_{\lambda\downarrow 0}\mathcal G^a(x;\mu,\lambda)
=\sum_{q\in Q}\mu(q)\sum_{y\in Y} q^a(y)\hspace{0.02in}x(y)$.
\item \textup{Worst-case}. $\displaystyle \lim_{\lambda\uparrow\infty}\mathcal G^a(x;\mu,\lambda)=\min_{y\in Y}x(y)$.
\end{enumerate}
\end{obs}

\paragraph{Long-run robustness types.}
For sequences $(a_t),(b_t)$ with $b_t>0$, we write $a_t=o(b_t)$ to mean $a_t/b_t\to 0$ as $t\rightarrow\infty$. We follow \citet{lanzani2025} by classifying the agent's types depending on how he updates his misspecification concerns over time.

\begin{definition}\label{def:dynamic:infty:types}
Let $\Lambda:\bigcup_{t\ge 1}(A\times Y)^{t-1}\to[0,\infty)$ be a rule and set $\lambda_t(h_t):=\Lambda(h_t)$. Then,
\begin{enumerate}[label=(\alph*),leftmargin=18pt]
\item The rule $\Lambda$ is \emph{lenient} if
$$
\Lambda(h_t)=o\left(\frac{\mathrm{LLR}(h_t,Q)}{t}\right)
\quad\text{for every history sequence }(h_t)_{t\ge 1}.
$$
\item The rule $\Lambda$ is \emph{demanding} if
$$
o\big(\Lambda(h_t)\big)=\frac{\mathrm{LLR}(h_t,Q)}{t}
\quad\text{for every history sequence }(h_t)_{t\ge 1}.
$$
\item Fix $\gamma>0$. The rule $\Lambda$ is \emph{statistically sophisticated with scale $\gamma$} if
$$
\Lambda(h_t)=\frac{\mathrm{LLR}(h_t,Q)}{\gamma(t-1)}
\quad\text{for all }t\ge 2\text{ and all }h_t,
$$
with $\Lambda(h_1)=0$.
\end{enumerate}
\end{definition}
The next result summarizes the limiting implications of each type above. It follows immediately from the proof of \citet[][Theorem 2]{lanzani2025}, so we omit its proof.
\begin{obs}\label{prop:dynamic:infty:limit-pref}
Fix $a\in A$ and suppose that along a realized path there exists $T<\infty$ such that $a_t=a$ for all
$t\ge T$. Then, under the true distribution induced by $p_\ast$, the following hold almost surely:
\begin{enumerate}[label=(\alph*),leftmargin=18pt]
\item \textup{Belief support.} Every subsequential limit $\nu$ of $(\mu_t(\cdot|h_t))_{t\ge 1}$ satisfies
$\mathrm{supp}(\nu)\subseteq {Q}^{a}_*$.
\item \textup{Limit of $\lambda_t$.} Let $D_\ast^{a}:=\min_{q\in Q}\DKL(p_\ast^a\|q^a)$, which equals the limit in
\eqref{eq:dynamic:infty:llr-rate}. Then:
\begin{enumerate}[label=(\roman*),leftmargin=18pt]
\item If $\Lambda$ is lenient and $D_\ast^{a}>0$, then $\lambda_t(h_t)\to 0$.
\item If $\Lambda$ is demanding and $D_\ast^{a}>0$, then $\lambda_t(h_t)\to +\infty$.
\item If $\Lambda$ is statistically sophisticated with scale $\gamma$, then $\lambda_t(h_t)\to D_\ast^{a}/\gamma$.
\end{enumerate}
If $D_\ast^{a}=0$ then $\mathrm{LLR}(h_t,Q)/t\to 0$, and under a lenient rule $\lambda_t(h_t)\to 0$.
\item \textup{Limiting one-step preference.} Fix any $x:Y\to\R$ and any subsequential limit $\nu$ from part
\textup{(a)}. Along the same subsequence,
$$
\mathcal G^a(x;\mu_t(\cdot|h_t),\lambda_t(h_t))
\longrightarrow
\begin{cases}
\displaystyle \sum_{q\in Q}\nu(q)\sum_{y\in Y}q^a(y)\hspace{0.02in}x(y), & \text{lenient with }D_\ast^{a}>0,\\
\displaystyle \min_{y\in Y}x(y), & \text{demanding with }D_\ast^{a}>0,\\
\displaystyle \sum_{q\in Q}\nu(q)\hspace{0.02in}g_{q^a} \big(x;D_\ast^{a}/\gamma\big), & \text{statistically sophisticated.}
\end{cases}
$$
\end{enumerate}
\end{obs}

In Section \ref{sec:dynamic}, continued innovation is implementable if and only if the
principal can choose a continuation vector $x$ that makes the innovation operator sufficiently favorable,
i.e., that yields $M(\mu,\lambda;x)\ge k$ and hence $C(\mu,\lambda)\ge k$ in (\ref{eq:dynamic:Delta-C}).
Observation \ref{prop:dynamic:infty:limit-pref} provides a long-run benchmark for this feasibility problem
under sustained play of a fixed action.
(i) Under lenient rules, misspecification concerns vanish and continuation incentives converge to the EU
environment of \citeauthor{manso2011} (\citeyear{manso2011});
(ii) under demanding rules, continuation evaluation collapses to \citeauthor{gilboa89}'s (\citeyear{gilboa89}) maxmin criterion, producing the
tightest possible continuation incentives;
(iii) under statistically sophisticated rules, misspecification concerns converge to an interior level
$\lambda_\infty=D_\ast^{a}/\gamma$ whenever $D_\ast^{a}>0$, in which case, endogenous misspecification concerns generate a \emph{persistent}
wedge relative to the EU benchmark.
Intuitively, $D_\ast^{a}$ measures how well the organization's roadmap can approximate the data under
action $a$.

\medskip
\noindent\textbf{Connection to the breakthrough trap.}
Recall that the two-period breakthrough trap in Section \ref{sec:dynamic} arises when a favorable intermediate outcome generates sufficiently strong
misspecification evidence that the resulting $(\mu_2,\lambda_2)$ makes continued innovation
uncontractible. The discussion above shows how to interpret this finite-horizon mechanism
through the lens of long-run limits: when sustained innovation has $D_\ast^{2}>0$ and robustness is
updated in a statistically sophisticated way, the continuation environment converges to an interior robust
benchmark rather than to expected utility. The breakthrough trap can therefore be viewed as an early
manifestation of a more general epistemic force: endogenous misspecification concerns can generate persistent
tightness of continuation incentives, even as data accumulates in the long run.

\subsection{Limiting Recursive Preference}\label{subsec:dynamic:infty}

We provide a bridge from convergence of the state
$(\mu_t(\cdot|h_t),\lambda_t(h_t))$ to convergence of the agent's induced recursive continuation values. This is
the infinite-horizon analogue of the backward-induction logic in Section \ref{sec:dynamic}: when the
local state stabilizes along continuation contingencies, the entire dynamically consistent evaluation of
future incentives becomes asymptotically time-homogeneous. To stay consistent with Section \ref{sec:dynamic}, let $A=\{1,2\}$.

\medskip
\noindent\textbf{Extended domain.}
Let $\overline{\R}_+:=[0,+\infty]$. Fix a metric $d$ on $\overline{\R}_+$ such that:
\begin{enumerate}[label=(\roman*),leftmargin=18pt]
\item $(\overline{\R}_+,d)$ is compact;
\item the topology induced by $d$ on $[0,+\infty)$ is the subspace topology inherited from the standard Euclidean metric;
\item for any sequence $(\lambda_n)$ in $[0,+\infty)$, $\lambda_n\to +\infty$ if and only if
$d(\lambda_n,+\infty)\to 0$.
\end{enumerate}
For example, the following metric satisfies (i)--(iii):
$d(\lambda,\lambda'):=\big|\arctan(\lambda)-\arctan(\lambda')\big|,$ where $\arctan(+\infty):=\frac{\pi}{2}.$

\medskip
\noindent\textbf{Continuous extension.}
For $q\in\Delta(Y)$ and $\lambda>0$, recall \eqref{eq:dynamic:gq}. Extend $g_q(\cdot;\lambda)$ to
$\lambda\in\overline{\R}_+$:
\begin{align}\label{eq:dynamic:infty:gq-ext}
g_q(x;0):=\sum_{y\in Y}q(y)x(y),\qquad
g_q(x;+\infty):=\min_{y\in Y}x(y).
\end{align}
Define the one-step ARC aggregator for innovation by the same formula as \eqref{eq:dynamic:G} for
$\lambda>0$, and extend it to $\lambda\in\overline{\R}_+$ via \eqref{eq:dynamic:infty:gq-ext}:
$\mathcal G(x;\mu,\lambda):=\sum_{q\in Q^{2}}\mu(q)\hspace{0.02in}g_q(x;\lambda)$, for $(\mu,\lambda)\in\Delta(Q^{2})\times\overline{\R}_+.$
Define also the action-contingent one-step operators (with costs $c^1=0$ and $c^2=k$ as in
Section \ref{subsec:dynamic:env}) by $\mathcal G^1(x;\mu,\lambda):=g_{q^1}(x;\lambda)$ and $\mathcal G^2(x;\mu,\lambda):=\mathcal G(x;\mu,\lambda)$ in \eqref{eq:dynamic:G}.

\medskip
\noindent\textbf{Infinite-horizon recursion.}
Let $S_1=\varnothing$ and $S_t:=Y^{t-1}$ for $t\ge 2$ denote public (observable) histories, and let
$H_1=\varnothing$ and $H_t:=(\{1,2\}\times Y)^{t-1}$ for $t\ge 2$ denote private (agent) histories.
For each $t\ge 1$, let $\chi_t:H_t\to S_t$ be the public projection $\chi_t(h_t)=(y_1,\dots,y_{t-1})$. Then, a (utility-space) infinite-horizon dynamic contract is a bounded sequence $(x_t)_{t\ge 1}$ with maps
$$
x_t:S_t\times Y\to\R,
\qquad
\sup_{t\ge 1}\sup_{s_t\in S_t}\sup_{y\in Y}|x_t(s_t,y)|<\infty.
$$
Thus, transfers depend only on the public outcome history, as required by moral hazard.

Given a history-dependent \emph{private} state $(\mu_t(\cdot|h_t),\lambda_t(h_t))\in\Delta(Q^{2})\times\overline{\R}_+$,
define continuation values $(W_t)_{t\ge 1}$ on private histories by
\begin{align}\label{eq:dynamic:infty:W}
W_t(h_t)
=
\max_{a\in\{1,2\}}
\left\{
\mathcal G^a \Big(y\mapsto x_t(\chi_t(h_t),y)+\beta\hspace{0.02in}W_{t+1}(h_t,a,y) ; \mu_t(\cdot|h_t),\lambda_t(h_t)\Big)
-c^a
\right\},
\end{align}
where $(h_t,a,y)\in H_{t+1}$ denotes the successor private history. When we say $(W_t)_{t\ge 1}$ is bounded, we mean that  $\text{sup}_{t\geq1}\text{sup}_{h_t\in H_t}|W_t(h_t)|<\infty$.

\medskip
\noindent\textbf{Finite-horizon continuation contingencies.}
Fix a realized private history sequence $(h_t)_{t\ge 1}$ with $h_1=\varnothing$ and
$h_{t+1}=(h_t,a_t,y_t)\in H_{t+1}$ for each $t\ge 1$. For $t\ge 1$ and an integer $N\ge 0$, define
the set of \emph{$N$-period continuation contingencies} from $h_t$ by
\begin{align}\label{eq:dynamic:infty:contingencies}
\mathcal H_{t,N}(h_t)
:=
\big\{\tilde h_{t+j}\in H_{t+j}: 0\le j\le N \text{ and } \tilde h_{t+j} \text{ extends } h_t\big\}.
\end{align}
Since $\{1,2\}\times Y$ is finite, $\mathcal H_{t,N}(h_t)$ is a finite set for each $(t,N)$.

\medskip
\noindent\textbf{local uniform convergence of states.}
Assume there exists a pair $(\mu_\infty,\lambda_\infty)\in\Delta(Q^{2})\times\overline{\R}_+$ such that for every
fixed $N\ge 0$, as $t\to\infty$,
\begin{align}\label{eq:dynamic:infty:state-unif}
\sup_{\tilde h\in\mathcal H_{t,N}(h_t)}\|\mu_{|\tilde h|}(\tilde h)-\mu_\infty\|_1\to 0,
\qquad
\sup_{\tilde h\in\mathcal H_{t,N}(h_t)} d \big(\lambda_{|\tilde h|}(\tilde h),\lambda_\infty\big)\to 0.
\end{align}
Define $\mathcal G^{a,\infty}(x):=\mathcal G^a(x;\mu_\infty,\lambda_\infty)$. We are now ready to state the main result of this appendix. Recall that the constant $\beta\in(0,1]$ denotes the agent's discount rate.

\begin{proposition}\label{thm:dynamic:infty:bridge}
Assume $\beta\in(0,1)$ and $\sup_{t,s_t,y}|x_t(s_t,y)|<\infty$. Then,
\begin{enumerate}[label=(\alph*),leftmargin=18pt]
\item The recursion \eqref{eq:dynamic:infty:W} admits a unique bounded solution $(W_t)_{t\ge 1}$ on private histories.
\item Under \eqref{eq:dynamic:infty:state-unif}, let $(W_t^\infty)_{t\ge 1}$ be the unique bounded solution of
\begin{align}\label{eq:dynamic:infty:Wlimit}
W_t^\infty(h_t)
=
\max_{a\in\{1,2\}}
\left\{
\mathcal G^{a,\infty} \Big(y\mapsto x_t(\chi_t(h_t),y)+\beta\hspace{0.02in}W_{t+1}^\infty(h_t,a,y)\Big)
-c^a
\right\}.
\end{align}
Then,
$|W_t(h_t)-W_t^\infty(h_t)| \longrightarrow 0$ as $t\to\infty.$
\end{enumerate}
\end{proposition}

\begin{corollary}\label{cor:dynamic:infty:bridge-subseq-tree}
In Proposition \ref{thm:dynamic:infty:bridge}, 
let $(\tau_n)_{n\ge 1}$ be any sequence with $\tau_n\to\infty$.
Suppose there exists a pair
$(\mu_\infty,\lambda_\infty)\in\Delta(Q^{2})\times\overline{\R}_+$
such that, for every fixed $N\ge 0$, as $n\to\infty$,
\begin{equation}\label{eq:dynamic:infty:state-unif-subseq}
\sup_{\tilde h\in\mathcal H_{\tau_n,N}(h_{\tau_n})}
\|\mu_{|\tilde h|}(\tilde h)-\mu_\infty\|_1\to 0,
\qquad
\sup_{\tilde h\in\mathcal H_{\tau_n,N}(h_{\tau_n})}
d \big(\lambda_{|\tilde h|}(\tilde h),\lambda_\infty\big)\to 0.
\end{equation}
Let $(W_t)_{t\ge 1}$ be the unique bounded solution of
\eqref{eq:dynamic:infty:W}, and let $(W_t^\infty)_{t\ge 1}$ be the unique bounded solution of
\eqref{eq:dynamic:infty:Wlimit}.
Then, for every fixed $L\ge 0$,
$$
\sup_{\tilde h\in\mathcal H_{\tau_n,L}(h_{\tau_n})}
\big|W_{|\tilde h|}(\tilde h)-W_{|\tilde h|}^\infty(\tilde h)\big|
\longrightarrow 0
\qquad\text{as }n\to\infty.
$$
\end{corollary}

We present a mapping from the on-path limits in Observations
\ref{prop:dynamic:infty:posterior}--\ref{prop:dynamic:infty:limit-pref} to the local uniform convergence
required by \eqref{eq:dynamic:infty:state-unif}.  This condition asks that, after enough learning, the state is insensitive to any
fixed finite-horizon perturbation of future actions and outcomes. The belief component satisfies this because
posterior odds can change only by a bounded likelihood-ratio factor over $N$ periods (Lemma
\ref{lem:dynamic:infty:mu-local}). The robustness component satisfies this because the statistically sophisticated
$\lambda_t$ is an \emph{average} log-likelihood statistic, and adding $N$ periods changes the normalized LLR by at most
order $(\log t)/t$ uniformly (Lemma \ref{lem:dynamic:infty:llr-local}). Thus, the stronger local uniformity in
\eqref{eq:dynamic:infty:state-unif} follows from the weaker on-path limits in Observations
\ref{prop:dynamic:infty:posterior}--\ref{prop:dynamic:infty:limit-pref} under a natural uniqueness condition on the
innovation environment. Proposition \ref{thm:dynamic:infty:bridge} then implies that, in the long run, the agent evaluates
continuation contracts as if he had a time-invariant operator $\mathcal G^{a,\infty}(\cdot)$.

\medskip
\noindent\textbf{Sequential optimality and equilibrium paths.}
Given a bounded contract $(x_t)_{t\ge 1}$, a strategy is a sequence
$\sigma=(\sigma_t)_{t\ge 1}$ with $\sigma_t:H_t\to\{1,2\}$. We say that $\sigma$ is
\emph{sequentially optimal} if, for the bounded solution $(W_t)_{t\ge 1}$ of
\eqref{eq:dynamic:infty:W}, for every $t\ge 1$ and every $h_t\in H_t$,
$$
\sigma_t(h_t)\in
\arg\max_{a\in\{1,2\}}
\left\{
\mathcal G^a \Big(y\mapsto x_t(\chi_t(h_t),y)+\beta\hspace{0.02in}W_{t+1}(h_t,a,y);
\mu_t(\cdot|h_t),\lambda_t(h_t)\Big)-c^a
\right\}.
$$
An equilibrium path is any realized private-history sequence generated by such a sequentially optimal
strategy together with the true outcome process $p_\ast$.
Since the action set is finite, the argmax is nonempty whenever \eqref{eq:dynamic:infty:W} has a bounded
solution. Thus, Proposition \ref{thm:dynamic:infty:bridge}.(a) implies that every bounded contract admits
at least one sequentially optimal strategy and hence at least one induced equilibrium path.

\subsection{Convergence of States}\label{app:dynamic:verification}
 \eqref{eq:dynamic:infty:state-unif} requires that 
$(\mu_t(\cdot|h_t),\lambda_t(h_t))$ stabilizes \emph{uniformly} over all $N$-period continuation contingencies
$\tilde h\in\mathcal H_{t,N}(h_t)$, for every fixed $N$. This is stronger than the on-path convergence
statements in Observations \ref{prop:dynamic:infty:posterior} and  \ref{prop:dynamic:infty:limit-pref}. We show that in the innovation environment, this uniformity is
a natural consequence of (i) posterior concentration and (ii) the fact that both Bayesian odds and the
normalized log-likelihood ratio move only by a bounded amount over any fixed finite horizon. All results below are used to prove Theorem \ref{thm:cycles:frontier}.

\medskip
\noindent\textbf{Step 1: Bayesian posteriors are locally stable.}
As maintained throughout, the routine action has a singleton structured benchmark, so learning about the roadmap concerns only the innovation component $Q^{2}\subset\Delta(Y)$, and periods with $a_\tau=1$ leave relative likelihoods across $q\in Q^{2}$ unchanged. Maintain full support on $Y$ within $Q^{2}$, so
$q(y)>0$ for all $q\in Q^{2}$ and $y\in Y$. Since $Q^{2}$ and $Y$ are finite, define the finite constant
\begin{align}\label{eq:app:infty:K}
K:= \max_{q,q'\in Q^{2}}\max_{y\in Y} \frac{q(y)}{q'(y)} \in [1,+\infty).
\end{align}

\begin{lemma}\label{lem:dynamic:infty:mu-local}
Fix any $q_\star\in Q^{2}$ and any private history $h_t\in H_t$ such that $\mu_t(q_\star|h_t)>0$.
Then, for every integer $N\ge 0$ and every continuation contingency $\tilde h\in\mathcal H_{t,N}(h_t)$,
\begin{align}\label{eq:dynamic:infty:mu-local}
1-\mu_{|\tilde h|}(q_\star|\tilde h)
 \le
K^{N}\frac{1-\mu_t(q_\star|h_t)}{\mu_t(q_\star|h_t)}.
\end{align}
Consequently, $\big\|\mu_{|\tilde h|}(\cdot|\tilde h)-\delta_{\{q_\star\}}\big\|_1
=
2\big(1-\mu_{|\tilde h|}(q_\star|\tilde h)\big)
\le
2K^{N}\frac{1-\mu_t(q_\star|h_t)}{\mu_t(q_\star|h_t)}.$
\end{lemma}

\medskip
\noindent\textbf{Step 2: LLR updates are locally stable.}
Recall innovation-roadmap LLR  \eqref{eq:dynamic:infty:LLR2}. Let
$$
M^2_t(h_t):=\max_{q\in Q^{2}}\prod_{\tau=1}^{t-1}q(y_\tau)^{\1\{a_\tau=2\}},
\quad
U^2_t(h_t):=\max_{p\in\Delta(Y)}\prod_{\tau=1}^{t-1}p(y_\tau)^{\1\{a_\tau=2\}},
\quad
\mathrm{LLR}^2(h_t,Q^{2})=\log\frac{U^2_t(h_t)}{M^2_t(h_t)}.
$$
Let
\begin{align}\label{eq:app:infty:qunderline}
\underline q\ :=\ \min_{q\in Q^{2}}\min_{y\in Y}q(y)\in(0,1),
\end{align}
which is strictly positive by finiteness and full support.

\begin{lemma}\label{lem:dynamic:infty:llr-local}
Fix $N\ge 0$. There exists a constant $B<\infty$ (depending only on $N$ and $\underline q$) such that for every
$t\ge 2$, every $h_t\in H_t$ and every $\tilde h\in\mathcal H_{t,N}(h_t)$,
\begin{align}\label{eq:dynamic:infty:llr-local}
\left|
\frac{\mathrm{LLR}^2(\tilde h,Q^{2})}{|\tilde h|-1}
-
\frac{\mathrm{LLR}^2(h_t,Q^{2})}{t-1}
\right|
 \le
B\frac{1+\log(t+N)}{t}.
\end{align}
In particular, for each fixed $N$, the right-hand side converges to $0$ as $t\to\infty$ uniformly over
$\tilde h\in\mathcal H_{t,N}(h_t)$.
\end{lemma}

\begin{lemma}\label{lem:cycles:Wlimit-public}
Let $(W_t^\infty)_{t\ge 1}$ be the unique bounded solution of \eqref{eq:dynamic:infty:Wlimit}. If
$\chi_t(h_t)=\chi_t(h_t')$, then $W_t^\infty(h_t)=W_t^\infty(h_t')$. In particular,
$W_{t+1}^\infty(h_t,a,y)$ is independent of $a$.
\end{lemma}

\subsection{Innovation-roadmap LLR}\label{subsec:app:infty:innovation-llr}
We now clarify how the innovation-roadmap statistic $\mathrm{LLR}^2$ in
\eqref{eq:dynamic:infty:LLR2} is obtained from the average LLR rule in
\citet[][eq. (2)]{lanzani2025}. Appendix \ref{subsec:app:infty:limit} defines the global
$\mathrm{LLR}(h_t,Q)$ in \eqref{eq:dynamic:infty:LLR} for environments where every action component
of the global structured model may be subject to misspecification. The baseline innovation environment
imposes a more specific structure: the routine method induces a familiar benchmark model $q^1$, while the innovation roadmap
$Q^2$ is the new method whose fit is being evaluated over time. Formally, specialize the global structured family and the admissible
unstructured family to
$\mathcal Q
:=
\{(q^1,q):q\in Q^2\}\subset\Delta(Y)^A,$ $\mathcal N
:=
\{(q^1,p):p\in\Delta(Y)\}\subset\Delta(Y)^A .$
Thus, both the structured and unstructured families maintain the same routine component $q^1$, while the
innovation component is tested against the unrestricted simplex $\Delta(Y)$, as in \citet[][Section 2.2]{lanzani2025}
for finite outcome spaces. This formulation captures the idea that routine outcomes are not diagnostic about
misspecification of innovation models $Q^2$.

Fix a private history $h_t=(a_1,y_1,\dots,a_{t-1},y_{t-1})$. Let
$I_a(h_t):=\{\tau<t:a_\tau=a\}$ for $a\in\{1,2\}$ and define the routine likelihood term
$R_t(h_t):=\prod_{\tau\in I_1(h_t)}q^1(y_\tau)$.
For $\bar q=(q^1,q)\in\mathcal Q$, the likelihood of $h_t$ is
$R_t(h_t)\prod_{\tau\in I_2(h_t)}q(y_\tau)$; for
$\bar p=(q^1,p)\in\mathcal N$, it is
$R_t(h_t)\prod_{\tau\in I_2(h_t)}p(y_\tau)$. Applying the LLR in
\citet[][eq. (2)]{lanzani2025} to $(\mathcal Q,\mathcal N)$ gives
\begin{align*}
-\log\left(
\frac{
\max_{\bar q\in\mathcal Q}
\prod_{\tau=1}^{t-1}\bar q^{a_\tau}(y_\tau)}
{
\max_{\bar p\in\mathcal N}
\prod_{\tau=1}^{t-1}\bar p^{a_\tau}(y_\tau)}
\right)
&=
-\log\left(
\frac{
R_t(h_t)\max_{q\in Q^2}\prod_{\tau\in I_2(h_t)}q(y_\tau)}
{
R_t(h_t)\max_{p\in\Delta(Y)}\prod_{\tau\in I_2(h_t)}p(y_\tau)}
\right)\\
&=
-\log\left(
\frac{
\max_{q\in Q^2}\prod_{\tau\in I_2(h_t)}q(y_\tau)}
{
\max_{p\in\Delta(Y)}\prod_{\tau\in I_2(h_t)}p(y_\tau)}
\right)\\
&=
-\log\left(
\frac{
\max_{q\in Q^2}\prod_{\tau=1}^{t-1}q(y_\tau)^{\1\{a_\tau=2\}}}
{
\max_{p\in\Delta(Y)}\prod_{\tau=1}^{t-1}p(y_\tau)^{\1\{a_\tau=2\}}}
\right)\\
&=\mathrm{LLR}^2(h_t,Q^2),
\end{align*}
which is \eqref{eq:dynamic:infty:LLR2}. The first equality separates routine and innovation
dates and uses the fact that $R_t(h_t)$ is independent of $q$ and $p$; the second cancels this
common positive routine likelihood term; the third rewrites the innovation-date products as
full-history products.

\section{Appendix: Outcome and Action Expansion}\label{app:dynamic:extensions}
We explore how enlarging the outcome and action spaces affect our impossibility result. 

\subsection{Diagnostic Milestones}\label{app:dynamic:milestones}

This appendix complements Theorem \ref{thm:dynamic:trap-scaleup} by clarifying when a richer, \emph{verifiable}
performance taxonomy can prevent breakthrough traps. 
We provide two insights:
(i) adding more outcomes does not necessarily prevent breakthrough traps;
(ii) adding a \emph{diagnostic} milestone that is substantially more likely under routine execution than under innovation
can restore implementability even when $\lambda_2(1)$ is large.
This analysis identifies when measurement design---what can be contracted on, and how the roadmap
models differ on that statistic---can complement the tools in Section \ref{sec:avoidtrap} to help avoid breakthrough traps.

\medskip
\noindent\textbf{Three outcomes.}
We modify Section \ref{subsec:dynamic:env} by letting the contractible outcome be
$Y=\{1,0,d\}$, where $1$ remains the breakthrough event, $0$ is ordinary failure, and
$d$ is a verifiable milestone that is much more likely when the agent pursues incremental routine tasks (e.g., a standardized benchmark pass or a benchmark prototype).
\par The routine action $a=1$ is described by a model $q^1\in\Delta(Y)$ with
$q^1(1)=p$ and $q^1(d)=\zeta$, and innovation $a=2$ is described by $Q^{2}=\{q^2_L,q^2_H\}\subset\Delta(Y)$ with
$q_i(1)=\theta_i$ and $q_i(d)=\varepsilon_i$ for $i\in\{L,H\}$, where $0<\theta_L<p<\theta_H<1$.
All remaining primitives and the definition of $g_q(\cdot;\lambda)$ in \eqref{eq:dynamic:gq} are unchanged such as full support assumptions, and we continue to maintain the LLR updating rule in Assumption \ref{ass:dynamic:llr} for the misspecification parameter $\lambda$. We first show that adding some arbitrary outcome in the framework of Section \ref{subsec:dynamic:env} may not affect the existence of breakthrough traps. 
\begin{obs}\label{prop:dynamic:milestones-irrelevance}
Suppose that conditional on $\{y\neq 1\}$, the relative frequencies of $d$ and $0$ are the same across all models:
there exists $\eta\in(0,1)$ such that
$q^1(d)=(1-p)\eta, q^1(0)=(1-p)(1-\eta),$
and for each $q\in Q^{2}$,
$q(d)=(1-q(1))\eta, q(0)=(1-q(1))(1-\eta).$
Then, for every $(\mu,\lambda)$, $C(\mu,\lambda)$ coincides with the binary-outcome capacity with outcomes $\{1,\neg 1\}$. Hence, adding outcome $d$ does not relax
breakthrough traps.
\end{obs}

The idea is that to affect the prevalence of breakthrough traps, the additional outcome has to  significantly affect the likelihood ratio of models. This is shown below.

\begin{obs}\label{prop:dynamic:milestones-diagnostic}
Fix the post-success state $(\mu_{\mathfrak{s}},\lambda_{\mathfrak{s}}):=(\mu_2(\cdot|y_1=1),\lambda_2(1))$.
If
\begin{align}\label{eq:dynamic:milestones-condition}
\sum_{q\in Q^{2}}\mu_{\mathfrak{s}}(q)\log\Big(\frac{q^1(d)}{q(d)}\Big)>\lambda_{\mathfrak{s}} k,
\end{align}
then $C(\mu_{\mathfrak{s}},\lambda_{\mathfrak{s}})\ge k$ and innovation is implementable after success.
In particular, if $\mu_{\mathfrak{s}}(q^2_H)\to 1$, then a sufficient limiting condition is
$\log(\zeta/\varepsilon_H)>\lambda_{\mathfrak{s}}k$. More explicitly, for any strict slack $\log(\zeta/\varepsilon_H)>\lambda_{\mathfrak{s}}k+\eta$ with $\eta>0$, there exists $\bar\delta>0$ such that the sufficient condition in \eqref{eq:dynamic:milestones-condition} holds whenever $\mu_{\mathfrak{s}}(q^2_H)\ge 1-\bar\delta$.
\end{obs}

To summarize, Observation \ref{prop:dynamic:milestones-irrelevance} formalizes that ``more outcomes'' is irrelevant unless the refinement
adds \emph{action-diagnostic} content: if models agree on the relative composition of non-breakthrough outcomes, then the
agent compresses them into a single certainty equivalent and the breakthrough trap is unchanged.
Observation \ref{prop:dynamic:milestones-diagnostic} shows that a verifiable milestone $d$ that is markedly more likely under
routine execution than under innovative models (large likelihood ratio $q^1(d)/q(d)$) enlarges the post-success
incentive capacity, thereby restoring implementability when $\lambda_2(1)$ is high.
Thus, beyond the tools in Section \ref{sec:avoidtrap} that act by shaping the evolution of $\lambda_t(\cdot)$,
\emph{measurement/model} design can also matter: identifying and contracting on diagnostic milestones can prevent the
post-success collapse of implementable innovation incentives that drives our impossibility result (Theorem \ref{thm:dynamic:trap-scaleup}).

\subsection{Shirking vs. Innovation}\label{app:dynamic:shirking}

This appendix extends Section \ref{sec:dynamic} by introducing a third action,
\emph{shirking} $a=0$, alongside exploitation $a=1$ and exploration $a=2$. The purpose is to clarify why shirking
\emph{reinforces} these results by tightening implementability constraints.

\medskip
\noindent\textbf{Primitives.}
Maintain all primitives from Section \ref{sec:dynamic} unless explicitly modified here.
The action set becomes $A=\{0,1,2\}$. At each date $t\in\{1,2\}$, conditional on action $a_t=a$,
output $y_t\in Y=\{0,1\}$ is drawn from:
\begin{itemize}
  \item \textit{Shirking} $a=0$: a known Bernoulli distribution $q^0\in\Delta(Y)$ with $q^0(1)=p_0\in(0,1)$.
  \item \textit{Exploitation} $a=1$: the known ``routine'' Bernoulli distribution $q^1\in\Delta(Y)$ from the main text
        (so $q^1(1)=p\in(0,1)$).
  \item \textit{Innovation} $a=2$: the ambiguous distribution as in Section \ref{subsec:dynamic:env}, with
        $Q^{2}=\{q^2_L,q^2_H\}$ and belief $\mu_t(\cdot|h_t)\in\Delta(Q^{2})$ updated by Bayes rule as in the main text.
\end{itemize}
Costs are $c^0=0$, $c^1=k_1>0$, and $c^2=k>0$,
where $k_1$ captures the
idea that even ``routine'' execution requires costly effort, whereas shirking avoids such costs. Let $k\geq k_1$ to capture that innovation is at least as costly as the routine action.

\medskip
\noindent\textbf{Robust one-step evaluation at $t=2$.}
Fix a continuation utility vector $x:Y\to\R$, a robustness parameter $\lambda>0$, and a posterior
$\mu\in\Delta(Q^{2})$. Recall the multiplier criterion $g_q(\cdot;\lambda)$ in
\eqref{eq:dynamic:gq} and the ARC aggregator $\mathcal G(\cdot;\mu,\lambda)$ in \eqref{eq:dynamic:G}.
At $t=2$ there is no continuation beyond the current outcome, so the agent evaluates each action as:
\begin{align}\label{eq:app:shirk:V2-actions}
V_2(\mu,\lambda; a, x)=
\begin{cases}
g_{q^0}(x;\lambda) & \text{if } a=0,\\
g_{q^1}(x;\lambda)-k_1 & \text{if } a=1,\\
\mathcal G(x;\mu,\lambda)-k & \text{if } a=2.
\end{cases}
\end{align}

\medskip
\noindent\textbf{Pre-cost advantages.}
Define, for any $(\mu,\lambda)$ and any $x:Y\to\R$,
\begin{align}\label{eq:app:shirk:Delta-a}
M^{0}(\mu,\lambda;x):=\mathcal G(x;\mu,\lambda)-g_{q^0}(x;\lambda),
\qquad
M^{1}(\mu,\lambda;x):=\mathcal G(x;\mu,\lambda)-g_{q^1}(x;\lambda).
\end{align}
Note that $M^{1}$ coincides with $M$ in \eqref{eq:dynamic:Delta-C}. Here,
$M^{0}$ and $M^{1}$ measure the \emph{robust} (pre-cost) value of innovation relative to,
respectively, shirking and exploitation, given the wage/utility schedule $x$ and local robustness $\lambda$.

\medskip
\noindent\textbf{Terminal implementability of continued innovation.}
At a given $(\mu,\lambda)$ and contract-induced $x$, innovation $a=2$ is optimal at $t=2$ if and only if it
weakly dominates \emph{both} non-innovative actions $a\in\{0,1\}$. Using \eqref{eq:app:shirk:V2-actions},
this is equivalent to the conjunction
\begin{align}\label{eq:app:shirk:two-IC}
M^{0}(\mu,\lambda;x) \ge k
\qquad\text{and}\qquad
M^{1}(\mu,\lambda;x) \ge k-k_1.
\end{align}
It is convenient to combine these into a single ``slack'' functional:
\begin{align}\label{eq:app:shirk:Delta-underline}
\underline{M}(\mu,\lambda;x)
:=
\min\Big\{M^{0}(\mu,\lambda;x)-k, M^{1}(\mu,\lambda;x)-(k-k_1)\Big\}.
\end{align}
Thus, $\underline{M}(\mu,\lambda;x)\ge 0$ if and only if \eqref{eq:app:shirk:two-IC} holds.

\begin{obs}[Terminal implementability with shirking]\label{prop:app:shirk:terminal-impl}
Fix any dynamic contract $(x_1,x_2)$. At $t=2$, conditional on $y_1$, let the state be
$(\mu_2(\cdot|y_1),\lambda_2(y_1))$ as in Section \ref{sec:dynamic}. Then, $a_2=2$ is (weakly) optimal after history
$y_1$ if and only if
\begin{align}\label{eq:app:shirk:terminal-criterion}
\underline{M} \big(\mu_2(\cdot|y_1),\lambda_2(y_1); x_2(y_1,\cdot)\big) \ge 0.
\end{align}
Define the \emph{shirking-adjusted incentive capacity} at $(\mu,\lambda)$ by
\begin{align}\label{eq:app:shirk:capacity}
\underline C(\mu,\lambda):=\sup_{x:Y\to\R} \underline{M}(\mu,\lambda;x).
\end{align}
Then, $a_2=2$ is implementable after history $y_1$ if and only if  either $\underline C(\mu_2(\cdot|y_1),\lambda_2(y_1))>0$, or
$\underline C(\mu_2(\cdot|y_1),\lambda_2(y_1))=0$ and the supremum defining $\underline C(\cdot,\cdot)$ in
\eqref{eq:app:shirk:capacity} is attained.
\end{obs}

\medskip
\noindent\textbf{Shirking tightens constraints.}
We compare the two-action problem that keeps the modified routine cost $c^1=k_1$ but deletes the shirking action. In that comparison, the no-shirking feasible set is
$\mathcal X^{\{1,2\}}_2(\mu,\lambda)
=
\big\{x: M^{1}(\mu,\lambda;x)\ge k-k_1\big\},$
whereas the three-action feasible set is
$\mathcal X^{\{0,1,2\}}_2(\mu,\lambda)
=
\big\{x: M^{0}(\mu,\lambda;x)\ge k \text{ and } M^{1}(\mu,\lambda;x)\ge k-k_1\big\}.$
Therefore,
$\mathcal X^{\{0,1,2\}}_2(\mu,\lambda)\subseteq \mathcal X^{\{1,2\}}_2(\mu,\lambda)$ $\forall(\mu,\lambda).$
Thus, relative to the same-cost two-action benchmark, adding shirking can only \textit{tighten} the implementability requirement for innovation.

\section{Appendix: Omitted Proofs}\label{sec:appproof}

\subsection{Proof of Lemmas \ref{lem:translation}, \ref{lem:F-lambda}, and \ref{lem:app:screening:decomposition}}
\begin{proof}[Proof of Lemma \ref{lem:translation}]
Fix $q\in Q^{2}$. For any $x$ and $\zeta$,
we have $\sum_{y\in Y} q(y)e^{-\lambda(x(y)-\zeta)}
= e^{\lambda\zeta} \sum_{y\in Y} q(y)e^{-\lambda x(y)}.$
Taking logarithms and multiplying by $-1/\lambda$ gives the expression
$-\frac{1}{\lambda}\log \big(\sum_{y\in Y} q(y)e^{-\lambda(x(y)-\zeta)}\big)
=
-\zeta
-\frac{1}{\lambda}\log \big(\sum_{y\in Y} q(y)e^{-\lambda x(y)}\big).$
Taking the expectation over $q$ with respect to $\mu$ yields
$\mathcal{G}(T_\zeta(x);\mu,\lambda)=\mathcal{G}(x;\mu,\lambda)-\zeta$.
\end{proof}

\begin{proof}[Proof of Lemma \ref{lem:F-lambda}]
Fix $q\in Q^{2}$ and recall
$g_q(x;\lambda)=
-\frac{1}{\lambda}\log\big(\sum_{y\in Y} q(y)e^{-\lambda x(y)}\big).$
Define the exponentially tilted probability (the ``worst-case distortion'' inside model $q$):
\begin{align}
p_{q,\lambda}(y;x):=
\frac{q(y)e^{-\lambda x(y)}}{\sum_{z\in Y}q(z)e^{-\lambda x(z)}}.
\label{eq:pqlambda}
\end{align}
Then,
$\frac{\partial g_q(x;\lambda)}{\partial \lambda}
=
\frac{1}{\lambda^2}\log\big(\sum_{y}q(y)e^{-\lambda x(y)}\big)
+\frac{1}{\lambda}\sum_{y}p_{q,\lambda}(y;x)x(y).$
Using the identity
\begin{align*}
\log\Big(\sum_y q(y)e^{-\lambda x(y)}\Big)
&=
\sum_y p_{q,\lambda}(y;x)\log\frac{q(y)e^{-\lambda x(y)}}{p_{q,\lambda}(y;x)}\\
&=
-\lambda\sum_y p_{q,\lambda}(y;x)x(y)
+\sum_y p_{q,\lambda}(y;x)\log\frac{q(y)}{p_{q,\lambda}(y;x)},
\end{align*}
we obtain $\frac{\partial g_q(x;\lambda)}{\partial \lambda}
=
-\frac{1}{\lambda^2}\sum_{y\in Y} p_{q,\lambda}(y;x)\log\frac{p_{q,\lambda}(y;x)}{q(y)}
=
-\frac{1}{\lambda^2}\DKL\big(p_{q,\lambda}(\cdot;x)\|q\big).$ Thus, we get
$\partial_\lambda g_q(x;\lambda)\le 0$ for all $\lambda>0$, so $g_q(x;\lambda)$ is weakly decreasing in
$\lambda$.

If $q$ has full support, $\DKL(p_{q,\lambda}\|q)=0$ if and only if $p_{q,\lambda}(\cdot;x)=q(\cdot)$, which
by \eqref{eq:pqlambda} holds if and only if $x$ is constant on $Y$. Hence, for non-constant $x$,
$\partial_\lambda g_q(x;\lambda)<0$ and $g_q$ is strictly decreasing. Finally,
$\mathcal{G}(x;\mu,\lambda)=\E_{\mu}[g_q(x;\lambda)],$
so the same monotonicity (and strictness under the stated conditions) holds for $\mathcal{G}$.
\end{proof}
\begin{proof}[Proof of Lemma \ref{lem:app:screening:decomposition}]
Fix $\mu,\lambda,z$ as in the statement and write $\pi^0(q,y)=\mu(q)q(y)$. Consider any
$\pi\in\Delta(Q^{2}\times Y)$ satisfying the marginal restriction $\pi_{Q^{2}}=\mu$. Let $\text{supp}(\mu):=\{q\in Q^{2}:\mu(q)>0\}$. The marginal restriction implies that for each
$q\in Q^{2}$ we have $\sum_{y\in Y}\pi(q,y)=\mu(q)$. In particular, if $\mu(q)=0$ then necessarily
$\pi(q,y)=0$ for all $y\in Y$. For each $q\in\text{supp}(\mu)$, define the conditional distribution $r_q(y):=\frac{\pi(q,y)}{\mu(q)}$  for all $y\in Y.$
Then, $r_q\in\Delta(Y)$. Conversely, any collection $\{r_q\}_{q\in\text{supp}(\mu)}$ with $r_q\in\Delta(Y)$
induces a feasible $\pi$ by setting $\pi(q,y)=\mu(q)r_q(y)$ for $q\in\text{supp}(\mu)$ and $\pi(q,y)=0$
for $q\notin\text{supp}(\mu)$, which automatically satisfies $\pi_{Q^{2}}=\mu$. In this case, the payoff term decomposes as
$\sum_{(q,y)}\pi(q,y) z(y)
=\sum_{q\in Q^{2}}\mu(q)\sum_{y\in Y} r_q(y) z(y).$
Next, the KL term also decomposes:
$$
\DKL(\pi\|\pi^0)
=\sum_{(q,y)}\pi(q,y)\log\frac{\pi(q,y)}{\pi^0(q,y)}
=\sum_{q\in Q^{2}}\sum_{y\in Y}\mu(q)r_q(y)\log\frac{\mu(q)r_q(y)}{\mu(q)q(y)}
=\sum_{q\in Q^{2}}\mu(q) \DKL(r_q\|q).
$$
Therefore, the minimization problem on the left-hand side of
\eqref{eq:app:screening:decomposition} is equivalent to
$$
\min_{\{r_q\in\Delta(Y)\}_{q\in Q^{2}}}
\sum_{q\in Q^{2}}\mu(q)\Big\{
\sum_{y\in Y}r_q(y) z(y)+\frac{1}{\lambda} \DKL(r_q\|q)
\Big\}.
$$
Since the choice variables $\{r_q\}$ are separable across $q$ and $\mu(q)$ is fixed,
this equals
$$
\sum_{q\in Q^{2}}\mu(q)
\min_{r\in\Delta(Y)}\Big\{\sum_{y\in Y}r(y) z(y)+\frac{1}{\lambda} \DKL(r\|q)\Big\}.
$$
The inner minimum equals $g_q(z;\lambda)$ as
defined in \eqref{eq:dynamic:gq}. This proves \eqref{eq:app:screening:decomposition}.

Finally, \citet[][Proposition 1.4.2]{dupuis97} delivers the minimizer for each $q$:
$r_q^\ast(y)=
\frac{q(y)\exp\{-\lambda z(y)\}}{\sum_{y'\in Y}q(y')\exp\{-\lambda z(y')\}},$
and hence $\pi^\ast(q,y)=\mu(q)r_q^\ast(y)$.
\end{proof}

\subsection{Proofs from Appendix \ref{sec:dynamic:optimal}}

\subsubsection*{Proof of Proposition \ref{prop:dynamic:mmr-recursive}}

\begin{proof}
Fix an arbitrary history-dependent action plan $\psi=(a_1,\sigma_2)$. Since $Y$ and $Q^{2}$ are finite, the simplices $\Delta(Y)$ and $\Delta(Q^{2}\times Y)$ are compact. As noted earlier, $\lambda\equiv0$ is the EU-limit case and the recursive formulation below still works, so we focus on $\lambda>0$. 

\medskip
\noindent\textit{Step 1: reduce the $a=2$ problem.}
Fix $t\in\{1,2\}$ and a history $h_t$.
For $r\in\Delta(Y)$ define
$$
\varPi(r) := \{\pi\in\Delta(Q^{2}\times Y): \pi_Y=r\},
\qquad
\tau_t(h_t,2;r) := \inf_{\pi\in\varPi(r)} \kappa_t(h_t,2;\pi),
$$
where $\pi_Y$ is the $Y$-marginal of $\pi$ and $\kappa_t(h_t,2;\pi)=+\infty$ unless $\pi_{Q^{2}}=\mu_t(\cdot|h_t)$.

\begin{claim}\label{clm:dynamic:tau-attainment}
For every $r\in\Delta(Y)$, the infimum defining $\tau_t(h_t,2;r)$ is attained: there exists $\pi^\ast(r)\in\varPi(r)$ such that
$\tau_t(h_t,2;r)=\kappa_t\big(h_t,2;\pi^\ast(r)\big)$.
\end{claim}
\textit{Proof of Claim \ref{clm:dynamic:tau-attainment}.}
The set $\varPi(r)$ is closed in $\Delta(Q^{2}\times Y)$ and hence compact. Notice that $\varPi(r)$ is nonempty: the product distribution
$\hat{\pi}\in\Delta(Q^{2}\times Y)$ defined by
$\hat{\pi}(q,y):=\mu_t(q|h_t)\hspace{0.02in}r(y)$
satisfies $\hat{\pi}_Y=r$ and $\hat{\pi}_{Q^{2}}=\mu_t(\cdot|h_t)$, so $\hat{\pi}\in\varPi(r)$. Moreover, $\pi\mapsto \kappa_t(h_t,2;\pi)$ is lower semicontinuous on $\Delta(Q^{2}\times Y)$---it is $\lambda_t(h_t)^{-1}\DKL(\pi\|\pi_t^0)$ on the closed affine set $\{\pi:\pi_{Q^{2}}=\mu_t(\cdot|h_t)\}$ and $+\infty$ outside it.
If $\inf_{\pi\in\varPi(r)}\kappa_t(h_t,2;\pi)=+\infty$, then any element of $\varPi(r)$ is a minimizer.
Otherwise, the infimum is finite, and Weierstrass' theorem for lower semicontinuous functions on compact sets yields a minimizer. \hfill$\blacksquare$

Let $z:Y\to\R$. Since $z$ depends only on $y$, for any $\pi$,
$\sum_{(q,y)}\pi(q,y)\hspace{0.02in}z(y)=\sum_{y\in Y}\pi_Y(y)\hspace{0.02in}z(y)$.
Hence, for any $\pi$ with $r:=\pi_Y$,
$\sum_{(q,y)}\pi(q,y)\hspace{0.02in}z(y)+\kappa_t(h_t,2;\pi)
 \ge \sum_{y\in Y} r(y)\hspace{0.02in}z(y)+\tau_t(h_t,2;r).$
Taking $\inf$ over $\pi\in\Delta(Q^{2}\times Y)$ gives
$$
\inf_{\pi\in\Delta(Q^{2}\times Y)}\Big\{\sum_{(q,y)}\pi(q,y)\hspace{0.02in}z(y)+\kappa_t(h_t,2;\pi)\Big\}
 \ge 
\inf_{r\in\Delta(Y)}\Big\{\sum_{y\in Y} r(y)\hspace{0.02in}z(y)+\tau_t(h_t,2;r)\Big\}.
$$
Conversely, fix $r\in\Delta(Y)$ and let $\pi^\ast(r)$ be a minimizer from Claim \ref{clm:dynamic:tau-attainment}. Then
\begin{align*}
\inf_{\pi\in\Delta(Q^{2}\times Y)}\Big\{\sum_{(q,y)}\pi(q,y)\hspace{0.02in}z(y)+\kappa_t(h_t,2;\pi)\Big\}
&\le
\sum_{(q,y)}\pi^\ast(r)(q,y)\hspace{0.02in}z(y)+\kappa_t\big(h_t,2;\pi^\ast(r)\big)\\
&=
\sum_{y\in Y} r(y)\hspace{0.02in}z(y)+\tau_t(h_t,2;r).
\end{align*}
Therefore, $\mathcal V_t(h_t,2;z)
=
\inf_{r\in\Delta(Y)}\big\{\sum_{y\in Y} r(y)\hspace{0.02in}z(y)+\tau_t(h_t,2;r)\big\}-k.$
By Step 2 below, $r\mapsto \tau_t(h_t,2;r)$ is lower semicontinuous on the compact set $\Delta(Y)$, so the infimum is attained and the infimum equals a minimum. This yields \eqref{eq:dynamic:mmr-one-step-innov}.

\medskip
\noindent\textit{Step 2: admissibility of the one-period-ahead indices.}
Define $\tau_t(h_t,1;r):=\kappa_t(h_t,1;r)=\lambda_t(h_t)^{-1}\DKL(r\|q^1)$.
We verify that for each fixed $(t,h_t,a)$ the map $r\mapsto\tau_t(h_t,a;r)$ is grounded, convex, and lower semicontinuous on $\Delta(Y)$. Since $\Delta(Y)$ is finite-dimensional, convexity and lower semicontinuity imply that
$r\mapsto \tau_t(h_t,a;r)$ is closed, which matches the admissibility requirement in
\citet[][Proposition 2(ii)]{mmr06_JET}.

\medskip
\noindent\emph{Case $a=1$.}
The map $r\mapsto\DKL(r\|q^1)$ is convex and lower semicontinuous on $\Delta(Y)$ and satisfies $\DKL(q^1\|q^1)=0$.
Since $\lambda_t(h_t)>0$, the same properties hold for $\tau_t(h_t,1;\cdot)$, which is grounded. 

\medskip
\noindent\emph{Case $a=2$.}

\smallskip
\textit{Groundedness.}
Let $r^0:=(\pi_t^0)_Y$ be the $Y$-marginal of the benchmark $\pi_t^0(q,y)=\mu_t(q|h_t)q(y)$.
Since $(\pi_t^0)_{Q^{2}}=\mu_t(\cdot|h_t)$, we have $\kappa_t(h_t,2;\pi_t^0)=0$ and $\pi_t^0\in\varPi(r^0)$, hence
$\tau_t(h_t,2;r^0)\le 0$.
Since $\kappa_t(h_t,2;\pi)\ge 0$ for all $\pi$, we also have $\tau_t(h_t,2;r)\ge 0$ for all $r$.
Thus, $\tau_t(h_t,2;r^0)=0$. Moreover, since $r^0\in\Delta_{++}(Y):=\{r\in\Delta(Y): r(y)>0 \hspace{0.03in}\forall y\}$ under full support and $\tau_t(h_t,2;r^0)=0<\infty$, we have
$\mathrm{dom}\big(\tau_t(h_t,2;\cdot)\big)\cap\Delta_{++}(Y)\neq\varnothing$, and similarly, $\mathrm{dom}\big(\tau_t(h_t,1;\cdot)\big)\cap\Delta_{++}(Y)\neq\varnothing$ because $q^1\in \Delta_{++}(Y)$ and $\tau_t(h_t,1;q^1)=0$, which verify the domain condition in \citet[][Proposition 2(ii)]{mmr06_JET}.

\smallskip
\textit{Convexity.}
Fix $r_0,r_1\in\Delta(Y)$ and $\theta\in[0,1]$.
If $\tau_t(h_t,2;r_i)=+\infty$ for some $i$, the convexity inequality is immediate.
Otherwise, pick minimizers $\pi_i\in\varPi(r_i)$ with $\kappa_t(h_t,2;\pi_i)=\tau_t(h_t,2;r_i)$.
Then, $\theta\pi_0+(1-\theta)\pi_1\in\varPi(\theta r_0+(1-\theta)r_1)$ and, by convexity of $\kappa_t(h_t,2;\cdot)$,
\begin{align*}
\tau_t\big(h_t,2;\theta r_0+(1-\theta)r_1\big)
\le \kappa_t\big(h_t,2;\theta\pi_0+(1-\theta)\pi_1\big)
&\le \theta\kappa_t(h_t,2;\pi_0)+(1-\theta)\kappa_t(h_t,2;\pi_1)\\
&= \theta\tau_t(h_t,2;r_0)+(1-\theta)\tau_t(h_t,2;r_1).
\end{align*}

\smallskip
\textit{Lower semicontinuity.} Let $r_n\to r$ in $\Delta(Y)$.
If $\liminf_{n\to\infty}\tau_t(h_t,2;r_n)=+\infty$, then $\tau_t(h_t,2;r)\le \liminf_{n\to\infty}\tau_t(h_t,2;r_n)$ holds trivially.
Otherwise, pass to a subsequence (not relabeled) such that $\tau_t(h_t,2;r_n)\to \liminf_{m\to\infty}\tau_t(h_t,2;r_m)<\infty$.
Let $\pi_n\in\varPi(r_n)$ be minimizers, so $\kappa_t(h_t,2;\pi_n)=\tau_t(h_t,2;r_n)<\infty$ for all $n$, hence $(\pi_n)_{Q^{2}}=\mu_t(\cdot|h_t)$ for all $n$.
By compactness of $\Delta(Q^{2}\times Y)$, extract a further subsequence (not relabeled) with $\pi_n\to\pi$. Continuity of marginals gives $\pi_Y=r$.
Moreover, since each $\pi_n$ is a minimizer with $\kappa_t(h_t,2;\pi_n)<\infty$ whenever $\tau_t(h_t,2;r_n)<\infty$,
we have $(\pi_n)_{Q^{2}}=\mu_t(\cdot|h_t)$ for all sufficiently large $n$ with finite value, and hence $\pi_{Q^{2}}=\mu_t(\cdot|h_t)$ by continuity of marginals.
Thus, $\kappa_t(h_t,2;\pi)<\infty$ whenever the liminf is finite.
By lower semicontinuity of $\kappa_t(h_t,2;\cdot)$,
$\kappa_t(h_t,2;\pi) \le \liminf_{n\to\infty}\kappa_t(h_t,2;\pi_n)
=\liminf_{n\to\infty}\tau_t(h_t,2;r_n).$
Since $\tau_t(h_t,2;r)\le \kappa_t(h_t,2;\pi)$ by definition of $\tau_t(h_t,2;r)$ as an infimum over $\varPi(r)$, we obtain $\tau_t(h_t,2;r) \le \liminf_{n\to\infty}\tau_t(h_t,2;r_n),$
establishing lower semicontinuity.

\medskip
\noindent\textit{Step 3: dynamic consistency, backward induction.}
Fix a plan $\psi=(a_1,\sigma_2)$ and the finite event tree: $\Omega=Y\times Y$,
$\mathcal F_1=\{\varnothing,\Omega\}$, $\mathcal F_2=\sigma(y_1)$, and
$\mathcal F_3=\sigma(y_1,y_2)$ (Appendix
\ref{subsec:dynamic:mmr}).

By Step 2, for each decision node $(t,h_t)$, the plan-induced map
$r\mapsto \tau_t\big(h_t,\psi_t(h_t);r\big)$
is grounded, convex, and lower semicontinuous on $\Delta(Y)$, and its domain intersects
$\Delta_{++}(Y)$. Measurability is immediate because the tree is finite: for each fixed
$r\in\Delta(Y)$, $\omega\mapsto \tau_1(\varnothing,a_1;r)$ is
$\mathcal F_1$-measurable, while
$\omega=(y_1,y_2)\mapsto \tau_2(y_1,\sigma_2(y_1);r)$ is
$\mathcal F_2$-measurable under the plan-induced convention stated in Appendix
\ref{subsec:dynamic:mmr}. Thus, the selected family
$\{\tau_t(h_t,\psi_t(h_t);\cdot)\}_{t=1,2}$ satisfies the one-period-ahead
ambiguity-index conditions in \citet[][Proposition 2]{mmr06_JET}. Then, Appendix \ref{subsec:dynamic:mmr} constructs the corresponding dynamic ambiguity
indices $c_3^\psi,c_2^\psi,c_1^\psi$ and the induced recursive values
$W_2^\psi,W_1^\psi$ in
\eqref{eq:dynamic:c3}--\eqref{eq:dynamic:W1-recursive}. These equations are the
specialization of \citet[][Theorem 1, eqs. (10)--(12), and Theorem 2]{mmr06_JET} to the
two-period tree. Therefore, the plan-induced preference is a recursive variational
preference in the sense of \citet[][Theorem 2]{mmr06_JET}. Thus, by \citet[][Theorem 1]{mmr06_JET}, it is
dynamically consistent.

It now remains to connect this plan-induced construction to the maximized values
\eqref{eq:dynamic:V2}--\eqref{eq:dynamic:V1}. Fix $(x_1,x_2)$. If the first-period action is
$b\in\{1,2\}$, define
$$
\mathcal V_2^b(y_1,a;z):=
\begin{cases}
\widehat{\mathcal V}_2(y_1,a;z), & \text{if }b=1,\\
\mathcal V_2(y_1,a;z), & \text{if }b=2,
\end{cases}
\qquad
V_2^b(y_1;x_2):=\max_{a_2\in\{1,2\}}\mathcal V_2^b\big(y_1,a_2;x_2(y_1,\cdot)\big).
$$
Thus, $V_2^1=\widehat V_2$ and $V_2^2=V_2$. For any $\sigma_2$,
$\mathcal V_2^b\big(y_1,\sigma_2(y_1);x_2(y_1,\cdot)\big)\le V_2^b(y_1;x_2)$ $\forall y_1\in Y.$
Each operator $\mathcal V_1(\varnothing,b;\cdot)$ is monotone in its continuation vector, since increasing the continuation vector raises the objective for every feasible distortion and hence raises the minimum. Thus, for each fixed $b$,
$\mathcal V_1\big(\varnothing,b;x_1(\cdot)+\beta V_2^b(\cdot;x_2)\big)
\ge
\mathcal V_1\big(\varnothing,b;x_1(\cdot)+\beta w^\sigma(\cdot)\big),$
where
$w^\sigma(y_1):=\mathcal V_2^b\big(y_1,\sigma_2(y_1);x_2(y_1,\cdot)\big).$
As a result, for each first-period action $b$, the optimal continuation policy is obtained by maximizing node by node at $t=2$. Taking the maximum over $b\in\{1,2\}$ gives exactly the expression in \eqref{eq:dynamic:V1}:
$$
V_1(x_1,x_2)
=
\max\left\{
 \mathcal V_1 \big(\varnothing,1; x_1(\cdot)+ \beta \widehat V_2(\cdot;x_2)\big),
 \mathcal V_1 \big(\varnothing,2; x_1(\cdot)+ \beta V_2(\cdot;x_2)\big)
 \right\}.
$$
Dynamic consistency ensures that these node-by-node maximizers do not generate any commitment/re-optimization discrepancy after $y_1$ is observed. \end{proof}

\subsubsection*{Proof of Proposition \ref{prop:dynamic:terminal-impl}}
\begin{proof}
Fix $\psi=(a_1,\sigma_2)$ and $y_1\in Y$. Let
$\mu:=\mu_2^\psi(\cdot|y_1)$, $\lambda:=\lambda_2^\psi(y_1)$, and
$x:=x_2(y_1,\cdot)$. If $\lambda=0$, use the EU-limit
$g_q(x;0):=\sum_{y\in Y}q(y)x(y)$ and
$\mathcal G(x;\mu,0):=\sum_{q\in Q^2}\mu(q)g_q(x;0)$. For $\lambda>0$, the routine value is
$\min_{r\in\Delta(Y)}
\left\{\sum_{y\in Y}r(y)x(y)+\frac{1}{\lambda}\DKL(r\|q^1)\right\}
=g_{q^1}(x;\lambda)$
by the variational formula \citep[][Proposition 1.4.2]{dupuis97}. The same identity holds at
$\lambda=0$ by the EU-limit convention.

For innovation, let $\pi^0(q,y):=\mu(q)q(y)$. Since the marginal on $Q^2$ is fixed at $\mu$,
any feasible $\pi\in\Delta(Q^2\times Y)$ can be written as
$\pi(q,y)=\mu(q)r_q(y)$ with $r_q\in\Delta(Y)$. The KL decomposition gives
$\DKL(\pi\|\pi^0)=\sum_{q\in Q^2}\mu(q)\DKL(r_q\|q).$
Hence, for $\lambda>0$,
$$
\begin{aligned}
\min_{\substack{\pi\in\Delta(Q^2\times Y):\\ \pi_{Q^2}=\mu}}
&\Big\{\sum_{(q,y)}\pi(q,y)x(y)+\frac{1}{\lambda}\DKL(\pi\|\pi^0)\Big\}-k \\
&\qquad =
\sum_{q\in Q^2}\mu(q)
\min_{r_q\in\Delta(Y)}
\Big\{\sum_{y\in Y}r_q(y)x(y)+\frac{1}{\lambda}\DKL(r_q\|q)\Big\}-k \\
&\qquad =
\sum_{q\in Q^2}\mu(q)g_q(x;\lambda)-k
=
\mathcal G(x;\mu,\lambda)-k.
\end{aligned}
$$
Again, the same identity holds at $\lambda=0$ by the EU-limit convention.

Therefore $a_2=2$ is weakly optimal if and only if
$\mathcal G(x;\mu,\lambda)-k\ge g_{q^1}(x;\lambda),$
which is equivalent to
$M(\mu,\lambda;x)\ge k.$
Substituting back $\mu=\mu_2^\psi(\cdot|y_1)$, $\lambda=\lambda_2^\psi(y_1)$, and
$x=x_2(y_1,\cdot)$ proves the first statement. The implementability characterization follows directly
from the definition
$C(\mu,\lambda)=\sup_{x\in\R^Y}M(\mu,\lambda;x)$: a finite vector satisfying $M\ge k$ exists if and
only if either $C>k$, or $C=k$ and the supremum in \eqref{eq:dynamic:Delta-C} is attained.
\end{proof}

\subsubsection*{Proof of Observation \ref{prop:dynamic:capacity-lb}}
\begin{proof} 
Fix $(\mu,\lambda)$ with $\mu(q^2_H)=m$ and $\mu(q^2_L)=1-m$, and fix $T>0$.
Define $x^T: Y\to\R$ by
$x^T(0)=0$ and $x^T(1)=-T.$
For any Bernoulli distribution $q(\eta)$ with $\eta\in(0,1)$,
$$
g_{q(\eta)}(x^T;\lambda)
=
-\frac{1}{\lambda}\log\big((1-\eta)+\eta e^{\lambda T}\big)
=
-T-\frac{1}{\lambda}\log\big(\eta+(1-\eta)e^{-\lambda T}\big).
$$
Therefore, for any $\eta,\vartheta\in(0,1)$,
$g_{q(\eta)}(x^T;\lambda)-g_{q(\vartheta)}(x^T;\lambda)
=
\frac{1}{\lambda}\log
\frac{\vartheta+(1-\vartheta)e^{-\lambda T}}{\eta+(1-\eta)e^{-\lambda T}}.$
As $T\to\infty$, $e^{-\lambda T}\downarrow 0$, so the right-hand side converges to
$\frac{1}{\lambda}\log(\vartheta/\eta)$. Hence
$$
\lim_{T\to\infty}\Big(g_{q^2_H}(x^T;\lambda)-g_{q^1}(x^T;\lambda)\Big)
=
\frac{1}{\lambda}\log \frac{p}{\theta_H},
\qquad
\lim_{T\to\infty}\Big(g_{q^2_L}(x^T;\lambda)-g_{q^1}(x^T;\lambda)\Big)
=
\frac{1}{\lambda}\log \frac{p}{\theta_L}.
$$
By linearity of $\mathcal G(\cdot;\mu,\lambda)$ in $\mu$,
$\lim_{T\to\infty}{M}(\mu,\lambda;x^T)
=
\frac{1}{\lambda}\big(m\log \frac{p}{\theta_H}+(1-m)\log \frac{p}{\theta_L}\big).$
Since $C(\mu,\lambda)=\sup_x {M}(\mu,\lambda;x)\ge \sup_{T>0}{M}(\mu,\lambda;x^T)$, we obtain
$$
C(\mu,\lambda) \ge \limsup_{T\to\infty}{M}(\mu,\lambda;x^T)
=
\frac{1}{\lambda}\left(m\log \frac{p}{\theta_H}+(1-m)\log \frac{p}{\theta_L}\right),
$$
which is the desired bound.
\end{proof}

\subsubsection*{Proof of Corollary \ref{cor:dynamic:failure-impl}}
\begin{proof}
Let $\mu_{\mathfrak{f}}:=\mu_2(\cdot|y_1=0)$ and $\lambda_{\mathfrak{f}}:=\lambda_2(0)=-\log(1-\theta_L)/\gamma$.
Apply Observation \ref{prop:dynamic:capacity-lb} with $(\mu,\lambda)=(\mu_{\mathfrak{f}},\lambda_{\mathfrak{f}})$:
\begin{align}
C(\mu_{\mathfrak{f}},\lambda_{\mathfrak{f}}) \ge\
\frac{1}{\lambda_{\mathfrak{f}}}\left(m_F(\theta_L)\log \frac{p}{\theta_H}+(1-m_F(\theta_L))\log \frac{p}{\theta_L}\right),
\label{eq:fail:lb}
\end{align}
where $m_F(\theta_L)=\mu_{\mathfrak{f}}(q^2_H)$ is given by \eqref{eq:dynamic:mu-update}.

\medskip
\noindent\emph{A valid lower bound on the bracket.}
Since $\theta_H>p$, we have $\log(p/\theta_H)<0$, and since $m_F(\theta_L)\in[0,1]$, $m_F(\theta_L)\log\frac{p}{\theta_H}\ \ge\ \log\frac{p}{\theta_H}.$
Thus, the bracket in \eqref{eq:fail:lb} satisfies
$$
m_F(\theta_L)\log \frac{p}{\theta_H}+(1-m_F(\theta_L))\log \frac{p}{\theta_L}
\ \ge\
\log\frac{p}{\theta_H}+(1-m_F(\theta_L))\log \frac{p}{\theta_L}.
$$
Moreover, $m_F(\theta_L)\to m_F(0):=\frac{m(1-\theta_H)}{m(1-\theta_H)+(1-m)}<1$ as $\theta_L\downarrow 0$,
so there exist $\tilde{\eta}>0$ and $\bar\theta_L\in(0,p)$ such that for all $\theta_L\in(0,\bar\theta_L)$,
$1-m_F(\theta_L)\ge \tilde{\eta}$. Also, $-\log(1-\theta_L)\le \theta_L/(1-\theta_L)$ implies that for all sufficiently small $\theta_L$,
$\lambda_{\mathfrak{f}} \le 2\theta_L/\gamma$, hence $1/\lambda_{\mathfrak{f}} \ge \gamma/(2\theta_L)$. Finally, since $\log(p/\theta_L)\to+\infty$, shrink $\bar\theta_L$ if needed so that
$\log(p/\theta_L)\ge 2|\log(p/\theta_H)|/\tilde{\eta}$ for all $\theta_L\in(0,\bar\theta_L)$; then
$\log\frac{p}{\theta_H}+\tilde{\eta}\log\frac{p}{\theta_L}\ \ge\ \frac{\tilde{\eta}}{2}\log\frac{p}{\theta_L}.$
Combining these bounds with \eqref{eq:fail:lb}, for $\theta_L\in(0,\bar\theta_L)$, we obtain
$C(\mu_{\mathfrak{f}},\lambda_{\mathfrak{f}})
 \ge \frac{1}{\lambda_{\mathfrak{f}}}\frac{\tilde{\eta}}{2}\log\frac{p}{\theta_L}
 \ge \frac{\gamma}{4\theta_L}\tilde{\eta}\log\frac{p}{\theta_L}
 \xrightarrow[\theta_L\downarrow 0]{}\ +\infty,$
since $\log(1/\theta_L)/\theta_L\to\infty$. This proves $\lim_{\theta_L\downarrow 0}C(\mu_{\mathfrak{f}},\lambda_{\mathfrak{f}})=+\infty$,
and in particular implies that for any fixed $k>0$ we can choose $\theta_L$ small enough so that
$C(\mu_{\mathfrak{f}},\lambda_{\mathfrak{f}})> k$. 
\end{proof}

\subsubsection*{Proof of Proposition \ref{prop:dynamic:first-period-ic}}

\begin{proof}
By \eqref{eq:dynamic:V1}, choosing $a_1=2$ gives $\mathcal V_1(\varnothing,2;x_1(\cdot)+\beta V_2(\cdot;x_2))$, while choosing $a_1=1$ gives $\mathcal V_1(\varnothing,1;x_1(\cdot)+\beta\widehat V_2(\cdot;x_2))$. ARC gives $\mathcal V_1(\varnothing,2;z)=\mathcal G(z;\mu_1,\lambda_1)-k$ and $\mathcal V_1(\varnothing,1;z)=g_{q^1}(z;\lambda_1)$. Substituting the two continuation payoff vectors gives \eqref{eq:dynamic:first-period-IC}.
\end{proof}

\subsection{Proofs from Appendix \ref{sec:general}}

\begin{proof}[Proof of Observation \ref{prop:disc:smooth-bayes-trap}]
Fix any continuation utility vector $x\in\R^Y$ and define $a_q:=g_q^c(x)$, for each $q\in Q^{2}.$
Since $\phi$ is continuous and strictly increasing, the map
$(a_q)_{q\in Q^{2}}
\mapsto
\phi^{-1}\big(
\sum_{q\in Q^{2}}\mu_2(q|y_1=1)\hspace{0.02in}\phi(a_q)
\big)$
is a quasi-arithmetic mean of the numbers $\{a_q\}_{q\in Q^{2}}$. Therefore, it lies between their minimum and maximum, i.e.,
$$
\min_{q\in Q^{2}} a_q
\le
\phi^{-1}\Big(
\sum_{q\in Q^{2}}\mu_2(q|y_1=1)\hspace{0.02in}\phi(a_q)
\Big)
\le
\max_{q\in Q^{2}} a_q.
$$
Substituting back $a_q=g_q^c(x)$ yields
$\mathcal G_c^\phi\big(x;\mu_2(\cdot|y_1=1)\big)
\le
\max_{q\in Q^{2}} g_q^c(x).$
Subtracting the routine-action valuation $g_{q^1}^c(x)$ from both sides gives
$M_c^\phi\big(\mu_2(\cdot|y_1=1);x\big)
\le
\max_{q\in Q^{2}}
\big\{
g_q^c(x)-g_{q^1}^c(x)
\big\}.$
Taking the supremum over $x\in\R^Y$ on both sides,
$$
C_c^\phi\big(\mu_2(\cdot|y_1=1)\big)
=
\sup_{x\in\R^Y}
M_c^\phi\big(\mu_2(\cdot|y_1=1);x\big)
\le
\sup_{x\in\R^Y}
\max_{q\in Q^{2}}
\Big\{
g_q^c(x)-g_{q^1}^c(x)
\Big\}
=
\overline C^{c}.
$$
Hence, $\overline C^{c}<k$ implies
$C_c^\phi\big(\mu_2(\cdot|y_1=1)\big)<k.$ Now, consider any continuation contract after $(a_1=2,y_1=1)$, with continuation utility vector $x$. Under \eqref{eq:disc:smooth:G}, innovation yields utility
$\mathcal G_c^\phi\big(x;\mu_2(\cdot|y_1=1)\big)-k,$
while the routine action yields utility $g_{q^1}^c(x)$. Therefore, inducing $a_2=2$ requires
$M_c^\phi\big(\mu_2(\cdot|y_1=1);x\big)\ge k.$
But the strict inequality
$C_c^\phi\big(\mu_2(\cdot|y_1=1)\big)<k$
means that no continuation utility vector $x$ can satisfy this incentive inequality. Thus, no continuation contract can implement $a_2=2$ after a breakthrough.
\end{proof}

\begin{proof}[Proof of Observation \ref{prop:disc:varrobust-trap}]
Fix any continuation utility vector $x\in\R^Y$ and let $\varepsilon>0$. Since
$\inf_{\nu\in\Delta(Q^{2})}d\big(\nu;\mu_2(\cdot|y_1=1)\big)=0,$
there exists $\nu_\varepsilon\in\Delta(Q^{2})$ such that
$d\big(\nu_\varepsilon;\mu_2(\cdot|y_1=1)\big)\le \varepsilon.$
By the definition of \eqref{eq:disc:varrobust:G},
$\mathcal G_c^d\big(x;\mu_2(\cdot|y_1=1)\big)
\le
\sum_{q\in Q^{2}}\nu_\varepsilon(q)\hspace{0.02in}g_q^c(x)+\varepsilon
\le
\max_{q\in Q^{2}} g_q^c(x)+\varepsilon.$
Letting $\varepsilon\downarrow 0$ yields
$\mathcal G_c^d\big(x;\mu_2(\cdot|y_1=1)\big)
\le
\max_{q\in Q^{2}} g_q^c(x).$
Subtracting the routine-action valuation $g_{q^1}^c(x)$ from both sides gives
$M_c^d\big(\mu_2(\cdot|y_1=1);x\big)
\le
\max_{q\in Q^{2}}
\big\{
 g_q^c(x)-g_{q^1}^c(x)
\big\}.$
Taking the supremum over $x\in\R^Y$ on both sides,
$$
C_c^d\big(\mu_2(\cdot|y_1=1)\big)
=
\sup_{x\in\R^Y}
M_c^d\big(\mu_2(\cdot|y_1=1);x\big)
\le
\sup_{x\in\R^Y}
\max_{q\in Q^{2}}
\Big\{
 g_q^c(x)-g_{q^1}^c(x)
\Big\}
=
\overline C^{c}.
$$
Hence, $\overline C^{c}<k$ implies
$C_c^d\big(\mu_2(\cdot|y_1=1)\big)<k.$
Now, consider any continuation contract after $(a_1=2,y_1=1)$, with continuation utility vector $x$. Under \eqref{eq:disc:varrobust:G}, innovation yields utility
$\mathcal G_c^d\big(x;\mu_2(\cdot|y_1=1)\big)-k,$
while the routine action yields utility $g_{q^1}^c(x)$. Therefore, inducing $a_2=2$ requires
$M_c^d\big(\mu_2(\cdot|y_1=1);x\big)\ge k.$
But the strict inequality
$C_c^d\big(\mu_2(\cdot|y_1=1)\big)<k$
means that no continuation utility vector $x$ can satisfy this incentive inequality. Thus, no continuation contract can implement $a_2=2$ after a breakthrough.
\end{proof}

\begin{proof}[Proof of Observation \ref{prop:disc:generic-update}]
After a success, the best-fitting model in $Q^{2}=\{q^2_L,q^2_H\}$ is $q^2_H$, so
$b(1)=\max_{q\in Q^{2}} q(1)=\theta_H$ and hence $\lambda_2(1)=\Psi(\varLambda(\theta_H)).$
Since $\varLambda$ is continuous and strictly decreasing with $\varLambda(1)=0$, and $p<1$, we have
$\varLambda(p)>0.$
Since $\Psi$ is continuous and strictly positive on $(0,\infty)$,
$\Psi(\varLambda(\theta_H))\rightarrow \Psi(\varLambda(p))>0$ as $\theta_H\downarrow p$. 
Now, define the critical post-success level
$\bar\lambda(\theta_H):=\frac{1}{k}\log\frac{1-p}{1-\theta_H}.$
As $\theta_H\downarrow p$,
$\bar\lambda(\theta_H)\downarrow 0.$
Thus,  there exists $\bar\theta_H\in(p,1)$ such that  $\forall\theta_H\in(p,\bar\theta_H)$,
$\lambda_2(1)=\Psi(\varLambda(\theta_H))>\bar\lambda(\theta_H).$ The remaining argument is the same capacity argument used in the proof of Theorem \ref{thm:dynamic:trap-scaleup}. For fixed $\theta_H\in(p,\bar\theta_H)$, $\lambda_2(1)$ is fixed as $\theta_L\downarrow0$ and satisfies
$\lambda_2(1)>\frac{1}{k}\log\frac{1-p}{1-\theta_H}.$
Moreover, $\mu_2(q_H^2|y_1=1)\to1$ as $\theta_L\downarrow0$. Thus, the post-success capacity converges to the one-model capacity
$\frac{1}{\lambda_2(1)}\log\frac{1-p}{1-\theta_H}<k.$
Therefore, by convergence, there exists $\bar\theta_L\in(0,p)$ such that $\forall \theta_L\in(0,\bar\theta_L)$,
$C(\mu_2(\cdot|y_1=1),\lambda_2(1))<k.$
The implementability conclusion follows from Proposition \ref{prop:dynamic:terminal-impl}.
\end{proof}

\subsection{Proofs from Appendix \ref{sec:disc:roadmap-no-free-lunch}}

\begin{proof}[Proof of Observation \ref{obs:roadmap-known-gamma}]
By Lemma \ref{lem:dynamic:lambda-closed}, $\lambda_2(1)=-\frac{1}{\gamma}\log\theta_H$, while Theorem \ref{thm:dynamic:trap-scaleup} defines $\lambda^\ast=\frac{1}{k}\log\frac{1-p}{1-\theta_H}$. Hence $\lambda_2(1)\leq \lambda^\ast$ is equivalent to
$F(\theta_H)\leq 0$, where
$F(\theta):=-\frac{1}{\gamma}\log\theta-\frac{1}{k}\log\frac{1-p}{1-\theta}$.
Now $F'(\theta)=-\frac{1}{\gamma\theta}-\frac{1}{k(1-\theta)}<0$ on $(p,1)$, while $\lim_{\theta\downarrow p}F(\theta)=-\frac{1}{\gamma}\log p>0$ and $\lim_{\theta\uparrow 1}F(\theta)=-\infty$. So there is a unique cutoff $\hat\theta_H(\gamma)\in(p,1)$ with $F(\hat\theta_H(\gamma))=0$, and the equivalence follows. The formula for $\lambda_2(1)$ gives the second claim.
\end{proof}

\begin{lemma}\label{lem:bernoulli-kl-integral}
For any $a,b\in(0,1)$,
$\DKL(a\|b)
=
a\log\frac{a}{b}+(1-a)\log\frac{1-a}{1-b}
=
\int_a^b \frac{u-a}{u(1-u)}du.$
\end{lemma}

\begin{proof}[Proof of Lemma \ref{lem:bernoulli-kl-integral}]
Fix $a\in(0,1)$ and define
$f(b):=a\log\frac{a}{b}+(1-a)\log\frac{1-a}{1-b},$ for $b\in(0,1).$
Then, $f$ is continuously differentiable on $(0,1)$, with
$f'(b)
=
-\frac{a}{b}+\frac{1-a}{1-b}
=
\frac{b-a}{b(1-b)}.$
Also, $f(a)=0$. Therefore, by the fundamental theorem of calculus,
$f(b)-f(a)=\int_a^b f'(u)\hspace{0.02in}du,$
so
$\DKL(a\|b)=f(b)=\int_a^b \frac{u-a}{u(1-u)}du.$
\end{proof}
\begin{proof}[Proof of Observation \ref{obs:roadmap-no-free-lunch-alpha}]
The first claim is immediate from $\lambda_2(1)=-\frac{1}{\gamma}\log\theta_H$.

For the second claim, write
$L(\theta):=\log\frac{1-p}{1-\theta}$ and
$D(\theta):=\DKL(\theta_\ast\|\theta)$ for $\theta\in(\theta_\ast,1)$, so that
$\alpha_\ast(\theta)=\frac{\gamma}{k}\frac{L(\theta)}{D(\theta)}$.
Since $D'(\theta)=\frac{\theta-\theta_\ast}{\theta(1-\theta)}$, we have
$\frac{d}{d\theta}\log \alpha_\ast(\theta)
=
\frac{1}{(1-\theta)L(\theta)}
-
\frac{\theta-\theta_\ast}{\theta(1-\theta)D(\theta)}.$
It is therefore enough to show
$D(\theta)<\frac{\theta-\theta_\ast}{\theta}L(\theta)$.
By Lemma \ref{lem:bernoulli-kl-integral},
$D(\theta)=\int_{\theta_\ast}^{\theta}\frac{u-\theta_\ast}{u(1-u)}du.$
Since $u\mapsto \frac{u-\theta_\ast}{u}=1-\frac{\theta_\ast}{u}$ is increasing on $(\theta_\ast,1)$,
$$
D(\theta)\leq \frac{\theta-\theta_\ast}{\theta}\int_{\theta_\ast}^{\theta}\frac{du}{1-u}
=
\frac{\theta-\theta_\ast}{\theta}\log\frac{1-\theta_\ast}{1-\theta}
<
\frac{\theta-\theta_\ast}{\theta}\log\frac{1-p}{1-\theta}
=
\frac{\theta-\theta_\ast}{\theta}L(\theta),
$$ strict inequality uses $p<\theta_\ast$. Thus, $\frac{d}{d\theta}\log \alpha_\ast(\theta)<0$, so $\alpha_\ast$ is strictly decreasing in $\theta_H$.
\end{proof}

\begin{proof}[Proof of Observation \ref{obs:roadmap-unknown-gamma}]
Since $\theta_H<1$, Lemma \ref{lem:dynamic:lambda-closed} gives
$\lambda_2(1)=-\frac{1}{\gamma}\log\theta_H\to+\infty$ as $\gamma\downarrow 0$.
For any $\mu$ and $\lambda$,
$M(\mu,\lambda;x)=\sum_{q\in Q^{2}}\mu(q)M(\delta_q,\lambda;x)\leq \sum_{q\in Q^{2}}\mu(q)C(\delta_q,\lambda),$
so taking the supremum over $x$ yields
$C(\mu,\lambda)\leq \sum_{q\in Q^{2}}\mu(q)C(\delta_q,\lambda)$.
By Lemma \ref{lem:cycles:capacity},
$C(\delta_q,\lambda)=\frac{1}{\lambda}\log\max\big\{\frac{p}{q(1)},\frac{1-p}{1-q(1)}\big\}$,
which is finite for every full-support $q$ and converges to $0$ as $\lambda\to+\infty$. Hence,
$C(\mu_2(\cdot|y_1=1),\lambda_2(1))\to 0$ as $\gamma\downarrow 0$. Since $k>0$, there exists $\bar\gamma>0$ such that
$C(\mu_2(\cdot|y_1=1),\lambda_2(1))<k$ for all $\gamma\in(0,\bar\gamma)$.

For the last claim, uniform elimination of the success-node shock means $\lambda_2(1)=0$ for every $\gamma>0$. Since $\lambda_2(1)=-\frac{1}{\gamma}\log\theta_H$, this holds if and only if $\theta_H=1$.
\end{proof}

\subsection{Proofs from Appendix \ref{app:infinite}}

\subsubsection*{Proof of Proposition \ref{thm:dynamic:infty:bridge} and Corollary \ref{cor:dynamic:infty:bridge-subseq-tree}}

\begin{proof}[Proof of Proposition \ref{thm:dynamic:infty:bridge}]
Fix $\beta\in(0,1)$ and an infinite-horizon contract $(x_t)_{t\ge 1}$ that is bounded, i.e.,
$\sup_{t\ge 1}\sup_{s_t\in S_t}\sup_{y\in Y}|x_t(s_t,y)|<\infty$.
Fix a metric $d$ on $\overline{\R}_+=[0,+\infty]$ satisfying (i)--(iii) in
Appendix \ref{subsec:dynamic:infty}. Fix a realized private history sequence $(h_t)_{t\ge 1}$ and
$(\mu_\infty,\lambda_\infty)$ satisfying \eqref{eq:dynamic:infty:state-unif}. Recall the public projection
$\chi_t:H_t\to S_t$.

\medskip
\noindent\textit{Step 1: One-step Lipschitz property in payoffs.}
Fix $q\in\Delta(Y)$ and $\lambda\in\overline{\R}_+$. For any payoff vectors $z,z':Y\to\R$,
\begin{align}\label{eq:app:infty:lipschitz-g}
\big|g_q(z;\lambda)-g_q(z';\lambda)\big|\le \|z-z'\|_\infty.
\end{align}
Consequently, for any $a\in\{1,2\}$ and any $(\mu,\lambda)\in\Delta(Q^{2})\times\overline{\R}_+$,
\begin{align}\label{eq:app:infty:lipschitz-G}
\big|\mathcal G^a(z;\mu,\lambda)-\mathcal G^a(z';\mu,\lambda)\big|\le \|z-z'\|_\infty.
\end{align}

\smallskip
\noindent\emph{Proof of \eqref{eq:app:infty:lipschitz-g} and \eqref{eq:app:infty:lipschitz-G}.}
If $\lambda\in(0,\infty)$, set $d_0:=\|z-z'\|_\infty$. Then, $z(y)\le z'(y)+d_0$ for all $y$, hence
$e^{-\lambda z(y)}\ge e^{-\lambda d_0}e^{-\lambda z'(y)}$ and so
$\sum_y q(y)e^{-\lambda z(y)}\ge e^{-\lambda d_0}\sum_y q(y)e^{-\lambda z'(y)}$.
Applying $-(1/\lambda)\log(\cdot)$ yields $g_q(z;\lambda)\le g_q(z';\lambda)+d_0$; swapping $z,z'$
gives the reverse inequality, hence $|g_q(z;\lambda)-g_q(z';\lambda)|\le d_0$.

If $\lambda=0$, then $g_q(z;0)=\sum_y q(y)z(y)$ by \eqref{eq:dynamic:infty:gq-ext}, so
$|g_q(z;0)-g_q(z';0)|\le \|z-z'\|_\infty$.
If $\lambda=+\infty$, then $g_q(z;+\infty)=\min_y z(y)$, so
$|\min_y z(y)-\min_y z'(y)|\le \|z-z'\|_\infty$.
This proves \eqref{eq:app:infty:lipschitz-g}. \eqref{eq:app:infty:lipschitz-G} follows by
averaging across $q$ (for $a=2$) and noting that $a=1$ is $g_{q^1}$ (independent of $\mu$).\hfill$\blacksquare$

\medskip
\noindent\textit{Step 2: Contraction of the date-$t$ Bellman operator.}
For each $t\ge 1$, define an operator $\mathcal T_t$ on bounded functions $W:H_{t+1}\to\R$ by
$$
(\mathcal T_t W)(h_t):=\max_{a\in\{1,2\}}
\left\{
\mathcal G^a \Big(y\mapsto x_t(\chi_t(h_t),y)+\beta W(h_t,a,y) ; \mu_t(\cdot|h_t),\lambda_t(h_t)\Big)-c^a
\right\}.
$$
Then, \eqref{eq:dynamic:infty:W} is $W_t=\mathcal T_t W_{t+1}$. By \eqref{eq:app:infty:lipschitz-G},
for any bounded $W,W'$,
\begin{align}\label{eq:app:infty:contraction}
\|\mathcal T_t W-\mathcal T_t W'\|_\infty\le \beta\|W-W'\|_\infty.
\end{align}

\medskip
\noindent\textit{Step 3: Existence and uniqueness of bounded solutions (part (a)).}
Let $\bar x:=\sup_{t,s_t,y}|x_t(s_t,y)|<\infty$ and $\bar c:=\max\{c^1,c^2\}=k$.
Set
\begin{align}\label{eq:app:infty:B}
B:=\frac{\bar x+\bar c}{1-\beta}.
\end{align}
For each finite horizon $T\ge 1$, define the truncated value functions $(W_t^T)_{t=1}^{T+1}$ by
$W_{T+1}^T\equiv 0$ and for $t=T,T-1,\dots,1$,
$W_t^T:=\mathcal T_t W_{t+1}^T.$
We first show by backward induction that $\|W_t^T\|_\infty\le B$ for all $t\le T+1$. Note that
$\|W_{T+1}^T\|_\infty=0\le B$. Suppose $\|W_{t+1}^T\|_\infty\le B$. Then, for each $h_t$ and $a\in\{1,2\}$,
the vector $y\mapsto x_t(\chi_t(h_t),y)+\beta W_{t+1}^T(h_t,a,y)$ lies in $[-(\bar x+\beta B),\bar x+\beta B]^Y$.
By \eqref{eq:app:infty:lipschitz-G} with $z'\equiv 0$ and $|\mathcal G^a(0;\cdot)|=0$, we obtain $\Big|\mathcal G^a \big(y\mapsto x_t(\chi_t(h_t),y)+\beta W_{t+1}^T(h_t,a,y) ; \mu_t(\cdot|h_t),\lambda_t(h_t)\big)\Big|
\le \bar x+\beta B.$
Therefore, using $c^a\le \bar c$ and the definition of $\mathcal T_t$,
$|W_t^T(h_t)|
\le
\max_{a\in\{1,2\}}\big\{(\bar x+\beta B)+\bar c\big\}
=
\bar x+\bar c+\beta B
=
B,$
where the last equality uses \eqref{eq:app:infty:B}. Hence, $\|W_t^T\|_\infty\le B$. Now fix $t\ge 1$ and take $T'>T\ge t$. Applying the contraction bound \eqref{eq:app:infty:contraction}
iteratively from date $t$ up to date $T$ yields the following
$\|W_t^{T'}-W_t^T\|_\infty
=
\|\mathcal T_t\cdots \mathcal T_T W_{T+1}^{T'}-\mathcal T_t\cdots \mathcal T_T W_{T+1}^{T}\|_\infty
\le
\beta^{T+1-t}\hspace{0.02in}\|W_{T+1}^{T'}-W_{T+1}^{T}\|_\infty.$
Since $W_{T+1}^{T}\equiv 0$ and $\|W_{T+1}^{T'}\|_\infty\le B$, we obtain
\begin{align}\label{eq:app:infty:cauchy}
\|W_t^{T'}-W_t^T\|_\infty\le \beta^{T+1-t}\hspace{0.02in}B.
\end{align}
Thus, $(W_t^T)_{T\ge t}$ is Cauchy in the Banach space of bounded functions on $H_t$, hence converges
to some bounded function $W_t$. Letting $T'\to\infty$ in \eqref{eq:app:infty:cauchy} gives
\begin{align}\label{eq:app:infty:trunc-error-W}
\|W_t-W_t^T\|_\infty\le \beta^{T+1-t}B\qquad\text{for all }T\ge t.
\end{align}
To verify that $(W_t)_{t\ge 1}$ satisfies \eqref{eq:dynamic:infty:W}, fix $t$ and note that
$W_t^T=\mathcal T_t W_{t+1}^T$ for all $T\ge t+1$. Taking $T\to\infty$ and using that $\mathcal T_t$ is
$\beta$-Lipschitz under the sup norm by \eqref{eq:app:infty:contraction}, we obtain
$W_t=\mathcal T_t W_{t+1}$. Hence, \eqref{eq:dynamic:infty:W} holds.

Uniqueness: if $(\widetilde W_t)$ is another bounded solution, then by \eqref{eq:app:infty:contraction},
$$
\|W_t-\widetilde W_t\|_\infty
=
\|\mathcal T_t W_{t+1}-\mathcal T_t \widetilde W_{t+1}\|_\infty
\le \beta\|W_{t+1}-\widetilde W_{t+1}\|_\infty
\le \cdots \le \beta^n \|W_{t+n}-\widetilde W_{t+n}\|_\infty.
$$
Since both sequences are bounded and $\beta^n\to 0$, $\|W_t-\widetilde W_t\|_\infty=0$.
This proves part (a).

\medskip
\noindent\textit{Step 4: Uniform continuity on bounded domains.}
Fix $B_0<\infty$ and consider the domain
$$
\mathscr D_{B_0}:=\Delta(Q^{2})\times\overline{\R}_+\times[-B_0,B_0]^Y
$$
equipped with the product metric $\|\mu-\mu'\|_1+d(\lambda,\lambda')+\|z-z'\|_\infty$.

\begin{lemma}\label{lem:app:infty:unifcont}
For each $a\in\{1,2\}$, the map
$\Phi_a(\mu,\lambda,z):=\mathcal G^a(z;\mu,\lambda)$
is continuous on $\mathscr D_{B_0}$ and therefore uniformly continuous.
Equivalently: for every $\varepsilon>0$ there exists $\rho>0$ such that if
$\|\mu-\mu'\|_1<\rho$, $d(\lambda,\lambda')<\rho$, and $\|z-z'\|_\infty<\rho$, then $\big|\mathcal G^a(z;\mu,\lambda)-\mathcal G^a(z';\mu',\lambda')\big|<\varepsilon.$
\end{lemma}

\par\noindent\textit{Proof of Lemma \ref{lem:app:infty:unifcont}}.
For $a=1$, $\mathcal G^1(z;\mu,\lambda)=g_{q^1}(z;\lambda)$ does not depend on $\mu$.
Continuity in $z$ is immediate from \eqref{eq:app:infty:lipschitz-g}. It remains to show continuity
in $\lambda\in\overline{\R}_+$.

Fix a full-support $q\in\Delta(Y)$ and $z\in[-B_0,B_0]^Y$. For $\lambda\in(0,\infty)$, set
$S(\lambda):=\sum_{y\in Y}q(y)e^{-\lambda z(y)}$, so $g_q(z;\lambda)=-(1/\lambda)\log S(\lambda)$.
Since $S$ is $C^1$ on $(0,\infty)$, $g_q$ is continuous on $(0,\infty)$. We check continuity at the
endpoints, which match \eqref{eq:dynamic:infty:gq-ext}.

\emph{As $\lambda\downarrow 0$.}
We have $S(0)=1$ and $S'(\lambda)=-\sum_y q(y)z(y)e^{-\lambda z(y)}$, so $S'(0)=-\sum_y q(y)z(y)$.
Since $\log S(\lambda)\to 0$ as $\lambda\downarrow 0$, by l'H\^{o}pital's rule,
$$
\lim_{\lambda\downarrow 0} g_q(z;\lambda)
=
-\lim_{\lambda\downarrow 0}\frac{\log S(\lambda)}{\lambda}
=
-\lim_{\lambda\downarrow 0}\frac{S'(\lambda)/S(\lambda)}{1}
=
-\frac{S'(0)}{S(0)}
=
\sum_{y\in Y}q(y)z(y)
=
g_q(z;0).
$$

\emph{As $\lambda\uparrow\infty$.}
Let $m:=\min_{y\in Y}z(y)$ and write $z(y)=m+\varDelta(y)$ with $\varDelta(y)\ge 0$.
Then, $S(\lambda)=e^{-\lambda m}\sum_{y\in Y}q(y)e^{-\lambda \varDelta(y)}$ and $g_q(z;\lambda)=m-\frac{1}{\lambda}\log\big(\sum_{y\in Y}q(y)e^{-\lambda \varDelta(y)}\big).$
Let $Y_\ast:=\{y:\varDelta(y)=0\}$ be the (nonempty) argmin set. Since $q$ has full support,
$\sum_{y\in Y_\ast}q(y)\in(0,1]$. Moreover, for all $\lambda\ge 0$,
$\sum_{y\in Y_\ast}q(y) \le \sum_{y\in Y}q(y)e^{-\lambda \varDelta(y)} \le \sum_{y\in Y}q(y)=1.$
Taking logs and dividing by $\lambda$ yields a quantity between $\frac{1}{\lambda}\log\sum_{y\in Y_\ast}q(y)$ and $0$,
which converges to $0$ as $\lambda\to\infty$. Hence, $\lim_{\lambda\uparrow\infty}g_q(z;\lambda)=m=g_q(z;+\infty)$.

To verify joint continuity at $\lambda=+\infty$, we show this convergence is uniform over the compact domain.
Define $q_{\min}:=\min_{q\in Q^{2}\cup\{q^1\}}\min_{y\in Y}q(y)>0$. Since $q_{\min}\le\sum_{y\in Y}q(y)e^{-\lambda\varDelta(y)}\le 1$ for all $\lambda\ge 0$, $q\in Q^{2}\cup\{q^1\}$, and $z\in[-B_0,B_0]^Y$, we obtain
$$
\big|g_q(z;\lambda)-g_q(z;+\infty)\big|
=
\frac{1}{\lambda}\Big|\log\Big(\sum_{y\in Y}q(y)e^{-\lambda\varDelta(y)}\Big)\Big|
\le
\frac{|\log q_{\min}|}{\lambda}
$$
for all $\lambda>0$. This bound is independent of $(q,z)$ and vanishes as $\lambda\to+\infty$, so the convergence $g_q(z;\lambda)\to g_q(z;+\infty)$ is uniform over $(q,z)\in (Q^{2}\cup\{q^1\})\times[-B_0,B_0]^Y$.

\par Thus, $g_q(z;\lambda)$ is jointly continuous on $(Q^{2}\cup\{q^1\})\times[-B_0,B_0]^Y\times\overline{\R}_+$ when $\overline{\R}_+$ is endowed with any metric $d$ satisfying (ii)--(iii), since these conditions ensure that $d$-convergence agrees with the usual boundary notions $\lambda\downarrow 0$ and $\lambda\uparrow\infty$.

For $a=2$, $\mathcal G^2(z;\mu,\lambda)=\sum_{q\in Q^{2}}\mu(q)g_q(z;\lambda)$ is a finite sum of products
of continuous functions in $(\mu,\lambda,z)$, hence continuous on $\mathscr D_{B_0}$ as well.

Finally, $\Delta(Q^{2})$ is compact, $[-B_0,B_0]^Y$ is compact, and $(\overline{\R}_+,d)$ is compact by (i), so
$\mathscr D_{B_0}$ is compact. By Heine--Cantor's theorem, $\Phi_a$ is uniformly continuous on $\mathscr D_{B_0}$. \hfill$\blacksquare$

\medskip
\noindent\textit{Step 5: The limiting recursion and its truncation error.}
Define the limiting Bellman operators $\mathcal T_t^\infty$ on bounded $W:H_{t+1}\to\R$ by
$$
(\mathcal T_t^\infty W)(h_t):=\max_{a\in\{1,2\}}
\left\{
\mathcal G^a \Big(y\mapsto x_t(\chi_t(h_t),y)+\beta W(h_t,a,y) ; \mu_\infty,\lambda_\infty\Big)-c^a
\right\}.
$$
By the same argument as in Steps 2--3 (replacing $\mathcal T_t$ by $\mathcal T_t^\infty$), the recursion
$W_t^\infty=\mathcal T_t^\infty W_{t+1}^\infty$ admits a unique bounded solution. Defining the
finite-horizon truncations $(W_t^{\infty,T})_{t=1}^{T+1}$ by $W_{T+1}^{\infty,T}\equiv 0$ and
$W_t^{\infty,T}:=\mathcal T_t^\infty W_{t+1}^{\infty,T}$ for $t\le T$, we obtain the analogue of
\eqref{eq:app:infty:trunc-error-W}:
\begin{align}\label{eq:app:infty:trunc-error-Winf}
\|W_t^\infty-W_t^{\infty,T}\|_\infty\le \beta^{T+1-t}B\qquad\text{for all }T\ge t.
\end{align}

\medskip
\noindent\textit{Step 6: Pathwise convergence $W_t(h_t)\to W_t^\infty(h_t)$ (part (b)).}
Fix an integer $N\ge 0$ and set $T:=t+N$ (so $T-t=N$). Consider the finite-horizon truncations
$W^T=(W_s^T)_{s\le T+1}$ and $W^{\infty,T}=(W_s^{\infty,T})_{s\le T+1}$. For $s\in\{t,t+1,\dots,T\}$, define
$\mathcal H_{t,N}^s(h_t):=\mathcal H_{t,N}(h_t)\cap H_s.$
For the terminal successor layer, set
$\mathcal H_{t,N}^{T+1}(h_t)
:=
\big\{(\tilde h_T,a,y):\tilde h_T\in\mathcal H_{t,N}^{T}(h_t),a\in\{1,2\}, y\in Y\big\}.$
Define
$e_s^{t,N}:=\sup_{\tilde h_s\in\mathcal H_{t,N}^s(h_t)}
\big|W_s^T(\tilde h_s)-W_s^{\infty,T}(\tilde h_s)\big|,$ for $ s=t,\dots,T+1.$
The terminal condition is $e_{T+1}^{t,N}=0$ because $W_{T+1}^T=W_{T+1}^{\infty,T}\equiv0$ on the terminal successor layer.

\medskip
\noindent\emph{Step 6.1: A finite-horizon error recursion.}
Fix $s\in\{t,\dots,T\}$ and $\tilde h_s\in\mathcal H_{t,N}^s(h_t)$. Using
$W_s^T=\mathcal T_s W_{s+1}^T$ and $W_s^{\infty,T}=\mathcal T_s^\infty W_{s+1}^{\infty,T}$, we write
\begin{align*}
\big|W_s^T(\tilde h_s)-W_s^{\infty,T}(\tilde h_s)\big| \le
\big|\big(\mathcal T_s W_{s+1}^T\big)(\tilde h_s)-\big(\mathcal T_s W_{s+1}^{\infty,T}\big)(\tilde h_s)\big|
+
\big|\big(\mathcal T_s W_{s+1}^{\infty,T}\big)(\tilde h_s)-\big(\mathcal T_s^\infty W_{s+1}^{\infty,T}\big)(\tilde h_s)\big|.
\end{align*}
We bound the first term using \eqref{eq:app:infty:lipschitz-G}. For each $a\in\{1,2\}$, define
$z_{s,a}^T(\tilde h_s)(y):=x_s(\chi_s(\tilde h_s),y)+\beta W_{s+1}^T(\tilde h_s,a,y)$ and $z_{s,a}^{\infty,T}(\tilde h_s)(y):=x_s(\chi_s(\tilde h_s),y)+\beta W_{s+1}^{\infty,T}(\tilde h_s,a,y).$
Then, for each $a$,
\begin{align*}
\Big|\mathcal G^a\big(z_{s,a}^T(\tilde h_s);\mu_s(\tilde h_s),\lambda_s(\tilde h_s)\big)
-
\mathcal G^a\big(z_{s,a}^{\infty,T}(\tilde h_s);\mu_s(\tilde h_s),\lambda_s(\tilde h_s)\big)\Big|
&\le
\|z_{s,a}^T(\tilde h_s)-z_{s,a}^{\infty,T}(\tilde h_s)\|_\infty\\
&=
\beta\sup_{y\in Y}\big|W_{s+1}^T(\tilde h_s,a,y)-W_{s+1}^{\infty,T}(\tilde h_s,a,y)\big|.
\end{align*}
Taking the maximum over $a$ yields
$$
\big|\big(\mathcal T_s W_{s+1}^T\big)(\tilde h_s)-\big(\mathcal T_s W_{s+1}^{\infty,T}\big)(\tilde h_s)\big|
\le
\beta\sup_{a\in\{1,2\}}\sup_{y\in Y}\big|W_{s+1}^T(\tilde h_s,a,y)-W_{s+1}^{\infty,T}(\tilde h_s,a,y)\big|.
$$
Since $(\tilde h_s,a,y)\in \mathcal H_{t,N}^{s+1}(h_t)$ for $\tilde h_s\in\mathcal H_{t,N}^s(h_t)$,
 then we obtain the following inequality $
\sup_{a\in\{1,2\}}\sup_{y\in Y}\big|W_{s+1}^T(\tilde h_s,a,y)-W_{s+1}^{\infty,T}(\tilde h_s,a,y)\big|
\le
e_{s+1}^{t,N}.
$
Thus,
\begin{align}\label{eq:app:infty:finite-rec-1}
\big|\big(\mathcal T_s W_{s+1}^T\big)(\tilde h_s)-\big(\mathcal T_s W_{s+1}^{\infty,T}\big)(\tilde h_s)\big|
\le \beta\hspace{0.02in}e_{s+1}^{t,N}.
\end{align}

For the second term, define
$$
\eta_s^{t,N}:=
\sup_{\tilde h_s\in\mathcal H_{t,N}^s(h_t)}
\max_{a\in\{1,2\}}
\Big|
\mathcal G^a\big(z_{s,a}^{\infty,T}(\tilde h_s);\mu_s(\tilde h_s),\lambda_s(\tilde h_s)\big)
-
\mathcal G^a\big(z_{s,a}^{\infty,T}(\tilde h_s);\mu_\infty,\lambda_\infty\big)
\Big|.
$$
Then, by definition of $\mathcal T_s$ and $\mathcal T_s^\infty$ as maxima over $a$ of the same cost-adjusted
terms, we have for each $\tilde h_s$,
$\big|\big(\mathcal T_s W_{s+1}^{\infty,T}\big)(\tilde h_s)-\big(\mathcal T_s^\infty W_{s+1}^{\infty,T}\big)(\tilde h_s)\big|
\le \eta_s^{t,N}.$
Combining with \eqref{eq:app:infty:finite-rec-1} and taking the supremum over
$\tilde h_s\in\mathcal H_{t,N}^s(h_t)$ yields the finite-horizon error recursion
\begin{align}\label{eq:app:infty:finite-rec}
e_s^{t,N}\le \beta\hspace{0.02in}e_{s+1}^{t,N}+\eta_s^{t,N}
\qquad\text{for }s=t,t+1,\dots,T,
\end{align}
with terminal condition $e_{T+1}^{t,N}=0$. Iterating \eqref{eq:app:infty:finite-rec} forward from $s=t$ to $s=T$ gives 
\begin{align}\label{eq:app:infty:finite-bound}
e_t^{t,N}\le \sum_{j=0}^{N}\beta^j\hspace{0.02in}\eta_{t+j}^{t,N}.
\end{align}

\medskip
\noindent\emph{Step 6.2: Control of $\eta_s^{t,N}$.}
Let $\bar W:=B$ from \eqref{eq:app:infty:B}. From the construction of $W^{\infty,T}$ in Step 5 and the same
bound argument as in Step 3, we have $\|W_{s}^{\infty,T}\|_\infty\le B$ for all $s\le T+1$.
Define $B':=\bar x+\beta B.$
Then, for all $s\le T$, all $\tilde h_s\in H_s$, and all $a\in\{1,2\}$, the vector
$z_{s,a}^{\infty,T}(\tilde h_s)$ lies in $[-B',B']^Y$. Fix $\varepsilon>0$. Apply Lemma \ref{lem:app:infty:unifcont} with $B_0=B'$ to obtain $\rho>0$ such that
if $\|\mu-\mu_\infty\|_1<\rho$ and $d(\lambda,\lambda_\infty)<\rho$, then
$\big|\mathcal G^a(z;\mu,\lambda)-\mathcal G^a(z;\mu_\infty,\lambda_\infty)\big|
<
\varepsilon(1-\beta)$ $\text{for all }a\in\{1,2\}\text{ and all }z\in[-B',B']^Y.$
Fix the horizon $N$. By assumption \eqref{eq:dynamic:infty:state-unif}, there
exists $T_\varepsilon$ such that for all $t\ge T_\varepsilon$, $\sup_{\tilde h\in\mathcal H_{t,N}(h_t)}\|\mu_{|\tilde h|}(\tilde h)-\mu_\infty\|_1<\rho,$ and $\sup_{\tilde h\in\mathcal H_{t,N}(h_t)} d \big(\lambda_{|\tilde h|}(\tilde h),\lambda_\infty\big)<\rho.$
In particular, for all $t\ge T_\varepsilon$, for all $s\in\{t,\dots,t+N\}$, and all
$\tilde h_s\in\mathcal H_{t,N}^s(h_t)$, we have
$\|\mu_s(\tilde h_s)-\mu_\infty\|_1<\rho$ and $d(\lambda_s(\tilde h_s),\lambda_\infty)<\rho$.
Therefore, for all such $(s,\tilde h_s)$ and all $a\in\{1,2\}$,
$\Big|
\mathcal G^a\big(z_{s,a}^{\infty,T}(\tilde h_s);\mu_s(\tilde h_s),\lambda_s(\tilde h_s)\big)
-
\mathcal G^a\big(z_{s,a}^{\infty,T}(\tilde h_s);\mu_\infty,\lambda_\infty\big)
\Big|
<
\varepsilon(1-\beta).$
Taking the maximum over $a$ and the supremum over $\tilde h_s\in\mathcal H_{t,N}^s(h_t)$ yields
$\eta_s^{t,N}\le \varepsilon(1-\beta)$ $\text{for all }s=t,t+1,\dots,t+N \text{ whenever }t\ge T_\varepsilon.$
Plugging into \eqref{eq:app:infty:finite-bound} gives, for $t\ge T_\varepsilon$, $e_t^{t,N}
\le
\sum_{j=0}^{N}\beta^j\hspace{0.02in}\varepsilon(1-\beta)
\le
\varepsilon.$

\medskip
\medskip
\noindent\emph{Step 6.3: Conclude the infinite-horizon.}
Fix $\varepsilon'>0$. We combine the truncation error from Step 3 with the finite-horizon operator error from Step 6.2.

\emph{Step (i): Choose truncation horizon.}
Select $N$ large enough so that
$2\hspace{0.02in}\beta^{N+1}B<\varepsilon'/3.$

\emph{Step (ii): Apply Step 6.2 with $\varepsilon:=\varepsilon'/3$.}
By Step 6.2 (with the choice $\varepsilon:=\varepsilon'/3$), there exists $T_{\varepsilon'}$ such that for all $t\ge T_{\varepsilon'}$,
$e_t^{t,N}
=
\sup_{\tilde h_t\in\mathcal H_{t,N}^t(h_t)}\big|W_t^{t+N}(\tilde h_t)-W_t^{\infty,t+N}(\tilde h_t)\big|\leq \varepsilon'/3.$
In particular, $\big|W_t^{t+N}(h_t)-W_t^{\infty,t+N}(h_t)\big|<\varepsilon'/3$ for all $t\ge T_{\varepsilon'}$.

\emph{Step (iii): Apply the bounds.}
By \eqref{eq:app:infty:trunc-error-W} and \eqref{eq:app:infty:trunc-error-Winf} with $T=t+N$, for all $t\ge 1$,
$$
|W_t(h_t)-W_t^{t+N}(h_t)|\le \beta^{N+1}B<\varepsilon'/6,
\qquad
|W_t^\infty(h_t)-W_t^{\infty,t+N}(h_t)|\le \beta^{N+1}B<\varepsilon'/6.
$$

\emph{Step (iv): Combine via the triangle inequality.}
For all $t\ge T_{\varepsilon'}$,
\begin{align*}
|W_t(h_t)-W_t^\infty(h_t)|
&\le
|W_t(h_t)-W_t^{t+N}(h_t)|
+
|W_t^{t+N}(h_t)-W_t^{\infty,t+N}(h_t)|
+
|W_t^{\infty,t+N}(h_t)-W_t^\infty(h_t)|\\
&
\leq\varepsilon'/6+\varepsilon'/3+\varepsilon'/6
 = 2\varepsilon'/3
 < \varepsilon'.
\end{align*}
Since $\varepsilon'>0$ was arbitrary, $|W_t(h_t)-W_t^\infty(h_t)|\to 0$ as $t\to\infty$, proving part (b).
\end{proof}

\begin{proof}[Proof of Corollary \ref{cor:dynamic:infty:bridge-subseq-tree}]
Fix $L\ge 0$ and $\varepsilon'>0$.
Let $B$ be the uniform bound in \eqref{eq:app:infty:B}.
Choose an integer $R\ge 0$ such that
$2\beta^{R+1}B<\varepsilon'/3,$
and define
$N:=L+R.$

Fix $n$ and write $t:=\tau_n$ and $T:=t+N$.
Consider the finite-horizon truncations
$(W_s^T)_{s\le T+1}$ and $(W_s^{\infty,T})_{s\le T+1}$
defined in the proof of Proposition \ref{thm:dynamic:infty:bridge}.
For each $s\in\{t,\dots,T\}$, define
$e_s^{t,N}
:=
\sup_{\tilde h_s\in\mathcal H_{t,N}^s(h_t)}
\big|W_s^T(\tilde h_s)-W_s^{\infty,T}(\tilde h_s)\big|.$

Step 6.1 in the proof of Proposition \ref{thm:dynamic:infty:bridge} gives
\begin{equation}\label{eq:app:infty:finite-rec-subseq}
e_s^{t,N}
\le
\beta\hspace{0.02in}e_{s+1}^{t,N}
+
\eta_s^{t,N},
\qquad
s=t,\dots,T,
\end{equation}
with terminal condition $e_{T+1}^{t,N}=0$, where $\eta_s^{t,N}$ is defined there.

By \eqref{eq:dynamic:infty:state-unif-subseq}, the local uniform convergence of states holds along
the base-date sequence $(\tau_n)$ for the fixed horizon $N=L+R$.
Therefore, Step 6.2 in the proof of Proposition
\ref{thm:dynamic:infty:bridge}, applied along this sequence, implies that there exists
$n_{\varepsilon'}$ such that for every $n\ge n_{\varepsilon'}$,
$\eta_s^{t,N}
\le
\frac{\varepsilon'}{3}(1-\beta)$ $\text{for all }s=t,t+1,\dots,T.$
Iterating \eqref{eq:app:infty:finite-rec-subseq} forward from any
$s\in\{t,\dots,T\}$ yields
$e_s^{t,N}
\le
\sum_{j=0}^{T-s}\beta^j\eta_{s+j}^{t,N}
\le
\frac{\varepsilon'}{3}(1-\beta)
\sum_{j=0}^{T-s}\beta^j
\le
\frac{\varepsilon'}{3}.$

Now, fix any
$\tilde h\in\mathcal H_{t,L}(h_t)$
and let $s:=|\tilde h|$.
Since $s\le t+L$ and $N=L+R$, we have
$\tilde h\in\mathcal H_{t,N}^s(h_t)$ and
$T+1-s
=
t+L+R+1-s
\ge
R+1.$
Hence,
$\big|W_s^T(\tilde h)-W_s^{\infty,T}(\tilde h)\big|
\le
e_s^{t,N}
\le
\frac{\varepsilon'}{3}.$
Moreover, by \eqref{eq:app:infty:trunc-error-W} and
\eqref{eq:app:infty:trunc-error-Winf},
$\big|W_s(\tilde h)-W_s^T(\tilde h)\big|
\le
\beta^{T+1-s}B
\le
\beta^{R+1}B
<
\frac{\varepsilon'}{6},$
and
$\big|W_s^\infty(\tilde h)-W_s^{\infty,T}(\tilde h)\big|
\le
\beta^{T+1-s}B
\le
\beta^{R+1}B
<
\frac{\varepsilon'}{6}.$
Therefore, for every $n\ge n_{\varepsilon'}$ and every
$\tilde h\in\mathcal H_{\tau_n,L}(h_{\tau_n})$,
$\big|W_{|\tilde h|}(\tilde h)-W_{|\tilde h|}^\infty(\tilde h)\big|
\le
\frac{\varepsilon'}{6}
+
\frac{\varepsilon'}{3}
+
\frac{\varepsilon'}{6}
<
\varepsilon'.$
Taking the supremum over
$\tilde h\in\mathcal H_{\tau_n,L}(h_{\tau_n})$
and using that $\varepsilon'>0$ was arbitrary proves the claim.
\end{proof}

\subsubsection*{Proofs of  Lemmas \ref{lem:cycles:llr-frequency}--\ref{lem:cycles:Wlimit-public}}
\begin{proof}[Proof of Lemma \ref{lem:cycles:llr-frequency}]
Let $N_{2,T}:=\sum_{t=1}^T\1\{a_t=2\}$. For $N_{2,T}>0$, let
$\widehat p_T^2(y):=N_{2,T}^{-1}\sum_{t=1}^T\1\{a_t=2,y_t=y\}$; if $N_{2,T}=0$, define $\widehat p_T^2$ arbitrarily. We first show that $\widehat p_T^2\to p_\ast^2$ a.s. on $\{N_{2,T}\to\infty\}$. Fix $y\in Y$ and set
$M_T(y):=\sum_{t=1}^T\1\{a_t=2\}\big(\1\{y_t=y\}-p_\ast^2(y)\big)$. Since $a_t$ is chosen before $y_t$ is realized and, conditional on $a_t=2$, $y_t$ has distribution $p_\ast^2$, $(M_T(y))_{T\ge1}$ is a martingale with increments bounded by one. Let $\tau_n:=\inf\{T:N_{2,T}=n\}$. Then $(M_{\tau_n}(y))_{n\ge1}$, stopped after $\tau_n=\infty$, is a bounded-increment martingale, so Azuma-Hoeffding's inequality gives
$\mathbb P(|M_{\tau_n}(y)|>\varepsilon n,\tau_n<\infty)\le 2e^{-\varepsilon^2 n/2}$ for every $\varepsilon>0$. Borel-Cantelli therefore implies $M_{\tau_n}(y)/n\to0$ a.s. on $\{\tau_n<\infty\ \forall n\}$, or equivalently $M_T(y)/N_{2,T}\to0$ a.s. on $\{N_{2,T}\to\infty\}$. Since $Y$ is finite, this proves $\widehat p_T^2\to p_\ast^2$ a.s. on $\{N_{2,T}\to\infty\}$. Next, by the definition of $\mathrm{LLR}^2$ in \eqref{eq:dynamic:infty:LLR2}, for every $T$ with $N_{2,T}>0$,
$\mathrm{LLR}^2(h_{T+1},Q^{2})
=
N_{2,T}\min_{q\in Q^{2}}\DKL(\widehat p_T^2\|q).$
Indeed, the unrestricted likelihood over $\Delta(Y)$ is maximized at the empirical distribution $\widehat p_T^2$, while the restricted likelihood is maximized over $q\in Q^{2}$. If $N_{2,T}=0$, both sides are zero by the convention that empty product is $1$. Since $Q^{2}$ is finite and all its elements have full support, $r\mapsto\min_{q\in Q^{2}}\DKL(r\|q)$ is continuous on $\Delta(Y)$ and bounded above by $\log(1/\underline q)$, where $\underline q:=\min_{q\in Q^{2}}\min_{y\in Y}q(y)>0$. Thus, if $\alpha=0$, then
$\mathrm{LLR}^2(h_{T_n+1},Q^{2})/T_n\le \alpha_{T_n}\log(1/\underline q)\to0=\alpha D_\ast^2$. If $\alpha>0$, then $N_{2,T_n}\to\infty$, so $\widehat p_{T_n}^2\to p_\ast^2$ a.s.; hence
$$
\frac{\mathrm{LLR}^2(h_{T_n+1},Q^{2})}{T_n}
=
\alpha_{T_n}\min_{q\in Q^{2}}\DKL(\widehat p_{T_n}^2\|q)
\longrightarrow
\alpha\min_{q\in Q^{2}}\DKL(p_\ast^2\|q)
=
\alpha D_\ast^2 \quad \text{a.s.}
$$
The statement for $\lambda_{T_n+1}$ follows from the definition
$\lambda_{T_n+1}=\mathrm{LLR}^2(h_{T_n+1},Q^{2})/(\gamma T_n)$.
\end{proof}

\begin{proof}[Proof of Lemma \ref{lem:cycles:posterior}]
Fix $q\notin{Q}_*^{2}$ and choose any $\hat q\in{Q}_*^{2}$. Since $\mu_1$ has full support on $Q^{2}$ and all models in $Q^{2}$ have full support on $Y$, Bayes rule gives
$\log\frac{\mu_{T+1}(q|h_{T+1})}{\mu_{T+1}(\hat q|h_{T+1})}
=
\log\frac{\mu_1(q)}{\mu_1(\hat q)}
+
\sum_{t=1}^{T}\1\{a_t=2\}\log\frac{q(y_t)}{\hat q(y_t)}.$
Divide by $N_{2,T}:=\sum_{t=1}^T\1\{a_t=2\}$. By the empirical-convergence argument in the proof of Lemma \ref{lem:cycles:llr-frequency}, the second term divided by $N_{2,T}$ converges a.s. to
$\sum_{y\in Y}p_\ast^2(y)\log(q(y)/\hat q(y))=\DKL(p_\ast^2\|\hat q)-\DKL(p_\ast^2\|q)<0$, while the prior-odds term divided by $N_{2,T}$ converges to zero. Hence the log posterior odds converge to $-\infty$, so
$\mu_{T+1}(q|h_{T+1})/\mu_{T+1}(\hat q|h_{T+1})\to0$ a.s. Since $\mu_{T+1}(\hat q|h_{T+1})\le1$, this implies $\mu_{T+1}(q|h_{T+1})\to0$ a.s. for every $q\notin{Q}_*^{2}$.
\end{proof}
\begin{proof}[Proof of Lemma \ref{lem:dynamic:infty:mu-local}]
Write $|\tilde h|=t+j$ with $0\le j\le N$. For any $q\in Q^{2}\setminus\{q_\star\}$, Bayes rule
\eqref{eq:dynamic:infty:mu-update} implies the posterior-odds identity $\frac{\mu_{t+j}(q|\tilde h)}{\mu_{t+j}(q_\star|\tilde h)}
=
\frac{\mu_t(q|h_t)}{\mu_t(q_\star|h_t)}
\prod_{\tau=t}^{t+j-1}\frac{q^{a_\tau}(y_\tau)}{q_\star^{a_\tau}(y_\tau)}.$
Along any period $\tau$ in the extension with $a_\tau=1$, the likelihood ratio equals $1$ because the routine action
distribution is common across all structured models. Along any period with $a_\tau=2$,
the likelihood ratio is $q(y_\tau)/q_\star(y_\tau)\le K$ by \eqref{eq:app:infty:K}. Since there are at most
$j\le N$ periods in the extension, we get $\frac{\mu_{t+j}(q|\tilde h)}{\mu_{t+j}(q_\star|\tilde h)}
\le
K^{N}\frac{\mu_t(q|h_t)}{\mu_t(q_\star|h_t)},$ $\forall q\neq q_\star$.
Summing over $q\neq q_\star$ yields
$$
\frac{1-\mu_{t+j}(q_\star|\tilde h)}{\mu_{t+j}(q_\star|\tilde h)}
=
\sum_{q\neq q_\star}\frac{\mu_{t+j}(q|\tilde h)}{\mu_{t+j}(q_\star|\tilde h)}
\le
K^{N}\sum_{q\neq q_\star}\frac{\mu_t(q|h_t)}{\mu_t(q_\star|h_t)}
=
K^{N}\frac{1-\mu_t(q_\star|h_t)}{\mu_t(q_\star|h_t)}.
$$
Since $\mu_{t+j}(q_\star|\tilde h)\le 1$, multiplying both sides by $\mu_{t+j}(q_\star|\tilde h)$ gives
\eqref{eq:dynamic:infty:mu-local}. The $\ell^1$ bound follows from
$\|\mu-\delta_{\{q_\star\}}\|_1 = 2(1-\mu(q_\star))$ on the finite set $Q^{2}$.
\end{proof}

\begin{proof}[Proof of Lemma \ref{lem:dynamic:infty:llr-local}]
Write $|\tilde h|=t+j$ with $0\le j\le N$.

\emph{(i) Bounding the structured likelihood term.}
Let $q_t\in\arg\max_{q\in Q^{2}}\prod_{\tau=1}^{t-1}q(y_\tau)^{\1\{a_\tau=2\}}$, so $M^2_t(h_t)=\prod_{\tau=1}^{t-1}q_t(y_\tau)^{\1\{a_\tau=2\}}$.
For any extension by $j$ periods,
$$
M^2_{t+j}(\tilde h)
=
\max_{q\in Q^{2}}\prod_{\tau=1}^{t+j-1}q(y_\tau)^{\1\{a_\tau=2\}}
 \ge
\prod_{\tau=1}^{t+j-1}q_t(y_\tau)^{\1\{a_\tau=2\}}
=
M^2_t(h_t)\prod_{\tau=t}^{t+j-1}q_t(y_\tau)^{\1\{a_\tau=2\}}
 \ge
M^2_t(h_t)\underline q^{j},
$$
where the last inequality uses \eqref{eq:app:infty:qunderline}. Since also $M^2_{t+j}(\tilde h)\le M^2_t(h_t)$, we obtain
\begin{align}\label{eq:app:infty:M-ratio}
0 \le \log\frac{M^2_t(h_t)}{M^2_{t+j}(\tilde h)} \le j\log\frac{1}{\underline q} \le N\log\frac{1}{\underline q}.
\end{align}

\emph{(ii) Bounding the unstructured likelihood term.}
Let $p_t\in\arg\max_{p\in\Delta(Y)}\prod_{\tau=1}^{t-1}p(y_\tau)^{\1\{a_\tau=2\}}$, so $U^2_t(h_t)=\prod_{\tau=1}^{t-1}p_t(y_\tau)^{\1\{a_\tau=2\}}$.
Define $\tilde p\in\Delta(Y)$ by setting
$$
\tilde p
:=
\left(1-\frac{n_2}{t+j}\right)p_t
+
\frac{1}{t+j}\sum_{\tau=t:a_\tau=2}^{t+j-1}\delta_{\{y_\tau\}},
$$
where $n_2:=\#\big\{\tau\in\{t,\ldots,t+j-1\}:a_{\tau}=2\big\}\in\{0,1,\ldots,j\}$. Then, $\tilde p\in\Delta(Y)$, for each new innovation observation $\tau\in\{t,\dots,t+j-1\}$ with $a_\tau=2$, we have
$\tilde p(y_\tau)\ge 1/(t+j)$, and for each past innovation observation $\tau\le t-1$ with $a_\tau=2$, we have
$\tilde p(y_\tau)\ge (1-\frac{n_2}{t+j})p_t(y_\tau)\ge \frac{t}{t+j}p_t(y_\tau)$ because $1-\frac{n_2}{t+j}=\frac{t+j-n_2}{t+j}\geq\frac{t}{t+j}$ (as $n_2\le j$). Thus,
$$
U^2_{t+j}(\tilde h)
=
\max_{p\in\Delta(Y)}\prod_{\tau=1}^{t+j-1}p(y_\tau)^{\1\{a_\tau=2\}}
 \ge
\prod_{\tau=1}^{t+j-1}\tilde p(y_\tau)^{\1\{a_\tau=2\}}
 \ge
\left(\frac{t}{t+j}\right)^{t-1}U^2_t(h_t) \left(\frac{1}{t+j}\right)^{j}.
$$
Using $\big(\frac{t}{t+j}\big)^{t-1}=\left(1+\frac{j}{t}\right)^{-(t-1)}\ge e^{-j}$, we get
\begin{align}\label{eq:app:infty:U-ratio}
0 \le \log\frac{U^2_t(h_t)}{U^2_{t+j}(\tilde h)} \le j\big(1+\log(t+j)\big) \le N\big(1+\log(t+N)\big).
\end{align}

\emph{(iii) Normalized bound.}
Since $\mathrm{LLR}^2=\log U^2-\log M^2$, combining \eqref{eq:app:infty:M-ratio}--\eqref{eq:app:infty:U-ratio} gives
$$
\big|\mathrm{LLR}^2(\tilde h,Q^{2})-\mathrm{LLR}^2(h_t,Q^{2})\big|
\le
N\log\frac{1}{\underline q}+N\big(1+\log(t+N)\big)
\le
C_N\big(1+\log(t+N)\big)
$$
for a constant $C_N<\infty$ depending only on $N$ and $\underline q$. Finally,
$\big|\frac{\mathrm{LLR}^2(\tilde h,Q^{2})}{t+j-1}-\frac{\mathrm{LLR}^2(h_t,Q^{2})}{t-1}\big|
\le
\frac{|\mathrm{LLR}^2(\tilde h,Q^{2})-\mathrm{LLR}^2(h_t,Q^{2})|}{t-1}
+
\mathrm{LLR}^2(\tilde h,Q^{2})\big|\frac{1}{t+j-1}-\frac{1}{t-1}\big|.$
For the second term, note $U^2_{t+j}(\tilde h)\le 1$ and $M^2_{t+j}(\tilde h)\ge \underline q^{t+j-1}$, so
$\mathrm{LLR}^2(\tilde h,Q^{2})=\log\frac{U^2_{t+j}(\tilde h)}{M^2_{t+j}(\tilde h)}
\le
\log\frac{1}{\underline q^{t+j-1}}
\le
(t+N)\log\frac{1}{\underline q}.$
Also, $\big|\frac{1}{t+j-1}-\frac{1}{t-1}\big|=\frac{j}{(t-1)(t+j-1)}\le \frac{N}{(t-1)^2}$, so the second term is bounded by
$\log(1/\underline q) \frac{N(t+N)}{(t-1)^2}\le C'_N/t$ for a constant $C'_N<\infty$ (fixed $N$). Combining these bounds
 yields \eqref{eq:dynamic:infty:llr-local}.
\end{proof}

\begin{proof}[Proof of Lemma \ref{lem:cycles:Wlimit-public}]
Define a recursion on public histories as the following function:
$\bar W_t^\infty(s_t)
=
\max_{a\in\{1,2\}}
\left\{
\mathcal G^{a,\infty}\big(y\mapsto x_t(s_t,y)+\beta\bar W_{t+1}^\infty(s_t,y)\big)-c^a
\right\}.$
The same contraction argument as in Proposition \ref{thm:dynamic:infty:bridge}.(a) gives a unique bounded solution $(\bar W_t^\infty)_{t\ge1}$. Define
$\widetilde W_t^\infty(h_t):=\bar W_t^\infty(\chi_t(h_t)).$
Since $x_t$ depends on private histories only through $\chi_t(h_t)$, the sequence $(\widetilde W_t^\infty)_{t\ge1}$ satisfies the private-history recursion \eqref{eq:dynamic:infty:Wlimit}. By uniqueness of the bounded solution to \eqref{eq:dynamic:infty:Wlimit}, $\widetilde W_t^\infty=W_t^\infty$ $\forall t$. Thus, $W_t^\infty(h_t)$ depends on $h_t$ only through $\chi_t(h_t)$. Finally, for any $a,a'\in\{1,2\}$ and $y\in Y$,
$\chi_{t+1}(h_t,a,y)=\chi_{t+1}(h_t,a',y),$
because the public history records outcomes but not actions. Therefore,
$W_{t+1}^\infty(h_t,a,y)
=
\bar W_{t+1}^\infty(\chi_{t+1}(h_t,a,y))
=
\bar W_{t+1}^\infty(\chi_{t+1}(h_t,a',y))
=
W_{t+1}^\infty(h_t,a',y),$
so $W_{t+1}^\infty(h_t,a,y)$ is independent of $a$.
\end{proof}

\subsection{Proofs from Appendix \ref{app:dynamic:extensions}}

\begin{proof}[Proof of Observation \ref{prop:dynamic:milestones-irrelevance}]
Fix $(\mu,\lambda)$, $\eta\in(0,1)$, and an arbitrary continuation utility vector $x:Y\to\R$, where $Y=\{1,0,d\}$.
Define $x_F:=-\frac{1}{\lambda}\log\big(\eta e^{-\lambda x(d)}+(1-\eta)e^{-\lambda x(0)}\big).$
We first show that under the stated restriction, each model's entropic certainty equivalent depends on $x$
only through the pair $(x(1),x_F)$. For any $q\in\{q^1\}\cup Q^{2}$, we are given that
$q(d)=(1-q(1))\eta$ and $q(0)=(1-q(1))(1-\eta)$, hence
$\sum_{y\in Y} q(y)e^{-\lambda x(y)}
= q(1)e^{-\lambda x(1)} + q(d)e^{-\lambda x(d)} + q(0)e^{-\lambda x(0)}
= q(1)e^{-\lambda x(1)} + (1-q(1))e^{-\lambda x_F}.$
Therefore, recalling \eqref{eq:dynamic:gq},
$g_q(x;\lambda)
= -\frac{1}{\lambda}\log\big(q(1)e^{-\lambda x(1)}+(1-q(1))e^{-\lambda x_F}\big).$
In particular, $g_q(x;\lambda)$ depends on $x$ only through $(x(1),x_F)$ and on $q$ only through $q(1)$.

Next, define the two-outcome (binary) certainty-equivalent functional
$\hat g_{\theta}(u,v;\lambda)
:=-\frac{1}{\lambda}\log\big(\theta e^{-\lambda u}+(1-\theta)e^{-\lambda v}\big),$ for $ \theta\in(0,1),\ (u,v)\in\R^2.$
Then, $g_q$ can be rewritten as
$g_q(x;\lambda)=\hat g_{q(1)}(x(1),x_F;\lambda)$. Let $\hat M(\mu,\lambda;u,v)$ denote the binary analogue of $M(\mu,\lambda;x)$ in
\eqref{eq:dynamic:Delta-C}, i.e.,
$\hat M(\mu,\lambda;u,v)
:=\sum_{q\in Q^{2}}\mu(q) \hat g_{q(1)}(u,v;\lambda)-\hat g_{q^1(1)}(u,v;\lambda)$, so $M(\mu,\lambda;x)=\hat M(\mu,\lambda;x(1),x_F).$
Taking suprema over $x$ yields
\begin{align}\label{eq:app:milestones:C-leq}
C(\mu,\lambda)=\sup_{x:Y\to\R}M(\mu,\lambda;x)
\le \sup_{(u,v)\in\R^2}\hat M(\mu,\lambda;u,v).
\end{align}

We will now show the reverse inequality.
Fix any $(u,v)\in\R^2$ and construct $x$ by setting $x(1)=u$ and $x(d)=x(0)=v$.
Then, $x_F=v$ by definition of $x_F$ (since $\eta e^{-\lambda v}+(1-\eta)e^{-\lambda v}=e^{-\lambda v}$), and hence
$M(\mu,\lambda;x)=\hat M(\mu,\lambda;u,v)$.
Thus,
$\sup_{(u,v)\in\R^2}\hat M(\mu,\lambda;u,v)
\le \sup_{x:Y\to\R}M(\mu,\lambda;x)=C(\mu,\lambda),$
which combined with \eqref{eq:app:milestones:C-leq} gives equality.
Thus, the three-outcome capacity coincides with the binary capacity.
\end{proof}

\begin{proof}[Proof of Observation \ref{prop:dynamic:milestones-diagnostic}]
Fix the post-success state $(\mu_{\mathfrak{s}},\lambda_{\mathfrak{s}})$ and suppose \eqref{eq:dynamic:milestones-condition} holds, i.e.,
$\sum_{q\in Q^{2}}\mu_{\mathfrak{s}}(q)\log\frac{q^1(d)}{q(d)}>\lambda_{\mathfrak{s}} k.$ For $T>0$, define the continuation utility vector $x^{T}:Y\to\R$ by $x^{T}(d)=0$ and $x^{T}(1)=x^{T}(0)=\frac{T}{\lambda_{\mathfrak{s}}}$. Then, $e^{-\lambda_{\mathfrak{s}} x^T(d)}=1$ and $e^{-\lambda_{\mathfrak{s}} x^T(1)}=e^{-\lambda_{\mathfrak{s}} x^T(0)}=e^{-T}$, so for any
$q\in\{q^1\}\cup Q^{2}$,
$\sum_{y\in Y}q(y)e^{-\lambda_{\mathfrak{s}} x^{T}(y)}
= q(d)\cdot 1+(1-q(d))e^{-T}$ and $g_q(x^{T};\lambda_{\mathfrak{s}})
=-\frac{1}{\lambda_{\mathfrak{s}}}\log\big(q(d)+(1-q(d))e^{-T}\big).$
Hence, recalling $M$ in \eqref{eq:dynamic:Delta-C},
\begin{align}
&M(\mu_{\mathfrak{s}},\lambda_{\mathfrak{s}};x^{T})
=\sum_{q\in Q^{2}}\mu_{\mathfrak{s}}(q)g_q(x^{T};\lambda_{\mathfrak{s}})-g_{q^1}(x^{T};\lambda_{\mathfrak{s}})\nonumber\\
&=\frac{1}{\lambda_{\mathfrak{s}}}\Bigg[
\log\big(q^1(d)+(1-q^1(d))e^{-T}\big)
-\sum_{q\in Q^{2}}\mu_{\mathfrak{s}}(q)\log\big(q(d)+(1-q(d))e^{-T}\big)
\Bigg].\label{eq:app:milestones:MT}
\end{align}

\smallskip
\noindent\textit{Step 1: bound approximation error.}
Fix any $a\in(0,1)$ and $T>0$. Write
$a+(1-a)e^{-T}=a(1+\frac{1-a}{a}e^{-T}).$
Therefore,
\begin{align}\label{eq:app:milestones:log-bound}
0\le \log\big(a+(1-a)e^{-T}\big)-\log a
=\log\Big(1+\frac{1-a}{a}e^{-T}\Big)
\le \frac{1-a}{a}e^{-T},
\end{align}
where the last inequality uses $\log(1+u)\le u$ for all $u\ge 0$.

\smallskip
\noindent\textit{Step 2: lower bound.}
Define the ``limit wedge'' $L:=\frac{1}{\lambda_{\mathfrak{s}}}\sum_{q\in Q^{2}}\mu_{\mathfrak{s}}(q)\log\big(\frac{q^1(d)}{q(d)}\big).$
By assumption, $L>k$. Using \eqref{eq:app:milestones:MT} and then \eqref{eq:app:milestones:log-bound} (applied to $a=q(d)$), we get
\begin{align*}
M(\mu_{\mathfrak{s}},\lambda_{\mathfrak{s}};x^{T})
&=\frac{1}{\lambda_{\mathfrak{s}}}\Bigg[
\log q^1(d)-\sum_{q\in Q^{2}}\mu_{\mathfrak{s}}(q)\log q(d)
+\Big(\log\big(q^1(d)+(1-q^1(d))e^{-T}\big)-\log q^1(d)\Big)\\
&\hspace{1.9in}
-\sum_{q\in Q^{2}}\mu_{\mathfrak{s}}(q)\Big(\log\big(q(d)+(1-q(d))e^{-T}\big)-\log q(d)\Big)
\Bigg]\\
&\ge L-\frac{1}{\lambda_{\mathfrak{s}}}\sum_{q\in Q^{2}}\mu_{\mathfrak{s}}(q) \frac{1-q(d)}{q(d)}e^{-T}.
\end{align*}
Let $Z:=\sum_{q\in Q^{2}}\mu_{\mathfrak{s}}(q)\frac{1-q(d)}{q(d)}\in(0,\infty).$
Then,
\begin{align}\label{eq:app:milestones:MT-lower}
M(\mu_{\mathfrak{s}},\lambda_{\mathfrak{s}};x^{T})\ge L-\frac{Z}{\lambda_{\mathfrak{s}}}e^{-T}.
\end{align}

\smallskip
\noindent\textit{Step 3: choose finite $T$.}
Since $L-k>0$, choose any $T$ such that
$e^{-T}\le \frac{\lambda_{\mathfrak{s}}(L-k)}{Z}$, i.e., $T\ge \log\frac{Z}{\lambda_{\mathfrak{s}}(L-k)}$.
Then, \eqref{eq:app:milestones:MT-lower} implies $M(\mu_{\mathfrak{s}},\lambda_{\mathfrak{s}};x^{T})\ge k$. Finally, by definition of capacity in \eqref{eq:dynamic:Delta-C}, $C(\mu_{\mathfrak{s}},\lambda_{\mathfrak{s}})=\sup_{x:Y\to\R}M(\mu_{\mathfrak{s}},\lambda_{\mathfrak{s}};x) \ge M(\mu_{\mathfrak{s}},\lambda_{\mathfrak{s}};x^{T})\ge k.$
Therefore, innovation is implementable after success by Proposition \ref{prop:dynamic:terminal-impl}.
\end{proof}

\begin{proof}[Proof of Observation \ref{prop:app:shirk:terminal-impl}]
Fix $(\mu,\lambda)$ and $x$. By \eqref{eq:app:shirk:V2-actions}, $V_2(\mu,\lambda;2,x)\ge V_2(\mu,\lambda;0,x)$
is equivalent to $\mathcal G(x;\mu,\lambda)-k \ge g_{q^0}(x;\lambda)$, i.e., $M^{0}(\mu,\lambda;x)\ge k$.
Similarly, $V_2(\mu,\lambda;2,x)\ge V_2(\mu,\lambda;1,x)$ is equivalent to
$\mathcal G(x;\mu,\lambda)-k \ge g_{q^1}(x;\lambda)-k_1$, i.e., $M^{1}(\mu,\lambda;x)\ge k-k_1$.
Thus, $a_2=2$ is optimal if and only if \eqref{eq:app:shirk:two-IC} holds, which is equivalent to
$\underline{M}(\mu,\lambda;x)\ge 0$ by definition \eqref{eq:app:shirk:Delta-underline}.
Applying this at the realized state $(\mu_2(\cdot|y_1),\lambda_2(y_1))$ and $x=x_2(y_1,\cdot)$ yields
\eqref{eq:app:shirk:terminal-criterion}. Existence of a contract that induces $a_2=2$ after $y_1$ is equivalent to existence of some finite vector $x:Y\to\R$ such that $\underline M(\mu_2(\cdot|y_1),\lambda_2(y_1);x)\ge 0$. Such a vector exists when either $\underline C(\mu_2(\cdot|y_1),\lambda_2(y_1))>0$, or with equality if the supremum is attained.
\end{proof}

\end{document}